\documentclass[12pt,a4paper]{article}

\usepackage[english]{babel}
\usepackage[T1]{fontenc}

\newcommand{\RN}[1]{%
  \textup{\uppercase\expandafter{\romannumeral#1}}%
}

\usepackage[a4paper,top=2cm,bottom=2cm,left=2cm,right=2cm,marginparwidth=2cm]{geometry}

\usepackage[colorinlistoftodos]{todonotes}
\usepackage[colorlinks=true, allcolors=blue]{hyperref}
\usepackage[most]{tcolorbox}
\usepackage[title]{appendix}

\usepackage{algorithm} 
\usepackage{algpseudocode} 
\usepackage{amsmath, amssymb, amsfonts} 
\usepackage{amsthm}
\allowdisplaybreaks
\usepackage{bbm}
\usepackage{centernot}
\usepackage{comment}
\usepackage{enumitem}
\usepackage{float}
\usepackage{graphicx}
\usepackage{natbib}
\usepackage{ragged2e}
\usepackage{xcolor}
\usepackage{booktabs}
\usepackage{mathtools}
\usepackage{orcidlink}

\newcommand{\probP}{\text{I\kern-0.15em P}}

\DeclareMathOperator*{\argmax}{argmax}

\numberwithin{equation}{section}

\theoremstyle{plain}
\newtheorem{theorem}{Theorem}[section]
\newtheorem{corollary}[theorem]{Corollary}
\newtheorem{lemma}[theorem]{Lemma}
\newtheorem{definition}[theorem]{Definition}
\newtheorem{assumption}[theorem]{Assumption}

\newtheorem{prop}[theorem]{Proposition}

\theoremstyle{remark}

\usepackage{setspace}

\title{Most Powerful Test with Exact Family-Wise Error Rate Control: Necessary Conditions and a Path to Fast Computing}
\author{%
  Prasanjit Dubey$^{1,\orcidlink{0000-0002-3667-5507},}
  $
  \thanks{Corresponding author. 755 Ferst Dr NW, Atlanta, GA 30332, USA. Email: \href{mailto:pdubey31@gatech.edu}{pdubey31@gatech.edu}}%
  \qquad
  Xiaoming Huo$^{1,\orcidlink{0000-0003-0101-1206}}$\\
  $^{1}$H.~Milton Stewart School of Industrial and Systems Engineering,\\
  Georgia Institute of Technology, Atlanta, GA 30332, USA 
}

\date{}

\begin{document}
\doublespacing
\maketitle




\begin{abstract}
Identifying the most powerful test in multiple hypothesis testing under strong family-wise error rate (FWER) control is a fundamental problem in statistical methodology. 
State-of-the-art approaches formulate this as a constrained optimisation problem, for which a dual problem with strong duality has been established in a general sense.
However, a constructive method for solving the dual problem is lacking, leaving a significant computational gap. 
This paper fills this gap by deriving novel, necessary optimality conditions for the dual optimisation. 
We show that these conditions motivate an efficient coordinate-wise algorithm for computing the optimal dual solution, which, in turn, provides the most powerful test for the primal problem. 
We prove the linear convergence of our algorithm, i.e., the computational complexity of our proposed algorithm is proportional to the logarithm of the reciprocal of the target error.
To the best of our knowledge, this is the first time such a fast and computationally efficient algorithm has been proposed for finding the most powerful test with family-wise error rate control.
The method's superior power is demonstrated through simulation studies, and its practical utility is shown by identifying new, significant findings in both clinical and financial data applications.
\end{abstract}

\noindent \textbf{Keywords:} Coordinate descent, FWER control, Multiple hypothesis testing, Optimal decision policy, Optimization











\section{Introduction}
Multiple-hypothesis testing (MHT) underpins much of modern empirical science by enabling the simultaneous assessment of many claims while guarding against spurious findings. 
The fundamental objective is to maximise power, the ability to detect true effects, subject to rigorous control of false rejections. 
Among competing error metrics, the family-wise error rate (FWER), defined as the probability of at least one false rejection, remains a gold standard when the cost of any Type~I error is high. 
Strong control of the FWER at a pre-specified level $\alpha$, uniformly over all configurations of null and non-null hypotheses, is the norm in settings where reliability and reproducibility are paramount.

Applications where FWER control is essential are widespread and diverse. 
In clinical research, for example, multiple endpoints or treatment contrasts are routinely evaluated, and inflating Type~I error can lead to unsafe or ineffective therapies being advanced (\cite{Pocock87, Bretz09, Vickerstaff19}). 
In the social sciences and economics, investigators often test families of related hypotheses about policy impacts, and erroneous rejections can misguide policy formulation (\cite{Romano05, list16, Viviano25}). 
Large-scale online experimentation and product A/B testing face similar risks, where spurious “wins” propagate costly product decisions; principled error control remains central even amid sequential and adaptive experimentation (\cite{Johari21}). 
Related concerns arise in contemporary machine-learning pipelines, where distribution-free notions of risk control for black-box algorithms are increasingly required (e.g., \cite{bates21}). 
These motivating examples reinforce the need for procedures that are simultaneously powerful and equipped with strong, finite-sample error guarantees.

Classical and modern literatures have studied optimality under FWER control from several angles, including foundational contributions on multiple comparisons and optimal tests (\cite{Spjtvoll1972OnTO, westfall98, Lehmann1}) and more recent developments that seek most-powerful procedures under global constraints (\cite{Dobriban, Rosenblum}). 
A major conceptual step was the formulation of optimal multiple testing as an infinite-dimensional binary program whose objective captures power and whose constraints enforce family-wise control. 
For simple hypotheses, \cite{RHPA22} demonstrated that this program admits a dual problem with strong duality: the optimal value of the primal (most-powerful test subject to error constraints) equals that of the dual.
When a dual vector $\mu^\ast$ of Lagrange multipliers exists and satisfies the complementary-slackness conditions (Proposition~2 of \cite{RHPA22}), the associated decision rule $\vec D^{\mu^\ast}$ is optimal for the primal. 
Under a mild non-redundancy assumption on the family of joint densities (Assumption~3 of \cite{RHPA22}), such a $\mu^\ast$ is guaranteed to exist.

Despite this elegant duality theory, a practical gap persists. 
The existing results establish existence but stop short of providing constructive characterisations or algorithms for computing the optimal dual multipliers. 
In practice, one is left with a high-dimensional, potentially ill-conditioned search over dual parameters with repeated evaluation of constraint integrals, an approach that can be computationally prohibitive and numerically fragile. 
This gap limits the day-to-day deployability of otherwise optimal procedures, particularly when the number of hypotheses is small to moderate and exact control is required.

The present paper addresses this computational bottleneck by developing a principled route to compute optimal dual multipliers for exact FWER control. 
We focus on the case $K=3$, which is of substantive interest in applications such as subgroup analyses in clinical trials, where only a handful of clinically meaningful groups are targeted (\cite{RHPA22, Rosenblum}). 
Within this setting, we derive necessary optimality conditions for the dual variables (Theorem~\ref{thm:mu_optimality_combined} in this paper) that translate the abstract existence arguments into concrete analytic relationships. 
These conditions exploit the symmetry and arrangement-increasing structure of the problem together with the linear representations of power and error in the $p$-value domain.

We introduce a coordinate-wise algorithm (Algorithm~\ref{alg:compute_optimal_mu_K3_main}) that solves the dual program efficiently and reliably by building on these conditions. 
The method alternates over dual coordinates, bracketing and solving the one-dimensional optimality equations while keeping the remaining multipliers fixed, and is coupled with numerically stable evaluation of the requisite integrals on the sorted $p$-value simplex. 
This yields a provably fast and practical procedure that returns the optimal Lagrange multipliers and hence the most powerful FWER-controlling policy for $K=3$.

Our results make the optimal multiple-testing framework of \cite{RHPA22} practically viable in small-$K$ regimes that are common in applied work, without sacrificing exact, strong error control, by rendering the dual approach operational. 
Beyond providing a computational contribution, the explicit optimality conditions clarify the structure of optimal policies and illuminate how power is traded against error constraints through the dual variables. 
While our theory and algorithms are developed for $K=3$, the analysis suggests promising pathways to extend these ideas to larger families, which we discuss alongside empirical evaluations and comparisons with established procedures in subsequent sections.

\section{Problem Formulation}
\label{sec:formulation}
This section establishes the mathematical framework for our analysis of multiple hypothesis testing (MHT). 
We begin in Section \ref{sec:mht} by introducing the fundamental notation for the testing problem, including the hypotheses, data, and decision rules. 
Following this, Section \ref{sec:powererrorfn} provides formal definitions for the core metrics used to evaluate statistical power and control error rates. 
Section \ref{sec:most-powerulf-via-optimize} then frames the search for an optimal testing procedure as a constrained optimization problem. 
To simplify this problem, we introduce two generic structural assumptions and motivate our focus on likelihood ratio-ordered decision policies. 
The problem is subsequently reformulated in the p-value domain in Section \ref{sec:p-values}. 
We also present a general linear representation for the power and error functions, and then derive the corresponding Lagrangian dual problem in Section \ref{sec:lagdualprob}, which forms the basis for our main theoretical results.

\subsection{Multiple Hypotheses Testing}
\label{sec:mht}
This subsection introduces the fundamental setting and notation for a general multiple hypothesis testing (MHT) problem. 
We define the core components, including the hypotheses, the observed data, the decision rules, and the underlying probabilistic structure that governs the relationship between them.

We formally define the MHT problem as the concurrent assessment of $K$ distinct hypotheses,
\[
\{H_1, H_2, \ldots, H_K\}.
\]
Within this collection, an unknown subset $\mathcal{N} \subseteq [K] = \{1, 2, \ldots, K\}$ represents the indices of the hypotheses that are true nulls. 
Each hypothesis $H_k$ is evaluated based on an observed data vector $\vec{X}_k \in \mathcal{X}_k$, and the complete dataset is $\vec{X} = (\vec{X}_1, \ldots, \vec{X}_K)$, residing in the product space $\mathcal{X} = \mathcal{X}_1 \times \cdots \times \mathcal{X}_K$. 
A testing policy maps the observed data to binary decisions $\vec{D}(\vec{X}) = (D_1(\vec{X}), \ldots, D_K(\vec{X})) \in \{0,1\}^K$, where $D_k(\vec{X}) = 1$ signifies rejection of $H_k$ and $D_k(\vec{X}) = 0$ signifies non-rejection.

The underlying true state of nature is captured by a fixed, unknown binary vector $\vec{h} \in \{0,1\}^K$, where $h_k = 0$ indicates that $H_k$ is a true null and $h_k = 1$ indicates the alternative is true. 
We denote the joint likelihood under configuration $\vec{h}$ by $\mathcal{L}_{\vec{h}}(\vec{X})$. 
We also define the likelihood ratios $\Lambda_k(\vec{X}_k) = f_{k,1}(\vec{X}_k)/f_{k,0}(\vec{X}_k)$ for $k=1,\ldots,K$, where $f_{k,0}$ and $f_{k,1}$ are the null and alternative densities for $\vec{X}_k$.
We frequently consider the structured configurations indexed by $l \in [K]$ in which the first $l$ hypotheses are false nulls and the remaining $K-l$ are true nulls:
\[
\vec{h}_l = (\underbrace{1, 1, \ldots, 1}_l, \underbrace{0, \ldots, 0}_{K-l})^t.
\]

Under the common assumption of independence between the individual tests, the joint distribution of $\vec{X}$ given $\vec{h}$ factorizes as:
\begin{equation}
\label{eq:true_joint_dist}
f_{\vec{h}}(\vec{X}) = \prod_{k=1}^K f_{k,h_k}(\vec{X}_k),
\end{equation}
where $f_{k,h_k}(\vec{X}_k) = f_{k,0}(\vec{X}_k)$ if $h_k = 0$ and $f_{k,h_k}(\vec{X}_k) = f_{k,1}(\vec{X}_k)$ if $h_k = 1$.
To formalize symmetry, let $S_K$ denote the set of all permutations of $[K]$. 
For $\sigma \in S_K$, we write $\sigma(\vec{X})$ (resp. $\sigma(\vec{h})$) for the vector obtained by permuting the coordinates of $\vec{X}$ (resp. $\vec{h}$) according to $\sigma$.
We collect additional permutation notation (including transpositions $\sigma_{ij}$) in Appendix~\ref{app:perm_notation}.

\subsection{Power and Error Functions in Multiple Hypotheses Testing}
\label{sec:powererrorfn}
This subsection formally defines the standard metrics used to evaluate the performance of a multiple testing procedure. 
We begin by categorizing the possible outcomes of a testing procedure for $K$ hypotheses, as summarized in Table \ref{tab:outcome}, and then detail the specific functions for quantifying both statistical power and error rates.
\begin{table}[htbp]
    \centering
    \begin{tabular}{c|c|c|c}
        \toprule
         & $H_i$'s Accepted & $H_i$'s Rejected & Total \\\midrule
        $H_i$'s True  & $U$ & $V$ & $K_0$ \\
        $H_i$'s False & $T$ & $S$ & $K - K_0$ \\\midrule
                     & $W$ & $R$ & $K$ \\
        \bottomrule
    \end{tabular}
    \caption{Summary of key quantities in multiple hypotheses testing.}
    \label{tab:outcome}
\end{table}

We consider two primary formulations for the power of a test, following \cite{RHPA22}.
The first power metric, minimal power ($\Pi_{\text{any}}$), measures the probability of making at least one discovery under the configuration $\vec{h}_K$ where all $K$ null hypotheses are false:
\begin{equation}
\label{eq:piany_final}
\Pi_{\text{any}}(\vec{D}) = \mathbb{P}_{\vec{h}_K}(R > 0).
\end{equation}
The second metric, average power ($\Pi_l$), provides a measure of average power under a configuration $\vec{h}_l$ where exactly $l$ hypotheses are false nulls:
\begin{equation}
\label{eq:pil_final}
\Pi_l(\vec{D}) = \frac{1}{l} \mathbb{E}_{\vec{h}_l}\left[\sum_{k=1}^l D_k(\vec{X})\right], \quad 1 \leq l \leq K.
\end{equation}

To constrain erroneous rejections, we define two standard error rates. 
The family-wise error rate (FWER), denoted $\mathrm{FWER}_l$, is:
\begin{equation}
\label{eq:fwer_final}
\mathrm{FWER}_l(\vec{D}) = \mathbb{P}_{\vec{h}_l}(V > 0), \quad 0 \leq l < K.
\end{equation}
The false discovery rate (FDR), denoted $\mathrm{FDR}_l$, is:
\begin{equation}
\label{eq:fdr_final}
\mathrm{FDR}_l(\vec{D}) = \mathbb{E}_{\vec{h}_l}\left[\frac{V}{R} ; R > 0\right], \quad  0 \leq l < K.
\end{equation}

It is well-established that FWER provides stricter control than FDR. 
This relationship is formalized in Appendix~\ref{app:proof_fwer_implies_fdr}.
Therefore any procedure that controls $\mathrm{FWER}_l(\vec{D}) \le \alpha$ automatically controls $\mathrm{FDR}_l(\vec{D}) \le \alpha$ at the same level.

\subsection{Searching for the Most Powerful Test via Optimization}
\label{sec:most-powerulf-via-optimize}
This subsection synthesizes the preceding definitions to cast the search for an optimal multiple testing procedure as a constrained optimization problem. 

The central objective can be framed as maximizing statistical power, $\Pi(\vec{D})$, subject to the constraint that a chosen error function, $\mathrm{Err}_{\vec{h}_l}(\vec{D})$, does not exceed a prespecified level $\alpha \in (0, 1)$:
\begin{equation}
\label{eq:mtpintitial1}
\max_{\vec{D} \in \{0,1\}^K} \ \Pi(\vec{D}), \quad \text{subject to} \quad \mathrm{Err}_{\vec{h}_l}(\vec{D}) \leq \alpha, \quad  0 \leq l < K.
\end{equation}
In this formulation, $\Pi(\vec{D})$ can be $\Pi_{\text{any}}(\vec{D})$ from \eqref{eq:piany_final} or $\Pi_l(\vec{D})$ from \eqref{eq:pil_final}, and $\mathrm{Err}_{\vec{h}_l}(\vec{D})$ can be $\mathrm{FDR}_l(\vec{D})$ from \eqref{eq:fdr_final} or $\mathrm{FWER}_l(\vec{D})$ from \eqref{eq:fwer_final}.

Our main theoretical results focus on the specific instance where we maximize $\Pi_{l=3}(\vec{D})$ under FWER constraints with $K=3$ hypotheses:
\begin{equation}\label{eq:mtpintitial2}
\begin{aligned}
\max_{\vec{D} \in \{0,1\}^3} \ &  \Pi_{l=3}(\vec{D}) \\
\text { s.t. } & \mathrm{FWER}_{l}(\vec{D})\leq \alpha, \ 0 \leq l\leq 2. \\
\end{aligned}
\end{equation} 
While the framework developed in the remainder of this section applies to the general problem \eqref{eq:mtpintitial1}, our primary contribution, presented in Section \ref{sec:main_result}, is a necessary condition for the optimal solution to the specific problem \eqref{eq:mtpintitial2}.

\paragraph{Two Conditions and Consequences}
We now introduce two fundamental conditions, following \cite{RHPA22}, that impose structure on the joint distribution of the data and are crucial for simplifying the optimization problem defined in \eqref{eq:mtpintitial1}.
The first condition, $\vec{h}$-exchangeability, requires the joint distribution to be symmetric under permutation, ensuring that hypothesis identities convey no information beyond their truth status.
\begin{assumption}[$\vec{h}$-Exchangeability]
\label{as:assumption1}
\label{assumption:a1}
The $K$ tests are $\vec{h}$-exchangeable if, for all $\sigma \in S_{K}$, we have 
$$
\mathcal{L}_{\vec{h}}(\vec{X})=\mathcal{L}_{\sigma(\vec{h})}[\sigma(\vec{X})],
$$
where $\mathcal{L}_{\vec{h}}(\vec{X})$ denotes the joint likelihood for data $\vec{X}$ and configuration $\vec{h}$.
\end{assumption}
This exchangeability motivates restricting our attention to symmetric decision functions, which are invariant to the ordering of the hypotheses.
\begin{definition}[Symmetric Decision Function]
\label{def:decision_rule}
A decision function $\vec{D}: \mathcal{X} \rightarrow \{0,1\}^{K}$ is symmetric if, for all $\vec{X} \in \mathcal{X}$ and $\sigma \in S_{K}$,
$$
\sigma(\vec{D}(\vec{X}))=\vec{D}(\sigma(\vec{X})).
$$
\end{definition}

The second condition, arrangement-increasing, imposes a monotonicity structure requiring the joint likelihood to increase when stronger evidence aligns with true alternatives.
\begin{assumption}[Arrangement-Increasing]
\label{as:assumption2}
The likelihood function $\mathcal{L}_{\vec{h}}(\cdot)$ is arrangement increasing if, for any pair of indices $i \neq j$ and data $\vec{X}$ such that
\begin{equation}\label{eq:assumption2_eq1}
    (\Lambda_i(\vec{X}_{i})-\Lambda_j(\vec{X}_{j}))(h_{i}-h_{j}) \leq 0
\end{equation}
holds, then upon exchanging entries $i$ and $j$ in $\vec{X}$ to create $\sigma_{i j}(\vec{X})$, we have
\begin{equation}
\mathcal{L}_{\vec{h}}(\vec{X}) \leq \mathcal{L}_{\vec{h}}(\sigma_{i j}(\vec{X})).
\end{equation}
\end{assumption}
The primary consequence of the arrangement-increasing assumption is that it justifies restricting our search to Likelihood Ratio-Ordered Decision Policies, which we discuss below.

\paragraph{Likelihood Ratio-Ordered Decision Policies}
We restrict attention to decision rules that are \emph{likelihood ratio (LR)-ordered}, meaning that hypotheses with stronger evidence against the null (larger likelihood ratios) are rejected at least as readily as those with weaker evidence. 
\begin{definition}[LR-ordered Decision Function]\label{def:lrordering}
A decision function $\vec{D}: \mathcal{X} \rightarrow\{0,1\}^{K}$ is LR-ordered if
$$
\Lambda_i(\vec{X}_{i}) \geq \Lambda_j(\vec{X}_{j}) \Rightarrow D_{i}(\vec{X}) \geq D_{j}(\vec{X}).
$$
\end{definition}
This restriction is without loss of optimality under the structural conditions introduced earlier. 
In particular, under Assumptions \ref{as:assumption1} (exchangeability) and \ref{as:assumption2} (arrangement-increasing), \citet[Theorem~1 and Corollary~1]{RHPA22} show that for the power criteria $\Pi_{\text{any}}$ or $\Pi_l$ and for error criteria $\mathrm{FDR}$ or $\mathrm{FWER}$, any symmetric decision rule can be replaced by an LR-ordered symmetric rule that weakly improves power while not increasing the error. 
Consequently, in solving \eqref{eq:mtpintitial1} (and its specialized instances), it suffices to search over the class of symmetric LR-ordered policies.
The full details are provided in Appendix~\ref{app:lrorder_full}.

\subsection{Reformulating the MHT Using P-values}
\label{sec:p-values}
In this subsection, we reformulate the multiple hypotheses testing problem in the domain of p-values, allowing decision rules to be expressed as functions of an ordered p-value vector. Additionally, we demonstrate a key structural property of the power and error functions defined earlier.

We now transition from likelihood ratios to their corresponding p-values, denoted $u_i$. 
For a realized likelihood ratio $\Lambda_{i}$ for the $i$-th hypothesis, the p-value $u_i$ is:
$$
u_{i}
=\int I\left\{\vec{X}_{i}: \frac{f_{i,1}(\vec{X}_{i})}{f_{i,0}(\vec{X}_{i})} \geq \Lambda_{i}\right\} f_{i,0}(\vec{X}_{i}) d \vec{X}_{i}, 
$$
where $f_{i,0}(\cdot)$ and $f_{i,1}(\cdot)$ are the null and alternative densities, respectively, as defined in \eqref{eq:true_joint_dist}.

Under $H_{0k}$, $u_k \sim U(0,1)$, and under $H_{Ak}$, we denote the alternative density by $g_k(\cdot)$:
\begin{equation}\label{eq:mtp1st1}
    H_{0k}\colon u_k \sim U(0, 1) \ \text{vs} \ H_{Ak}\colon u_k \sim g_k, \ \text{for} \  k=1,2,\cdots,K. 
\end{equation} 
To maintain symmetry, we assume $g_k(\cdot)=g(\cdot)$ for all $k$. 
Given symmetry and LR-ordering (small p-values correspond to large LRs), it suffices to define $\vec{D}$ on the ordered subset
$$
Q =\left\{u: 0 \leq u_{1} \leq u_{2} \leq \cdots \leq u_{K} \leq 1\right\},
$$
and extend by symmetry to $[0,1]^K$.
Let $\lambda(u)$ denote the likelihood ratio threshold corresponding to a p-value $u$, defined implicitly by:
\begin{equation}\label{eq:pvalue}
    u=\int \mathbb{I}\left\{\frac{f_1(x)}{f_0(x)}\geq \lambda(u)\right\}f_0(x) dx, 
\end{equation}
where, for notational simplicity, we suppress the hypothesis index under the symmetric setup.
The following lemma connects $\lambda(\cdot)$ to the alternative p-value density $g(\cdot)$.
\begin{lemma}
\label{lem:gdot}
Suppose $g(\cdot)$ is the identical alternative density function defined in \eqref{eq:mtp1st1}, and $\lambda(u)$ is the likelihood ratio threshold defined in \eqref{eq:pvalue} for $0 \le u \le 1$. 
If $g(\cdot)$ is continuous on $(0,1)$ and differentiable almost everywhere, then we have, for $0 \le u \le 1$, 
$$
g(u)=\lambda(u).
$$
\end{lemma}
\begin{proof}
See the proof in Appendix \ref{app:proof-g-monotone}.
\end{proof}
Since $\lambda(\cdot)$ is non-increasing in $u$, Lemma~\ref{lem:gdot} implies that $g(\cdot)$ is also non-increasing on $[0,1]$.

\paragraph{Linear Representation of the Power and Error Functions}
It can be further shown that the power and error functions admit linear representations as functionals of $\vec{D}$ on $Q$, following \cite{RHPA22}. 
\begin{theorem}(\cite[Equations (4) and (7)]{RHPA22}) 
\label{thm:powererrorfn} 
Taking into account $\vec{h}$-exchangeability and symmetry, the power function $\Pi_l(\vec{D})$, defined in \eqref{eq:pil_final}, can be written as:
\begin{equation}\label{eq:pillinear1}
\begin{aligned}
\Pi_l(\vec{D}) = (l-1)!(K-l)! \int_{Q} \sum_{i \in {\binom{K}{l}}} f_{i}(\vec{u}) \sum_{k \in i} D_{k}(\vec{u}) d\vec{u}, \ 1 \leq l \leq K.
\end{aligned}
\end{equation}
Additionally, the family-wise error rate admits the linear representation: 
\begin{equation}\label{eq:fwer_linear}
    FWER_l(\vec{D}) = l!(K - l)! \int_Q \sum_{k=1}^K D_k(\vec{u}) \sum_{\substack{i \in \binom{K}{l} \\ \Bar{i}_{\min}=k}} f_i(\vec{u}) \, d\vec{u}, \quad  0 \leq l < K.
\end{equation}  
\end{theorem}
Here $f_i(\vec{u})$ denotes the joint density of $\vec{u}$ under the configuration in which the hypotheses indexed by $i$ are false nulls (and the complement are true nulls); the constraint $\Bar{i}_{\min}=k$ identifies the relevant term in the representation of $\mathrm{FWER}_l$ (see Appendix~\ref{app:linear_rep_notation} for exhaustive notational details).
To unify the problem into a generic optimization framework, we express both power and error functions in a standard linear form:
\begin{equation}
\label{eq:power_general}
    \Pi(\vec{D}) = \int_{Q}\left(\sum_{k=1}^{K} a_{k}(\vec{u}) D_{k}(\vec{u})\right) d \vec{u}
\end{equation}
and
\begin{equation}\label{eq:error_general}
    Err_{\vec{h}_{l}}(\vec{D}) 
    = \int_{Q}\left(\sum_{k=1}^{K} b_{l, k}(\vec{u}) D_{k}(\vec{u})\right) d \vec{u} \leq \alpha, \quad  0 \leq l < K.
\end{equation} 
The coefficient functions $a_k(\vec{u})$ and $b_{l,k}(\vec{u})$ are fixed, non-negative functions defined over $Q$. 
With these linear representations, the general problem \eqref{eq:mtpintitial1} can be rewritten in the canonical form:
\begin{equation} \label{eq:mtpopt1}
\begin{aligned}
\max_{\vec{D}} \ & \int_{Q}\left(\sum_{i=1}^{K} a_{i}(\vec{u}) D_{i}(\vec{u})\right) d \vec{u} \\ 
\text{ s.t. } \ & \int_{Q}\left(\sum_{i=1}^{K} b_{l, i}(\vec{u}) D_{i}(\vec{u})\right) d \vec{u} \leq \alpha, \quad  0 \leq l < K,
\end{aligned}
\end{equation} 
where $\vec{D}$ is restricted to be LR-ordered and symmetric. We denote the optimal value of this primal problem by $f^\ast$.

\subsection{Lagrangian Dual of the Central Optimization Problem} 
\label{sec:lagdualprob}
This subsection constructs the Lagrangian dual of \eqref{eq:mtpopt1}, following \cite{RHPA22}, and states strong duality conditions that justify solving the primal via the dual.

\begin{lemma}[Dual Problem, \cite{RHPA22}]\label{lem:dual_problem}
Let the Lagrangian associated with the primal problem~\eqref{eq:mtpopt1} be
\begin{equation}\label{eq:lagrangian1}
  L(\vec D,\mu)=
    \sum_{l=0}^{K-1}\mu_{l}\alpha
    +\int_{Q}\sum_{i=1}^{K}D_i(\vec u)
      R_i(\mu,\vec u)d\vec u,
\end{equation}
where the dual vector is
\(
  \mu=(\mu_{0},\dots,\mu_{K-1})^{\top}\!\ge 0
\)
and
\begin{equation}\label{eq:rimu1}
  R_{i}(\mu,\vec u)
  =
  a_i(\vec u)-\sum_{l=0}^{K-1}\mu_lb_{l,i}(\vec u),\qquad
  1\le i\le K.
\end{equation}
For every fixed \(\mu\ge 0\), define
\begin{equation}\label{eq:l*}
  l^{\ast}(\mu,\vec u)
  =
  \argmax_{1\le l\le K}
      \Bigl\{0,\sum_{i=1}^{l}R_i(\mu,\vec u)\Bigr\},
\end{equation}
and the corresponding (LR-ordered) policy
\begin{equation}\label{eq:optdecision1}
  D^{\mu}_{i}(\vec u)
  =
  \mathbf 1\!\bigl\{1\le i\le l^{\ast}(\mu,\vec u)\bigr\},
  \qquad 1\le i\le K.
\end{equation}
Then
\begin{equation}
\label{eq:opt_decision1}
  \max_{\vec D \in \mathcal{D}} L(\vec D,\mu)
  =
  L\bigl(\vec D^{\mu},\mu\bigr),
\end{equation}
where $\mathcal{D}$ is the set of symmetric, LR-ordered decision functions. Consequently, the dual program reduces to:
\begin{equation}\label{eq:ld_mu}
\min_{\mu\geq 0} \max_{\vec{D} \in \mathcal{D}} L(\vec{D}, \mu) \stackrel{\mbox{\eqref{eq:opt_decision1}}}{=} \min_{\mu\geq 0}  L(\vec{D}^{\mu}, \mu).
\end{equation}
\end{lemma}
\begin{proof}
See the proof in Appendix \ref{app:dual_problem}.
\end{proof}

We denote the optimal value of the dual problem by $d^\ast$, and the task reduces to finding $\mu^\ast \ge 0$ such that
\[
\mu^*=\underset{\mu \geq 0}{\operatorname{argmin}} \quad L(\vec{D}^{\mu}, \mu).
\]
The connection between $\mu^\ast$ and the optimal primal policy $\vec{D}^{\mu^\ast}$ is formalized by the KKT complementary slackness conditions ensuring strong duality:
\begin{prop}
\label{prop:optmu}
(\cite[Proposition 1]{RHPA22}) If there exists a $\mu^{*} \ge \vec{0}$ such that the policy $\vec{D}^{\mu^{*}}$ defined in \eqref{eq:optdecision1} satisfies for all $l=0, \ldots, K-1:$
\begin{equation*}
\begin{gathered}
\left(\mu_{l}^{*} > 0 \text { and } \int_{Q}\left(\sum_{i=1}^{K} b_{l, i}(\vec{u}) D_{i}^{\mu^{*}}(\vec{u})\right) d \vec{u}=\alpha\right) \text { or } \\
\left(\mu_{l}^{*}=0 \text { and } \int_{Q}\left(\sum_{i=1}^{K} b_{l, i}(\vec{u}) D_{i}^{\mu^{*}}(\vec{u})\right) d \vec{u} \leq \alpha\right),
\end{gathered}
\end{equation*}
then $\vec{D}^{\mu^{*}}$ is an optimal policy for Problem \eqref{eq:mtpopt1}.
\end{prop}
While an optimal $\mu^\ast$ satisfying Proposition~\ref{prop:optmu} need not exist in general, \cite{RHPA22} guarantees its existence under the following assumption on the joint densities.
\begin{assumption}(\cite[Assumption 3]{RHPA22})
\label{assumption:optmu}
The set of density functions $\left\{f_{\vec{h}}: \vec{h} \in\{0,1\}^{K}\right\}$ is non-redundant:
$$
\sum_{\vec{h}} \gamma_{\vec{h}} f_{\vec{h}}(\vec{u}) \neq 0 \text{ almost everywhere for any fixed vector } \gamma_{\vec{h}} \in \mathbb{R} \text{ with } \sum_{\vec{h}}\left|\gamma_{\vec{h}}\right|>0.
$$
\end{assumption}
While \cite{RHPA22} establishes this theoretical foundation, it proposes only a generic search method for finding $\mu^\ast$. 
This method involves searching the parameter space and numerically evaluating the constraint integrals $\int_{Q}(\sum_{i=1}^{K} b_{l, i}(\vec{u}) D_{i}^{\mu}(\vec{u})) d \vec{u}$ to check the optimality conditions of Proposition \ref{prop:optmu}. 
This leaves a significant gap, as no explicit form or constructive algorithm for obtaining $\mu^\ast$ is provided. 
A primary contribution of our paper, detailed in Section \ref{sec:main_result}, is to fill this gap by providing an explicit characterization and a practical algorithm for computing the optimal Lagrange multipliers, thus making the dual solution method computationally viable.
```

\section{Main Result}
\label{sec:main_result}

This section presents the central contributions of our work, focusing on a specialized yet informative multiple hypothesis testing problem. 
We begin in Section \ref{sec:specialMHT} by formally defining this specific problem, which involves $K=3$ hypotheses with FWER control. 
In Section \ref{sec:optimal-policy}, we derive the explicit analytic form of the optimal decision policy as a function of the dual Lagrangian multipliers. 
We then derive the explicit expression for the dual objective function itself in Section \ref{sec:explorelagrangian}. 
Following this, Section \ref{sec:condition-optimal} establishes the necessary optimality conditions for the dual minimization problem. 
A crucial property, the coordinatewise monotonicity of the target functions, is proven in Section \ref{sec:monotonicty-in-Fs}, with further monotonicity results detailed in Section \ref{sec:more-monotone}. 
Finally, based on this complete theoretical framework, we propose a computationally efficient coordinate-update algorithm in Section \ref{sec:computing01} that numerically solves for the optimal testing procedure.

\subsection{A Special MHT Problem}
\label{sec:specialMHT}

This subsection defines the specific multiple hypothesis testing (MHT) problem that serves as the focus of our theoretical analysis. 
We depart from the general formulation in \eqref{eq:mtpopt1} to identify a computationally tractable case, for which we derive a complete solution.

The specific MHT problem we investigate is a specialized instance of the p-value formulation from \eqref{eq:mtp1st1} where we consider $K=3$ hypotheses. 
This problem is expressed as:
\begin{equation}
\label{eq:mtp}
H_{0k}\colon u_k \sim U(0, 1) \ \text{vs} \ H_{Ak}\colon u_k \sim g, \ \text{for} \  k=1,2,3.
\end{equation}
We assume, as in Section \ref{sec:p-values}, that the alternative density $g(\cdot)$ is identical for all tests and, per Lemma \ref{lem:gdot}, is a nonincreasing monotonic function. 
While our theoretical results are contingent on the $K=3$ setting, this case provides essential insights, and its extension to more general $K$ remains a key direction for future research.

We further specify the problem by selecting the average power $\Pi_{l=3}$ as the objective function to be maximized and the family-wise error rate $FWER_l$ as the error function to be controlled, both of which were formally defined in Section \ref{sec:powererrorfn}. 
This specification leads to the optimization problem previously stated in \eqref{eq:mtpintitial2}, which we restate here for clarity:
\begin{equation} 
\label{eq:MHTK3}
\begin{aligned}
\max_{\vec{D} \in \{0,1\}^3} \ &  \Pi_{3}(\vec{D}) \\
\text { s.t. } & FWER_{l}(\vec{D})\leq \alpha, \ 0 \leq l\leq 2. \\
\end{aligned}
\end{equation} 

By applying the general linear representations from Theorem \ref{thm:powererrorfn}, specifically \eqref{eq:pillinear1} for power and \eqref{eq:fwer_linear} for FWER, and assuming independence between tests, we can derive the explicit optimization problem for our $K=3$ case. 
The following lemma provides this concrete formulation, which forms the basis for all subsequent analysis.

\begin{lemma}
\label{lemma:optimization_problem_k_3}
Consider a multiple hypotheses testing problem with $K = 3$ hypotheses. Let $\vec{D} = (D_1, D_2, D_3): Q \rightarrow [0,1]^3$ denote a decision policy defined on the ordered p-value domain $Q = \{ \vec{u} \in [0,1]^3 : 0 \le u_1 \le u_2 \le u_3 \le 1 \}$. 
Suppose the objective is to maximize the average power $\Pi_3(\vec{D})$ under family-wise error rate (FWER) constraints at a threshold $\alpha \in (0,1)$. 
Then, the problem is equivalent to the following optimization program: 
\begin{equation}
\label{eq:objconst}
\begin{aligned}
\max _{\vec{D}: Q \rightarrow[0,1]^{3}} & 2 \int_{Q}\left(D_{1}(\vec{u})+D_{2}(\vec{u})+D_{3}(\vec{u})\right) g(u_1) g(u_2) g(u_3) d \vec{u} \\
\text { s.t. } & FWER_{0}(\vec{D})=6 \int_{Q} D_{1}(\vec{u}) d \vec{u} \leq \alpha \text {, } \\
& FWER_{1}(\vec{D})=2 \int_{Q}\left[D_{1}(\vec{u})(g(u_2)+g(u_3))+D_{2}(\vec{u}) g(u_1) \right] d \vec{u} \leq \alpha \text {, } \\
& FWER_{2}(\vec{D})=2 \int_{Q}\left[D_{1}(\vec{u}) g(u_2) g(u_3)+ D_{2}(\vec{u}) g(u_1) g(u_3) + D_{3}(\vec{u}) g(u_1) g(u_2)\right] d \vec{u} \leq \alpha, \\
& 0 \leq D_{3}(\vec{u}) \leq D_{2}(\vec{u}) \leq D_{1}(\vec{u}) \leq 1, \ \forall \vec{u} \in Q.
\end{aligned}
\end{equation} 
This formulation is equivalent to the general problem \eqref{eq:mtpopt1} under the symmetry and exchangeability Assumptions \ref{as:assumption1} and \ref{as:assumption2}.
\end{lemma}
\begin{proof}
See the proof in Appendix \ref{app:proof_optimization_problem_k_3}.
\end{proof}

\subsection{Optimal Decision Policy}
\label{sec:optimal-policy}
This subsection derives the explicit structure of the optimal decision policy $\vec{D}^\mu$ that maximizes the Lagrangian $L(\vec{D}, \mu)$ for any fixed vector of dual multipliers $\mu = (\mu_0, \mu_1, \mu_2)$.

To determine the optimal policy $\vec{D}^\mu = (D_1^\mu, D_2^\mu, D_3^\mu)$, we follow the general structure outlined in Lemma \ref{lem:dual_problem}, specifically equations \eqref{eq:l*} and \eqref{eq:optdecision1}. 
For any point $\vec{u} \in Q$, the optimal decisions $D_i^\mu(\vec{u})$ are determined by checking whether the cumulative sums of the functions $R_i(\mu, \vec{u})$ are positive. 
This construction, which inherently respects the LR-ordering constraint $D_3 \le D_2 \le D_1$, is formalized in the following lemma.

\begin{lemma}
\label{lem:opt_decision_mu}
For the MHT problem with $K=3$ defined in \eqref{eq:objconst}, the optimal decision policy $$\vec{D}^\mu(\vec{u}) = (D_1^\mu(\vec{u}), D_2^\mu(\vec{u}), D_3^\mu(\vec{u}))$$ that maximizes the Lagrangian $L(\vec{D}, \mu)$ for a fixed $\mu$, as given in \eqref{eq:opt_decision1}, has the form:
\begin{equation}\label{eq:dimus2}
\begin{aligned}
D_{1}^{\mu}(\vec{u})= & \alpha_1^{\mu}(\vec{u}), \\
D_{2}^{\mu}(\vec{u})= & \alpha_{1}^{\mu}(\vec{u}) \alpha_2^{\mu}(\vec{u}), \\
D_{3}^{\mu}(\vec{u})= & \alpha_{1}^{\mu}(\vec{u}) \alpha_{2}^{\mu}(\vec{u}) \alpha_3^{\mu}(\vec{u}).
\end{aligned}
\end{equation}
where $\alpha_i^\mu: Q \to \{0, 1\}$ are indicator functions defined as:
\begin{equation}\label{eq:alpha_mu_u}
\begin{aligned}
\alpha_1^{\mu}(\vec{u}) &= \mathbbm{1}\left\{R_{1}(\mu,\vec{u}) > 0 \cup R_{1}(\mu,\vec{u})+R_{2}(\mu,\vec{u}) > 0 \cup R_{1}(\mu,\vec{u})+R_{2}(\mu,\vec{u})+R_{3}(\mu,\vec{u}) > 0 \right\}, \\
\alpha_2^{\mu}(\vec{u}) &= \mathbbm{1}\left\{R_{2}(\mu,\vec{u}) > 0 \cup R_{2}(\mu,\vec{u})+R_{3}(\mu,\vec{u}) > 0 \right\}, \\
\alpha_3^{\mu}(\vec{u}) &= \mathbbm{1}\left\{R_{3}(\mu,\vec{u}) > 0 \right\}.
\end{aligned}
\end{equation}
and $R_{i}(\mu,\vec{u})=a_i(\vec u)-\sum_{l=0}^{2}\mu_lb_{l,i}(\vec u), i=1,2,3$ is defined as in Lemma \ref{lem:dual_problem}.
\end{lemma}
\begin{proof}
See the proof in Appendix \ref{app:opt_decision_mu}. 
\end{proof}

The specific expressions for the coefficient functions $a_i(\vec{u})$ and $b_{l,i}(\vec{u})$ are implicitly defined by aligning the objective and constraint functions in Lemma \ref{lemma:optimization_problem_k_3} with the general linear forms in \eqref{eq:power_general} and \eqref{eq:error_general}. 
Their detailed derivation is provided in the proof of Lemma \ref{lem:opt_decision_mu}.

A key property of these coefficient functions, which is evident from their derivation (see proof of Lemma \ref{lem:opt_decision_mu}), is their non-negativity. 
We state this formally in the following lemma, which is presented without proof.

\begin{lemma} 
\label{lem:nonnegative-as-bs}
For the coefficient functions $a_i(\vec{u})$ and $b_{l,i}(\vec{u})$ corresponding to the power and error functions in the $K=3$ problem \eqref{eq:objconst}, we have:
\begin{equation}
\label{eq:non_negative_ai_bli}
a_i(\vec{u}) \geq 0 \quad \text{and} \quad b_{l,i}(\vec{u}) \geq 0 \quad \text{for all } i = 1,2,3 \text{ and } l = 0,1,2.
\end{equation}
\end{lemma}

Finally, for notational convenience in subsequent sections, we define the complements of the indicator functions $\alpha_i^\mu(\vec{u})$ as $\beta_i^{\mu}(\vec{u}) = 1-\alpha_i^{\mu}(\vec{u})$, for $i=1,2,3$.

\subsection{Expression of the Dual Problem}
\label{sec:explorelagrangian}

This subsection derives the explicit form of the Lagrangian dual objective function, $L(\vec{D}^\mu, \mu)$, for our specific $K=3$ problem. 
This is achieved by substituting the optimal policy $\vec{D}^\mu$ (from Lemma \ref{lem:opt_decision_mu}) into the general Lagrangian (from \eqref{eq:lagrangian1}) and rearranging the terms based on the dual multipliers $\mu$.

We begin by reformulating the general Lagrangian $L(\vec{D}^\mu, \mu)$ from \eqref{eq:lagrangian1} to explicitly articulate its dependence on $\mu_0, \mu_1,$ and $\mu_2$. 
By substituting the coefficients $a_i$ and $b_{l,i}$ from \eqref{eq:objconst} and the optimal policy $\vec{D}^\mu$ from \eqref{eq:dimus2} into the Lagrangian, we can group terms to isolate the contribution of each multiplier. 
This rearrangement yields a more tractable structure for the dual minimization problem.

The following lemma presents the resulting objective function for the dual problem, which is an unconstrained minimization problem with respect to $\mu$.

\begin{lemma}
\label{lemma:lagrangian_intermsof_mu0}
Consider the $K=3$ MHT problem \eqref{eq:objconst} and its dual formulation \eqref{eq:ld_mu}. 
The dual objective function $L(\vec{D}^\mu, \mu)$, defined in \eqref{eq:lagrangian1} and evaluated at the optimal policy $\vec{D}^\mu$ from Lemma \ref{lem:opt_decision_mu}, admits the following decomposition:
\begin{equation}\label{eq:lagrangian5_general}
\begin{aligned}
L(\vec{D}^\mu, \mu) =&\alpha(\mu_0+\mu_1+\mu_2)\\
&+2\int_Qg(u_1) g(u_2) g(u_3)\alpha_1^{\mu}(\vec{u}) \biggl(1 + \alpha_2^{\mu}(\vec{u})+\alpha_2^{\mu}(\vec{u})\alpha_3^{\mu}(\vec{u}) \biggr)d\vec{u}\\
&-6\int_Q \alpha_1^\mu(\vec{u})\mu_0d\vec{u} \\
&-2\int_Q \alpha_1^\mu(\vec{u})\Bigl(g(u_1)\alpha_2^{\mu}(\vec{u})+g(u_2) + g(u_3)\Bigr)\mu_1d\vec{u}\\
& -2\int_Q\alpha_{1}^{\mu}(\vec{u})\Bigl(g(u_{2})g(u_{3})+\alpha_{2}^{\mu}(\vec{u})g(u_{1})g(u_{3})+\alpha_{2}^{\mu}(\vec{u})\alpha_{3}^{\mu}(\vec{u})g(u_{1})g(u_{2})\Bigr)\mu_2 d\vec{u}.
\end{aligned}
\end{equation}
\end{lemma}
\begin{proof}
The proof is provided in Appendix \ref{app:lagrangian_intermsof_mu0}.
\end{proof}

As established in Section \ref{sec:lagdualprob} (Proposition \ref{prop:optmu} and Assumption \ref{assumption:optmu}), strong duality holds for this problem. 
Therefore, finding the vector $\mu^* = (\mu_0^\ast, \mu_1^\ast, \mu_2^\ast)$ that minimizes the function in \eqref{eq:lagrangian5_general} is equivalent to solving the original primal optimization problem \eqref{eq:MHTK3}.

\subsection{Optimality Conditions for Dual Lagrangian Minimization}
\label{sec:condition-optimal}

This subsection establishes the necessary conditions for a vector $\mu^*$ to be a minimizer of the dual objective function $L(\vec{D}^\mu, \mu)$. 
We first introduce regularity conditions on the alternative density $g(\cdot)$ and then derive the coordinate-wise optimality conditions that $\mu^*$ must satisfy.

To derive the optimality conditions for minimizing $L(\vec{D}^\mu, \mu)$ with respect to $\mu$, we first impose several regularity conditions on the alternative density function $g(\cdot)$ defined in \eqref{eq:mtp1st1}.

The first assumption, a lower Lipschitz bound, ensures that the density $g(\cdot)$ is strictly monotonic and does not flatten out over any interval, which guarantees a sufficiently responsive relationship between p-values and their densities.
\begin{assumption}[Lower Lipschitz Bounds]
\label{as:assumption3}
There exists a constant $c_3 > 0$ such that for all $u, u' \in [0,1]$,
\begin{equation}
|g(u) - g(u')| \ge c_3 |u - u'|.
\end{equation}
\end{assumption}

The second assumption, strict positivity, requires that the alternative density be bounded away from zero. 
This is necessary to ensure that terms involving $g(u)$ in the Lagrangian do not vanish.
\begin{assumption}[Strict Positivity]
\label{as:assumption4}
There exists a constant $c_4 > 0$ such that for all $u \in [0,1]$,
\begin{equation}
g(u) \ge c_4.
\end{equation}
\end{assumption}

The third assumption provides a simple upper bound on the density function, which ensures all integrals are well-defined.
\begin{assumption}[Upper Bound]
\label{as:assumption5}
There exists a constant $c_5 > 0$ such that for all $u \in [0,1]$,
\begin{equation}\label{eq:assumption5}
g(u) \le c_5.
\end{equation}
\end{assumption}

Under these assumptions, we can characterize the minimizer of the Lagrangian $L(\vec{D}^\mu, \mu)$ by analyzing its partial derivatives. 
The following theorem presents the necessary conditions for optimality, derived by considering small perturbations in each coordinate $\mu_i$ while holding the others fixed.

\begin{theorem}[Optimality Conditions for Minimizing $L(\vec{D}^\mu, \mu)$]
\label{thm:mu_optimality_combined}
Suppose Assumptions~\ref{as:assumption3},~\ref{as:assumption4}, and~\ref{as:assumption5} hold. 
Consider the dual formulation \eqref{eq:ld_mu} of the MHT problem \eqref{eq:objconst}, with the Lagrangian $L(\vec{D}^\mu, \mu)$ as defined in \eqref{eq:lagrangian5_general}. 
Then, for fixed values of the other coordinates, the following characterizations of a local minimizer $\mu^\ast = (\mu_0^\ast, \mu_1^\ast, \mu_2^\ast)$ hold:
\begin{enumerate}
  \item  For fixed $\mu_1$ and $\mu_2$, the coordinate $\mu_0^\ast \in \mathbb{R}_+$ is a local minimizer of $L(\vec{D}^\mu, \mu)$ with respect to $\mu_0$ only if
  \begin{equation}
  \label{eq:mu_0_optimality_condition}
  \alpha = 6 \int_Q \alpha_1^{(\mu_0^\ast,\mu_1,\mu_2)}(\vec{u}) d\vec{u}.
  \end{equation}

\item  For fixed $\mu_0$ and $\mu_2$, the coordinate $\mu_1^\ast \in \mathbb{R}_+$ is a local minimizer of $L(\vec{D}^\mu, \mu)$ with respect to $\mu_1$ only if
  \begin{equation}
  \label{eq:mu_1_optimality_condition}
  \alpha = 1 - 2 \int_Q \beta_2^{(\mu_0,\mu_1^\ast,\mu_2)}(\vec{u})g(u_1) d\vec{u}.
  \end{equation}

\item  For fixed $\mu_0$ and $\mu_1$, the coordinate $\mu_2^\ast \in \mathbb{R}_+$ is a local minimizer of $L(\vec{D}^\mu, \mu)$ with respect to $\mu_2$ only if
  \begin{equation}
  \label{eq:mu_2_optimality_condition}
  \alpha = 2 \int_Q \alpha_3^{(\mu_0,\mu_1,\mu_2^\ast)}(\vec{u}) g(u_1)g(u_2) d\vec{u}.
  \end{equation}
\end{enumerate}
\end{theorem}
\begin{proof}
The proof is detailed in Appendix \ref{app:mu_optimality}.
\end{proof}

\subsection{Coordinatewise Monotonicity in the Target Functions}
\label{sec:monotonicty-in-Fs}

This subsection establishes a crucial property of the optimality conditions derived in Theorem \ref{thm:mu_optimality_combined}. 
We demonstrate that each of the `target functions' (the right-hand sides of the equations in Theorem \ref{thm:mu_optimality_combined}) is a non-increasing function of its respective dual variable. 
This monotonicity is the key property that ensures a unique solution and enables the use of efficient root-finding algorithms.

We begin by observing that the optimality conditions in Theorem \ref{thm:mu_optimality_combined} share a common structure: they each equate a function of a single $\mu_i$ (holding the other coordinates fixed) to the constant $\alpha$. 
The following lemma establishes the coordinate-wise monotonicity of these functions.

\begin{lemma}
\label{lem:monotonic_optimality}
Consider the functions on the right-hand sides of the optimality conditions in Theorem \ref{thm:mu_optimality_combined}. 
Provided that the error function coefficients $b_{l,i}(\vec{u}) \ge 0$ for all $l \in \{0, 1, 2\}$ and $i \in \{1, 2, 3\}$ (as established in Lemma \ref{lem:nonnegative-as-bs}), these functions are non-increasing in their respective coordinates. 
Equivalently, the following maps are non-increasing on $\mathbb{R}_+$:
\begin{enumerate}
    \item For fixed $\mu_1, \mu_2 \in \mathbb{R}_+$, the map
\begin{equation}
\label{eq:mu0_mapping}
\mu_0 \mapsto \int_Q \alpha_1^{(\mu_0, \mu_1, \mu_2)}(\vec{u}) d\vec{u}
\end{equation}
is non-increasing with respect to $\mu_0$.

    \item For fixed $\mu_0, \mu_2 \in \mathbb{R}_+$, the map
\begin{equation}
\label{eq:mu1_mapping}
\mu_1 \mapsto 1 - 2 \int_Q \beta_2^{(\mu_0,\mu_1,\mu_2)}(\vec{u})g(u_1) d\vec{u}
\end{equation}
is non-increasing with respect to $\mu_1$.

    \item For fixed $\mu_0, \mu_1 \in \mathbb{R}_+$, the map
\begin{equation}
\label{eq:mu2_mapping}
\mu_2 \mapsto 2 \int_Q \alpha_3^{(\mu_0,\mu_1,\mu_2)}(\vec{u}) g(u_1)g(u_2) d\vec{u}
\end{equation}
is non-increasing with respect to $\mu_2$.
\end{enumerate}
\end{lemma} 
\begin{proof}
The proof is provided in Appendix \ref{app:monotonic_optimality}.
\end{proof}

The monotonicity established in Lemma \ref{lem:monotonic_optimality} is instrumental for computation, as it transforms the search for each optimal $\mu_i^\ast$ (holding the others fixed) into a one-dimensional root-finding problem for a monotonic function.
This ensures that standard numerical methods, such as the bisection method, can be applied to efficiently and uniquely find the root.

We can now define these target functions, which we aim to set equal to $\alpha$, as follows:
\begin{equation}
\label{eq:f_gamma}
\begin{aligned}
F_0(\mu_0;\mu_1,\mu_2) &:= 6\!\int_Q \alpha_1^{(\mu_0,\mu_1,\mu_2)}(\vec{u})d\vec{u}, \\
F_1(\mu_1;\mu_0,\mu_2) &:= 1-2\!\int_Q \beta_2^{(\mu_0,\mu_1,\mu_2)}(\vec{u})g(u_1)d\vec{u}, \\
F_2(\mu_2;\mu_0,\mu_1) &:= 2\!\int_Q \alpha_3^{(\mu_0,\mu_1,\mu_2)}(\vec{u})g(u_1)g(u_2)d\vec{u}.
\end{aligned}
\end{equation}
Lemma \ref{lem:monotonic_optimality} guarantees that each function $F_\gamma(\cdot)$ is strictly decreasing in its first argument, and thus possesses at most one root that satisfies $F_\gamma(\mu_\gamma^\ast; \dots) = \alpha$.

The final objective is to find an optimal triplet $(\mu_0^\ast, \mu_1^\ast, \mu_2^\ast)$ that simultaneously solves the system of equations defined by Theorem \ref{thm:mu_optimality_combined}:
\begin{equation}
\label{eq:f_gamma2}
\begin{aligned}
F_0(\mu_0^\ast;\mu_1^\ast,\mu_2^\ast) = \alpha, \\
F_1(\mu_1^\ast;\mu_0^\ast,\mu_2^\ast) = \alpha, \\
F_2(\mu_2^\ast;\mu_0^\ast,\mu_1^\ast) = \alpha.
\end{aligned}
\end{equation}

This formulation as a system of equations, where each component function is monotonic in its primary variable, naturally suggests a coordinate-wise root-finding algorithm, which we will detail in Section \ref{sec:computing01}.

\subsection{More Monotonicity Results}
\label{sec:more-monotone}

This subsection presents additional monotonicity properties of the target functions $F_\gamma$ defined in \eqref{eq:f_gamma}. 
We establish that each function $F_\gamma$ is not only monotonic in its primary coordinate (as shown in Lemma \ref{lem:monotonic_optimality}) but also in its other two coordinate arguments. 
These cross-monotonicity properties are essential for construction and guaranteeing the convergence of the coordinate-update algorithm presented later.

The following lemmas (Lemmas \ref{lem:alpha1_mu12_monotone}–\ref{lem:alpha3_mu01_monotone}) establish that each of the target functions $F_0, F_1, F_2$ is a non-increasing function with respect to all of its arguments, not just the primary coordinate. For example, Lemma \ref{lem:alpha1_mu12_monotone} shows that $F_0(\mu_0;\mu_1,\mu_2)$ decreases as $\mu_1$ or $\mu_2$ increases (in addition to decreasing with $\mu_0$). 
This global monotonicity is a powerful property that ensures the coordinate-update algorithm (Algorithm \ref{alg:compute_optimal_mu_K3_main}) will converge to a unique solution.

\begin{lemma}
\label{lem:alpha1_mu12_monotone}
For any fixed $\mu_{0}\in\mathbb R_{+}$, the mappings 
\[
\mu_{1} \longmapsto \int_{Q}\alpha_{1}^{(\mu_{0},\mu_{1},\mu_{2})}(\vec u)\,d\vec u
\quad\text{and}\quad
\mu_{2} \longmapsto \int_{Q}\alpha_{1}^{(\mu_{0},\mu_{1},\mu_{2})}(\vec u)\,d\vec u
\]
are non-increasing on $\mathbb R_{+}$ when the remaining coordinates are held fixed, provided the coefficients $b_{l,i}(\vec{u}) \ge 0$ for all $l, i$.
\end{lemma}
\begin{proof}
The proof is provided in Appendix \ref{app:alpha1_mu12_monotone}.
\end{proof}

\begin{lemma}
\label{lem:beta2_mu02_monotone}
For any fixed $\mu_{1}\in\mathbb R_{+}$, the mappings 
\[
\mu_{0}\longmapsto 1-2\!\int_{Q}\beta_{2}^{(\mu_{0},\mu_{1},\mu_{2})}(\vec u)\,g(u_{1})\,d\vec u
\quad\text{and}\quad
\mu_{2}\longmapsto 1-2\!\int_{Q}\beta_{2}^{(\mu_{0},\mu_{1},\mu_{2})}(\vec u)\,g(u_{1})\,d\vec u
\]
are non-increasing on $\mathbb R_{+}$ when the remaining coordinates are held fixed, provided the coefficients $b_{l,i}(\vec{u}) \ge 0$ for all $l, i$.
\end{lemma}
\begin{proof}
The proof is provided in Appendix \ref{app:beta2_mu02_monotone}.
\end{proof}

\begin{lemma}
\label{lem:alpha3_mu01_monotone}
For any fixed $\mu_{2}\in\mathbb R_{+}$, the mappings 
\[
\mu_{0}\longmapsto 2\!\int_{Q}\alpha_{3}^{(\mu_{0},\mu_{1},\mu_{2})}(\vec u)\,
         g(u_{1})g(u_{2})\,d\vec u
\quad\text{and}\quad
\mu_{1}\longmapsto 2\!\int_{Q}\alpha_{3}^{(\mu_{0},\mu_{1},\mu_{2})}(\vec u)\,
         g(u_{1})g(u_{2})\,d\vec u
\]
are non-increasing on $\mathbb R_{+}$ when the remaining coordinates are held fixed, provided the coefficients $b_{l,i}(\vec{u}) \ge 0$ for all $l, i$.
\end{lemma}
\begin{proof}
The proof is provided in Appendix \ref{app:alpha3_mu01_monotone}.
\end{proof}

\subsection{Computational Strategy Based on the Optimality Conditions}
\label{sec:computing01}

This final subsection presents our proposed computational strategy for solving the MHT problem \eqref{eq:MHTK3}. 
We leverage the theoretical foundations established, namely, the optimality conditions (Theorem \ref{thm:mu_optimality_combined}) and the monotonicity properties (Lemmas \ref{lem:monotonic_optimality}–\ref{lem:alpha3_mu01_monotone}), to develop a principled and efficient coordinate-update algorithm that computes the optimal Lagrange multipliers $\mu^\ast$.

Algorithm \ref{alg:compute_optimal_mu_K3_main} provides the pseudocode for a coordinate-update scheme designed to find the optimal vector $\hat{\vec{\mu}} = (\mu_0, \mu_1, \mu_2)$ that satisfies the optimality system \eqref{eq:f_gamma2} up to a specified tolerance $\varepsilon$. 
The algorithm initializes $\vec{\mu}^{(0)} = (0,0,0)$ and iteratively cycles through the coordinates. 
In each step, it updates one coordinate, $\mu_\gamma^{(t)}$, by finding the root of its corresponding target function $F_\gamma(\cdot)$ (holding the other two coordinates fixed at their most recent values). 
This root-finding step is performed by the subroutine \texttt{ComputeCoordinateMu} (described in Algorithm \ref{alg:ComputeCoordinateMu} in Appendix \ref{app:algos}). The outer loop continues until the $\ell_2$ norm of the change between successive iterates, $\|\vec{\mu}^{(t)}-\vec{\mu}^{(t-1)}\|_2$, falls below the tolerance $\varepsilon$, or a maximum iteration count $T_{\max}$ is reached.

\begin{algorithm}[htbp]
\caption{Coordinate–Update Algorithm for Optimal $\mu$}
\label{alg:compute_optimal_mu_K3_main}
\begin{algorithmic}[1]
\Statex \textbf{Input:}
\Statex \hspace{0.6em} Desired FWER level $\alpha$, inner tolerance $\delta$, outer tolerance $\varepsilon$, maximum outer iterations $T_{\max}$, search ceiling $U_{\max}$.
\vspace{2pt}

\Statex \textbf{Initialisation:}
\Statex \hspace{0.6em} $\vec{\mu}^{(0)}\gets(0,0,0)$, \; $t\gets 0$.

\While{True}
    \State $t\gets t+1$
    \State $\mu_0^{(t)}\gets  
            \texttt{ComputeCoordinateMu}\bigl(F_0(\cdot;\mu_1^{(t-1)},\mu_2^{(t-1)}),
                                         \alpha,\delta,U_{\max}\bigr)$   \Comment{See \eqref{eq:f_gamma}; update first coordinate}
    \State $\mu_1^{(t)}\gets  
            \texttt{ComputeCoordinateMu}\bigl(F_1(\cdot;\mu_0^{(t)},\mu_2^{(t-1)}),
                                         \alpha,\delta,U_{\max}\bigr)$   \Comment{See \eqref{eq:f_gamma}; update second coordinate}
    \State $\mu_2^{(t)}\gets  
            \texttt{ComputeCoordinateMu}\bigl(F_2(\cdot;\mu_0^{(t)},\mu_1^{(t)}),
                                         \alpha,\delta,U_{\max}\bigr)$ \Comment{See \eqref{eq:f_gamma}; update third coordinate}
    
    \If{$\|\vec{\mu}^{(t)}-\vec{\mu}^{(t-1)}\|_2\le\varepsilon$} \Comment{Convergence check}
        \State \textbf{break}
    \EndIf
    \If{$t\ge T_{\max}$}
        \State \textbf{break}
    \EndIf
\EndWhile
\Statex \textbf{Output:}\quad $\hat{\vec{\mu}}\gets\vec{\mu}^{(t)}$
\end{algorithmic}
\end{algorithm}
The subroutine \texttt{ComputeCoordinateMu} is a robust one-dimensional root-finder. 
It first brackets the root by starting at $L=0$ and expanding an upper bound $U$ (initially $U_s$, then multiplied by $U_f$) until $F_{\gamma}(U) \le \alpha$ or the ceiling $U_{\max}$ is hit. If $F_{\gamma}(0) < \alpha$, it fails, as no non-negative root exists. 
Once the root is bracketed in $[L, U]$, the algorithm applies standard bisection until the tolerance $\delta$ is met or a maximum number of iterations $MaxIter_b$ is reached. 
Given the monotonicity guaranteed by Lemma \ref{lem:monotonic_optimality}, this subroutine reliably returns a $\delta$-accurate root or a clear diagnostic.
Furthermore, in Appendix \ref{sec:convergence} we provide a theoretical guarantee for the convergence of Algorithm \ref{alg:compute_optimal_mu_K3_main}. 
We establish that under mild regularity conditions, the coordinate updates ensure linear convergence to the unique optimal $\mu^\ast$. 
As a result, the computational complexity scales with the logarithm of the reciprocal of the target error, $O(\log(1/\varepsilon))$.

Once the optimal Lagrange multipliers $\hat{\vec{\mu}}=(\hat\mu_{0},\hat\mu_{1},\hat\mu_{2})$ have been computed by Algorithm \ref{alg:compute_optimal_mu_K3_main}, the original MHT problem is effectively solved. 
The optimal decision policy $D^{\star}(\vec{u})$ is then given explicitly by Lemma \ref{lem:opt_decision_mu}:
$$
D^{\star}_{1}(\vec u)=\alpha_{1}^{\hat{\mu}}(\vec u),\qquad
D^{\star}_{2}(\vec u)=\alpha_{1}^{\hat{\mu}}(\vec u)\,\alpha_{2}^{\hat{\mu}}(\vec u),\qquad
D^{\star}_{3}(\vec u)=\alpha_{1}^{\hat{\mu}}(\vec u)\,\alpha_{2}^{\hat{\mu}}(\vec u)\,\alpha_{3}^{\hat{\mu}}(\vec u),
$$
where the indicator functions $\alpha_{k}^{\hat{\mu}}(\vec u)$ are defined in \eqref{eq:alpha_mu_u} using the computed $\hat{\vec{\mu}}$.

Since the vector $\hat{\vec{\mu}}$ satisfies the optimality system \eqref{eq:f_gamma2} (to within tolerance $\varepsilon$), the resulting policy $\vec{D}^{\star}$ is the optimal decision policy. 
See Appendix \ref{sec:d_mu_optimality} for details establishing the strong duality that guarantees this result.
Algorithm \ref{alg:compute_optimal_mu_K3_main} successfully controls the family-wise error rate at the desired level $\alpha$ while simultaneously maximizing the statistical power, thus providing a complete and computationally efficient solution to the MHT problem for $K=3$.

\section{Simulation Study}
\label{sec:simulations}

To assess the practical performance of our proposed method, we conduct four simulation studies comparing Algorithm~\ref{alg:compute_optimal_mu_K3_main} against standard FWER-controlling procedures: Bonferroni (\cite{Bonferroni1936,Miller1981}), Holm (\cite{holm}), Hochberg (\cite{BH95}), Hommel (\cite{hommel}) and the Romano--Wolf stepdown procedure for independent tests (\cite{romano}). 
These studies cover a range of settings: a one-sided truncated normal model, a two-sided mixture normal model, a heavy-tailed $t$-distribution alternative, and a Beta model directly specified on $p$-values. 
In all four cases, our algorithm is configured to maximise the average power $\Pi_3$ while strictly controlling the FWER at $\alpha=0.05$.

Across all designs (Figures~\ref{fig:sim_pi3_trunc}–\ref{tab:sim_pi3_beta}), the resulting $\Pi_3$-optimised policy is substantially more powerful than standard procedures for its target metric $\Pi_3$, often by a wide margin for moderate and strong signals, and typically also delivers comparable or superior minimal power $\Pi_{\text{any}}$. 
Details of the simulation setup, including the construction of $p$-values, the corresponding alternative densities $g(u)$ and verification of our assumptions, are given in Appendix~\ref{app:simulations}.

\paragraph{\texorpdfstring{Testing $K=3$ Independent Normal Means}{Testing K=3 Independent Normal Means}}

We first consider $K=3$ independent and identical one-sided normal mean tests with truncation to enforce our theoretical assumptions. 
For each $k=1,2,3$, the test statistic $X_k$ is drawn from a normal distribution truncated to $[-M,M]$ with $M=6$:
\begin{equation}
\label{eq:sim1_trunc_normal_mean_mht}
    H_{0k}: X_k \sim N(0,1)_{\text{trunc}} \quad \text{vs.} \quad 
    H_{Ak}: X_k \sim N(\theta,1)_{\text{trunc}}, \qquad k=1,2,3,
\end{equation}
with $\theta$ varying from $0$ down to $-4.0$. 
Under $H_{0k}$ we construct $p$-values via the truncated-null CDF so that $u_k \sim U(0,1)$, and derive the corresponding alternative density $g(u)$ by the probability integral transform; the explicit form and verification of the assumptions on $g$ are given in Appendix~\ref{app:sim_truncnormal_details}.

Figure~\ref{fig:sim_pi3_trunc} reports the average power $\Pi_3$ (left) and minimal power $\Pi_{\text{any}}$ (right) under the global alternative (all three hypotheses false), based on $N=120{,}000$ Monte Carlo replications. 
As expected, our $\Pi_3$-optimised policy achieves the highest average power across signal strengths; for example, at $\theta=-2.0$ the average power is $0.650$ versus $0.555$ for Hommel, and at $\theta=-3.0$ the values are $0.965$ versus $0.900$. 
The same policy also uniformly dominates competing procedures for $\Pi_{\text{any}}$ (e.g.\ $0.932$ versus $0.853$ for Hommel when $\theta=-2.0$), indicating that the gain is not restricted to a single power metric but provides a generally more powerful test. 
Furthermore, at the global null ($\theta=0$), all procedures maintain FWER near $0.05$, while our optimised policy is conservative (empirical FWER $\approx 0$).

\begin{figure}[htbp]
\centering
\includegraphics[width=0.7\textwidth]{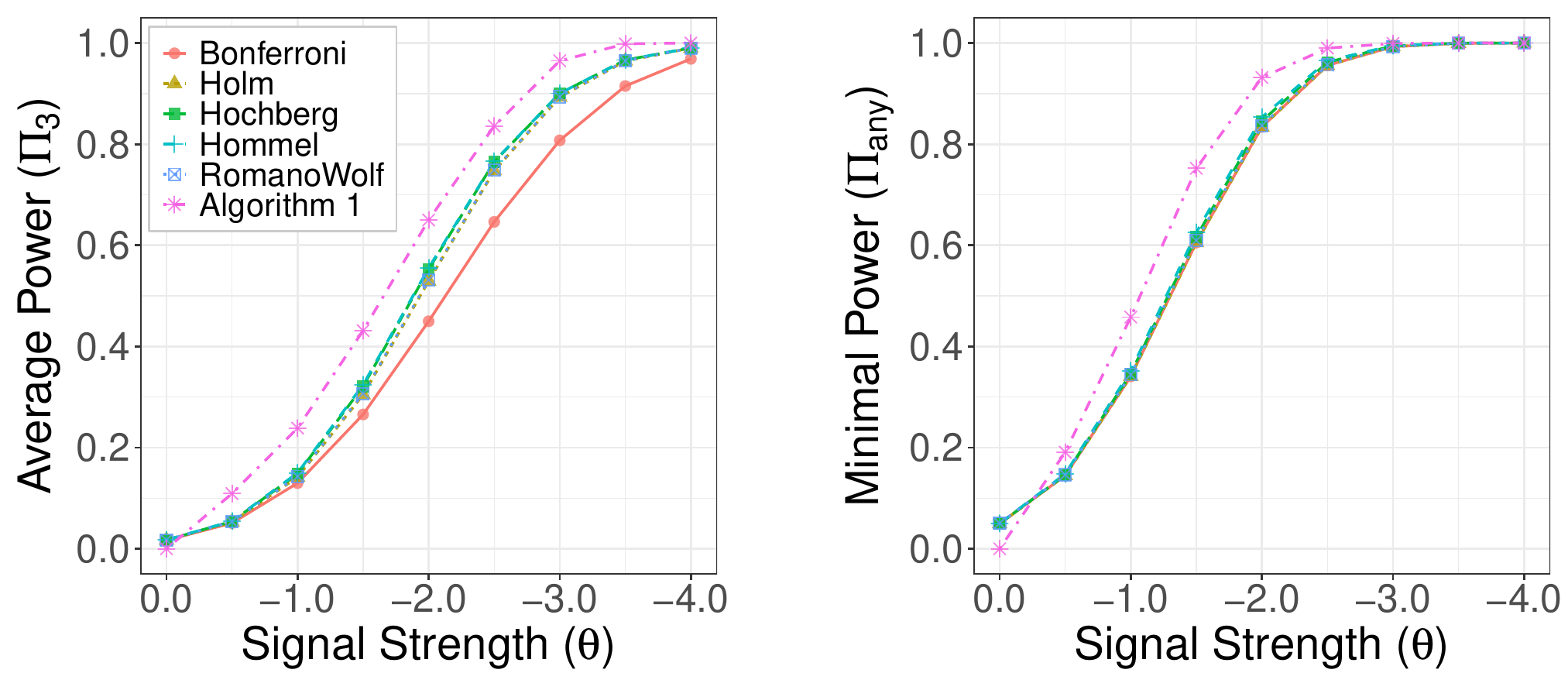}
\caption{Average power $\Pi_3$ (left) and minimal power $\Pi_{\text{any}}$ (right) of the $\Pi_3$-optimised policy under the truncated normal alternative with all three hypotheses false. 
}
\label{fig:sim_pi3_trunc}
\end{figure}

\paragraph{Testing Two-Sided Mixture Normal Alternatives}
We next consider a two-sided testing problem with a symmetric Gaussian mixture alternative:
\begin{equation}
\label{eq:sim2_mixture_normal_mht}
    H_{0k}: X_k \sim N(0,1) \quad \text{vs.} \quad 
    H_{Ak}: X_k \sim 0.5\,N(\theta,1) + 0.5\,N(-\theta,1), \qquad k=1,2,3,
\end{equation}
with $\theta$ again varying from $0$ to $-4.0$. 
We use standard two-sided $p$-values $u_k = 2 \Phi(-|X_k|)$, which are uniform under $H_{0k}$; the induced alternative density $g(u)$ is obtained via the likelihood ratio and shown in Appendix~\ref{app:sim_mixture_details} to be monotone and strictly positive, satisfying our key assumptions, though unbounded near~0, violating the regularity condition in Assumption \ref{as:assumption5}.

Figure~\ref{fig:sim_mixture_pi3} shows that Algorithm~\ref{alg:compute_optimal_mu_K3_main} again yields higher average power than standard procedures across the range of $\theta$. 
At $\theta=-2.0$ the $\Pi_3$ of our method is $0.504$ versus $0.415$ for Hommel, increasing to $0.905$ versus $0.828$ at $\theta=-3.0$. 
The minimal power $\Pi_{\text{any}}$ for the same policy is also uniformly higher (e.g.\ $0.804$ versus $0.737$ at $\theta=-2.0$), indicating a generally more powerful test, and the procedure remains conservative at the global null (empirical FWER $\approx 0$). 
Taken together, these results show that the gains from the optimal policy extend to multimodal, two-sided alternatives and are robust to violations of the boundedness assumption on $g$, further underscoring the robustness of Algorithm~\ref{alg:compute_optimal_mu_K3_main} to mild departures from the theoretical conditions.

\begin{figure}[htbp]
\centering
\includegraphics[width=0.7\textwidth]{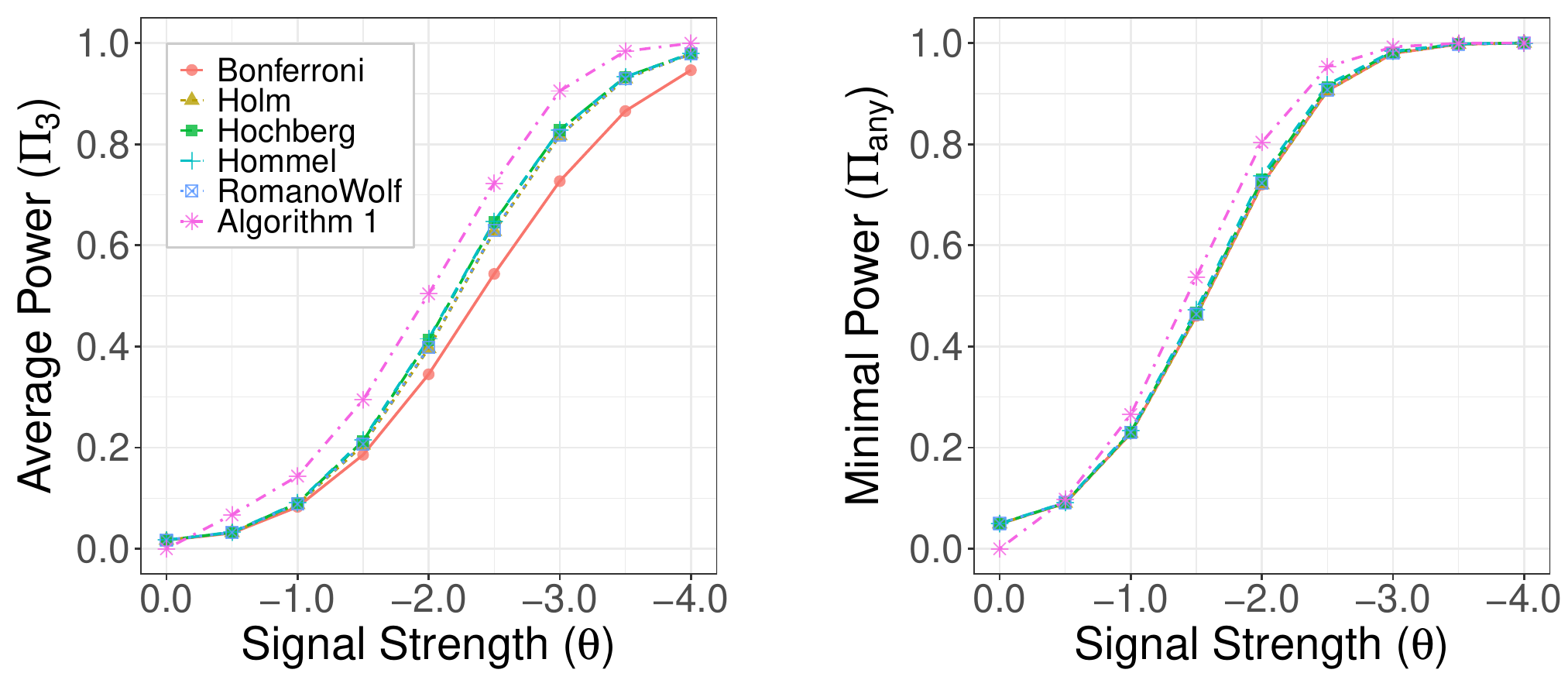}
\caption{Average power $\Pi_3$ (left) and minimal power $\Pi_{\text{any}}$ (right) of the $\Pi_3$-optimized policy under the two-sided mixture normal model with all three hypotheses false. 
}
\label{fig:sim_mixture_pi3}
\end{figure}

\paragraph{\texorpdfstring{Testing Heavy-Tailed Alternatives ($t$-Distribution)}{Testing Heavy-Tailed Alternatives (t-Distribution)}}

To examine robustness to heavy tails, we consider a Student $t$-distribution alternative:
\begin{equation}
\label{eq:sim3_t_dist_mht}
    H_{0k}: X_k \sim N(0,1) \quad \text{vs.} \quad 
    H_{Ak}: X_k \sim t_{\text{df}}, \qquad k=1,2,3,
\end{equation}
with degrees of freedom $\text{df} \in \{2,4,\dots,20\}$. 
Two-sided $p$-values $u_k = 2\Phi(-|X_k|)$ are again uniform under $H_{0k}$; the corresponding density $g(u)$ is derived and is shown in Appendix~\ref{app:sim_t_details} to be non-increasing, although unbounded near $0$.

Figure~\ref{fig:sim_t_dist} displays the results. 
For the target metric $\Pi_3$ (left), our algorithm consistently achieves higher power than the benchmarks; at $\text{df}=2$ the average power is $0.166$ versus $0.146$ for Hommel, with similar relative gains across other degrees of freedom. 
For $\Pi_{\text{any}}$ (right), the $\Pi_3$-optimised policy is the performance of our algorithm is comparable to the standard procedures (e.g.\ at $\text{df}=2$ it yields $0.322$ versus $0.365$ for Hommel).
Consistent with the previous simulations, the algorithm maintains stable convergence and exact FWER control even when the alternative density is unbounded, confirming robustness to the boundedness assumption and yielding a generally more powerful test.

\begin{figure}[htbp]
\centering
\includegraphics[width=0.7\textwidth]{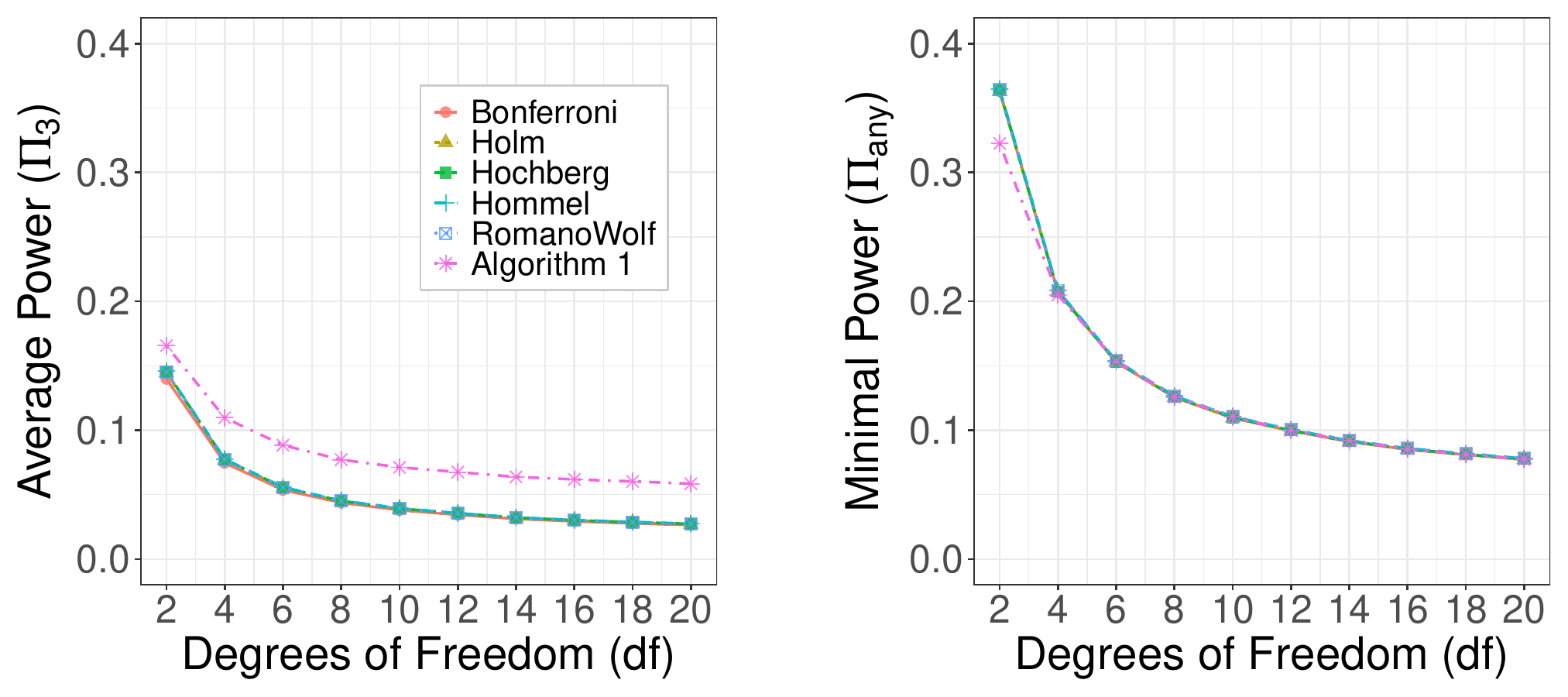}
\caption{Average power $\Pi_3$ (left) and minimal power $\Pi_{\text{any}}$ (right) under the $t$-distribution alternative. Algorithm~\ref{alg:compute_optimal_mu_K3_main} is optimised for $\Pi_3$, showing consistent power gains (left).}
\label{fig:sim_t_dist}
\end{figure}

\paragraph{Testing Parameters of the Beta Density Function}

Finally, we consider a setting where the $p$-values themselves follow a Beta distribution. 
Under the null, $u_k \sim \text{Beta}(1,1)$ (uniform on $(0,1)$), while under the alternative $u_k \sim \text{Beta}(\theta,1)$ with $\theta < 1$:
\begin{equation}
\label{eq:sim2_beta_mht}
    H_{0k}: u_k \sim \text{Beta}(1,1) \quad \text{vs.} \quad 
    H_{Ak}: u_k \sim \text{Beta}(\theta,1), \qquad k=1,2,3,
\end{equation}
with $\theta \in \{0.8,0.6,0.4,0.2\}$. 
As $\theta$ decreases, the alternative places more mass near $0$, corresponding to stronger signals. 
Here the alternative density $g(u)$ is simply the $\text{Beta}(\theta,1)$ density, $g(u) = \theta u^{\theta-1}$; full details are in Appendix~\ref{app:sim_beta_details}.

Figure~\ref{tab:sim_pi3_beta} shows that our policy achieves uniformly higher average power $\Pi_3$ than Hommel across all $\theta$ (e.g.\ $0.576$ vs.\ $0.490$ at $\theta=0.2$), and also substantially higher $\Pi_{\text{any}}$ (e.g.\ $0.854$ vs.\ $0.830$ at $\theta=0.2$), confirming that the gains persist in a fully $p$-value-based setting.

\begin{figure}[htbp]
\centering
\includegraphics[width=0.7\textwidth]{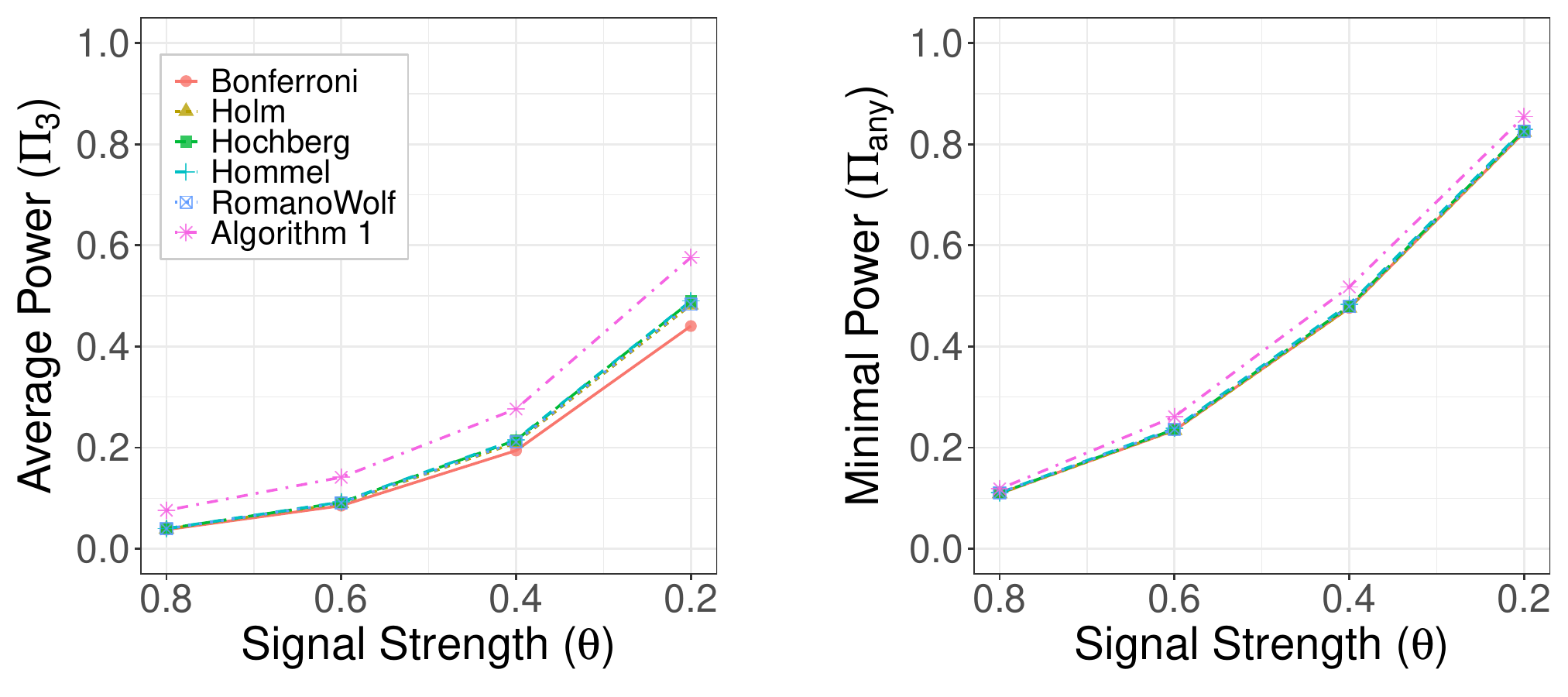}
\caption{Average power $\Pi_3$ (left) and any-discovery power $\Pi_{\text{any}}$ (right) for the $\Pi_3$-optimised policy under the $\text{Beta}(\theta,1)$ alternative. 
}
\label{tab:sim_pi3_beta}
\end{figure}

\section{Real-World Applications}
\label{sec:realdata}
To demonstrate the practical value and domain-general applicability of our algorithm, we apply it to two distinct real-world datasets from different scientific fields: a clinical trial subgroup analysis (Section~\ref{sec:bcg}) and a financial asset pricing problem (Section~\ref{sec:finance_app}). 
Both experiments represent common $K=3$ testing scenarios where standard FWER methods are known to suffer from a loss of power.
In both domains, these results show that our optimal test can find practical, significant, and interpretable results that are overlooked by standard procedures.
Details regarding the experiment setup, code, and datasets are provided in Appendix \ref{app:experiments}.

\subsection{BCG Vaccine Subgroup Analysis}
\label{sec:bcg}
We evaluate Algorithm~\ref{alg:compute_optimal_mu_K3_main} on the open BCG vaccine dataset ($K=3$, $\alpha=0.05$) across four subgroup splits, benchmarking against Holm, Hommel, and Closed-Stouffer (\cite{Marcus1976}, \cite{Stouffer1949}).
While all methods agree where evidence is strong, our algorithm uniquely rejects the \texttt{low\_risk} subgroup ($p\approx 0.094$) in the \emph{baseline-risk} (underlying incidence) split.
Algorithm \ref{alg:compute_optimal_mu_K3_main} detects this borderline effect, providing clinically relevant evidence of the vaccine efficacy in low-incidence populations that standard procedures miss.

\paragraph{Dataset, subgroup construction and hypotheses}
We analyse the public \texttt{dat.bcg} dataset from the \texttt{metadat} package, which aggregates independent trials comparing BCG vaccination against a control.
For each trial, we extract the counts of tuberculosis cases and non-cases in the vaccinated (treatment) (\texttt{tpos}, \texttt{tneg}) and control (\texttt{cpos}, \texttt{cneg}) arms.
We also utilise the allocation method (\texttt{alloc} $\in\{\text{random}, \text{alternate}, \text{systematic}\}$) to proxy for trial conduct, and absolute latitude (\texttt{ablat}) to capture environmental variance.
This dataset provides a robust, reproducible benchmark due to its heterogeneity across diverse eras, geographies, and experimental protocols.
See Appendix \ref{app:bcg_details} for additional details.

We define the MHT we are solving in this scenario.
A “three-way split” means dividing the trials into three \emph{non-overlapping} subgroups along one meaningful axis, then analyzing each subgroup separately.
We consider four splits:
\begin{enumerate}\setlength\itemsep{0pt}
\item \emph{Allocation (trial conduct):} $S_1=\{\texttt{alloc}=\text{random}\}$, $S_2=\{\texttt{alloc}=\text{alternate}\}$, $S_3=\{\texttt{alloc}=\text{systematic}\}$.
This isolates differences in assignment and concealment.
\item \emph{Latitude (geography/environment):} Defined by sample tertiles of absolute latitude (\texttt{ablat}): $S_1$ (low), $S_2$ (mid), $S_3$ (high).
Geography proxies background exposure to environmental mycobacteria.
\item \emph{Baseline risk (underlying incidence):} Defined by tertiles of the control-arm risk $\mathrm{CER}=\texttt{cpos}/(\texttt{cpos}+\texttt{cneg})$: $S_1$ (low), $S_2$ (mid), $S_3$ (high).
This tests whether benefit is consistent in low- versus high-incidence settings.
\item \emph{Era (time proxy):} Defined by tertiles of the trial date: $S_1$ (early), $S_2$ (mid), $S_3$ (late).
\end{enumerate}

For a fixed split with subgroups $S_1, S_2, S_3$, let $\theta_k$ denote the true log risk ratio (BCG vs.\ control) in subgroup $S_k$.
We formulate the simultaneous test of $K=3$ hypotheses:
\begin{equation}
\label{eq:mht_experiment}
H_{0k}:\ \theta_k = 0 \qquad \text{vs.} \qquad H_{Ak}:\ \theta_k \neq 0,\qquad k=1,2,3,
\end{equation}
with strong FWER control at \(\alpha=0.05\) \emph{within that split}.
For each subgroup, we compute a pooled log risk ratio and standard error using fixed-effects meta-analysis (see Appendix \ref{app:bcg_details} for estimation details) and convert these to two-sided $p$-values.

\paragraph{Results} We compare the rejection decisions of Algorithm \ref{alg:compute_optimal_mu_K3_main} against the baseline procedures.
In splits where the evidence is uniformly strong, specifically the \emph{allocation}, \emph{latitude}, and \emph{era} splits, all procedures reject all three null hypotheses (see Appendix \ref{app:bcg_details} for tabulated results).
This confirms that our method maintains power comparable to standard approaches when signals are unambiguous.

However, the \emph{baseline-risk} outcome contains a borderline subgroup where methods differ.
The pooled \(p\)-values for the low, mid, and high risk subgroups are $p_{\text{low}} \approx 0.0942$, $p_{\text{mid}} \approx 8.38\times 10^{-28}$, and $p_{\text{high}} \approx 1.02\times 10^{-24}$ respectively (Figure \ref{fig:risk_separate}).
Holm, Hommel, and Closed-Stouffer fail to reject the null for the \texttt{low\_risk} subgroup.
In contrast, Algorithm \ref{alg:compute_optimal_mu_K3_main} rejects $H_{0,\text{low}}$, leveraging the decisive evidence in the other two subgroups to gain power on the borderline case while maintaining exact FWER control.

This result has substantive implications for public health.
The \texttt{low\_risk} subgroup corresponds to populations where TB incidence is naturally low.
By rejecting the null hypothesis for this group, our method provides statistical support that the vaccine remains effective in low-incidence settings, a conclusion missed by standard procedures.

\begin{figure}[htbp]
\centering
\includegraphics[width=0.7\textwidth]{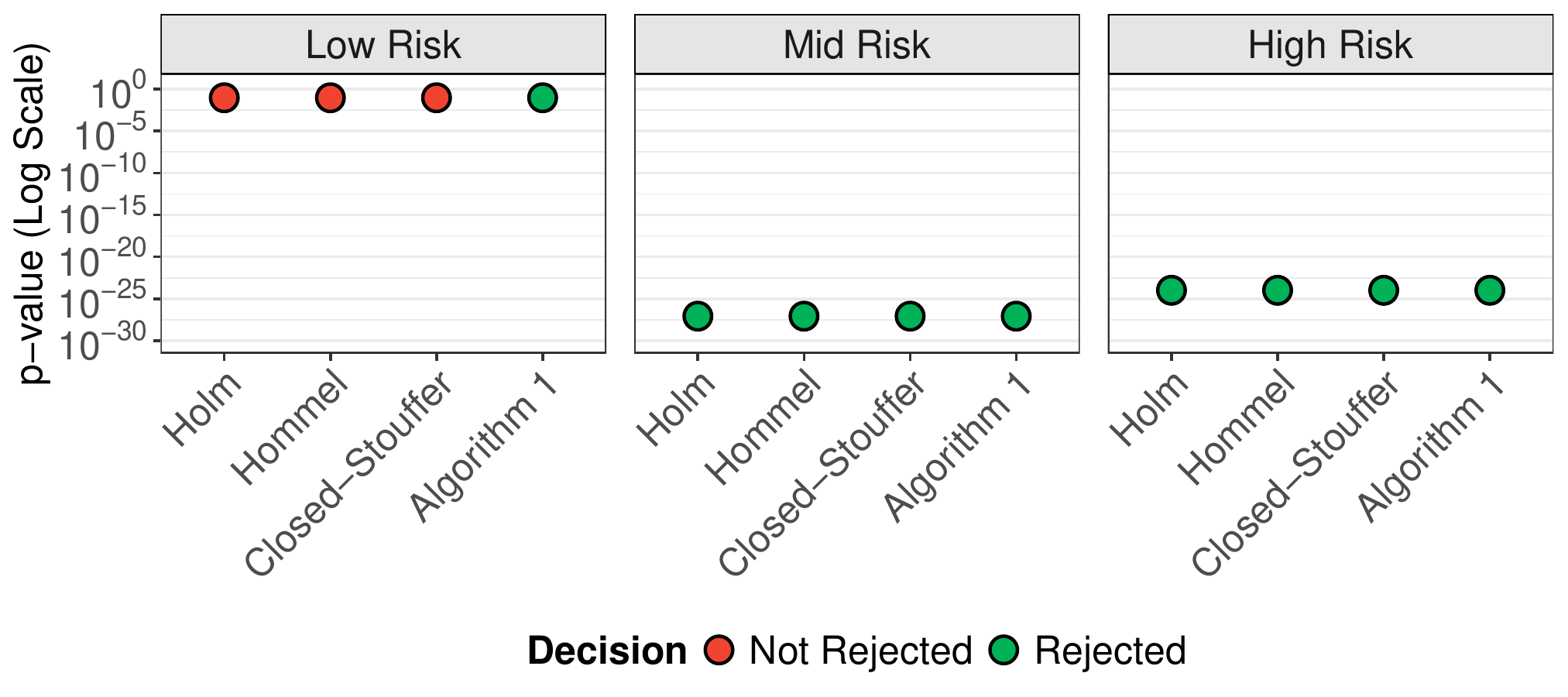}

\caption{Baseline-risk outcome: pooled \(p\)-values and decisions at \(\alpha=0.05\). Standard methods fail to reject the first hypothesis ($p \approx 0.094$), while Algorithm \ref{alg:compute_optimal_mu_K3_main} successfully rejects all three.}
\label{fig:risk_separate}
\end{figure}

\subsection{Finding Missed Discoveries in the Financial Factor Zoo}
\label{sec:finance_app}
To demonstrate the practical value of our proposed algorithm across domains, we apply it to a central problem in empirical finance: identifying genuine return-predicting signals from the `factor zoo' (\cite{cochrane11}; \cite{harvey16}).
We use the open-source Fama-French data library \cite{famadata} to test if three prominent factors: Momentum, Profitability, and Investment, provide significant explanatory power beyond the foundational Fama and French (1993) 3-factor model \cite{fama93}.
This creates a multiple testing problem where standard adjustments like Bonferroni or Holm are often criticized as overly conservative, substantially reducing power and causing missed true discoveries (\cite{zhu4}).

We run $K=3$ separate time-series regressions using monthly data from July 1963 to the present.
The test factors are defined as follows:
\begin{enumerate}
\item $R_{\text{MOM},t}$: The Momentum (MOM) factor, which buys stocks that have performed well in the recent past (winners) and sells stocks that have performed poorly (losers).
\item $R_{\text{RMW},t}$: The Profitability (Robust-Minus-Weak) factor, which buys stocks of firms with high profitability and sells stocks of firms with low profitability.
\item $R_{\text{CMA},t}$: The Investment (Conservative-Minus-Aggressive) factor, which buys stocks of firms that invest conservatively and sells stocks of firms that invest aggressively.
\end{enumerate}
These factors are theoretically significant, with RMW and CMA forming the basis of the later Fama-French 5-factor model \cite{fama15}.

The control variables correspond to the standard Fama-French (1993) 3-factors: the market excess return ($Mkt_t$), the Size factor ($SMB_t$), and the Value factor ($HML_t$); detailed definitions are provided in Appendix \ref{app:finance_details}.
Using the risk-free rate $RF_t$, we estimate the following regression for each test factor $k \in \{\text{MOM, RMW, CMA}\}$:
\begin{equation} \label{eq:ff_regression}
(R_{k,t} - RF_t) = \alpha_k + \beta_{k,mkt}(Mkt_t) + \beta_{k,smb}(SMB_t) + \beta_{k,hml}(HML_t) + \epsilon_{k,t}
\end{equation}
The parameter of interest is the intercept, $\alpha_k$, which represents the `abnormal return' not explained by the 3-factor model.
We simultaneously test the family of $K=3$ composite null hypotheses that these alphas are zero:
$$H_{0,k}: \alpha_k = 0 \quad \text{vs.} \quad H_{A,k}: \alpha_k \neq 0 \quad \text{for } k=1, 2, 3.$$

We model the alternative $p$-value density $g(u)$ using a Student-$t$ distribution with 4 degrees of freedom ($t_4$), motivated by heavy tails of asset and factor returns \citep{bollerslev1987, harvey2000}.
We apply our algorithm and standard FWER-controlling procedures at $\alpha = 0.05$.
The resulting two-sided $p$-values for the intercepts are $p_1 = 0.0023$ (MOM), $p_2 = 0.574$ (RMW), and $p_3 = 0.0008$ (CMA).

\paragraph{Results} 
Table \ref{tab:ff_results} summarizes the results.
Panel A shows that all methods correctly reject the strong Momentum and Investment signals ($p < 0.003$).
Crucially, only our algorithm rejects the Profitability (RMW) factor.
Its high marginal $p$-value is a known artifact of collinearity with the Value factor \citep{fama15}, causing standard methods like Holm to fail.
However, by evaluating the joint evidence in the full $p$-value vector, the optimal policy $\vec{D}^{\mu^*}$ leverages the decisive MOM and CMA signals to correctly identify RMW as significant while guaranteeing FWER $\le 0.05$.
Panel B validates this via simulation, confirming robust FWER control (0.049) and a substantial gain in average power ($\Pi_3$) over competing methods.
\begin{table}[htbp]
\centering
\begin{tabular}{lrrrrrrr}
\multicolumn{8}{l}{\textbf{Panel A: Rejection decisions for three Fama-French factors}} \\
\hline
Hypothesis & $p$-value & Bonf. & Holm & Hoch. & Homm. & RW & Alg. 1 \\
\hline
Momentum (MOM) & 0.0023 & R & R & R & R & R & \textbf{R} \\
Profitability (RMW) & 0.5741 & . & . & . & . & . & \textbf{R} \\
Investment (CMA) & 0.0008 & R & R & R & R & R & \textbf{R} \\
\hline
\multicolumn{8}{l}{\small Note: 'R' indicates rejection of the null hypothesis. `.' indicates failure to reject.}
\end{tabular}

\vspace{1.5em}

\begin{tabular}{lrrr}
\multicolumn{4}{l}{\textbf{Panel B: Simulated FWER and Power under an assumed $t_4$ alternative}} \\
\hline
Method & FWER & Average Power ($\Pi_3$) & Minimal Power ($\Pi_{\text{any}}$) \\
\hline
Bonferroni & 0.050 & 0.074 & 0.207 \\
Holm & 0.050 & 0.077 & 0.207 \\
Hochberg & 0.050 & 0.077 & 0.208 \\
Hommel & 0.050 & 0.077 & 0.209 \\
Romano-Wolf & 0.051 & 0.077 & 0.209 \\
Algorithm \ref{alg:compute_optimal_mu_K3_main} &\textbf{ 0.049} & \textbf{0.110} & \textbf{0.205} \\
\hline
\end{tabular}
\caption{Fama-French `Factor Zoo' results ($\alpha=0.05$). Panel A: Rejection decisions for factor alphas. Panel B: Simulations confirming FWER control and superior power of Algorithm \ref{alg:compute_optimal_mu_K3_main}.}
\label{tab:ff_results}
\end{table}

\section{Conclusion}
\label{sec:conclusion}
This paper bridges the gap between the theoretical existence of optimal multiple testing procedures and their practical implementation, demonstrating that such optimality can be made operational, at least in small-$K$ regimes, through a synthesis of duality theory and rigorous analytical design. 
Building on the framework of \cite{RHPA22}, we provide the first constructive method to compute the most powerful test under exact FWER control for $K=3$.
Methodologically, we completely characterize the dual problem by deriving an explicit optimization over the ordered $p$-value simplex (Lemma~\ref{lemma:optimization_problem_k_3}) and decomposing the dual Lagrangian (Lemma~\ref{lemma:lagrangian_intermsof_mu0}).
Our central theoretical contribution, Theorem~\ref{thm:mu_optimality_combined}, establishes the necessary coordinate-wise optimality conditions for the dual multipliers.
This theoretical structure underpins Algorithm~\ref{alg:compute_optimal_mu_K3_main}, a coordinate-descent procedure that leverages efficient bracketing and bisection to achieve linear convergence to the optimal policy.


Simulations show that this optimality yields substantial power gains: our method consistently outperforms standard procedures such as Holm and Hommel in average power $\Pi_3$, often while also improving minimal power $\Pi_{\text{any}}$. 
Real-data studies corroborate these gains: in the BCG vaccine analysis, our method uniquely detects efficacy in low-incidence populations, and in the financial factor zoo, it recovers the Profitability (RMW) factor missed by conservative corrections.

Future work should prioritize extending these explicit optimality conditions to general $K$.
While directly tackling the full $K$-dimensional dual may be infeasible, low-dimensional projections or scalable approximations offer promising avenues to preserve FWER control while achieving near-optimal power.
Relaxing the independence assumption to accommodate positive or block dependence structures would also significantly broaden the method's applicability.
Furthermore, the linear representations of power and error suggest analogous dual formulations for alternative objectives, such as weighted or minimal power, and error metrics like $k$-FWER or FDR.
Finally, since practical implementation often requires estimating the alternative density $g$, developing robust formulations or empirical-Bayes approaches to account for estimation error remains a critical step for large-scale deployment.

\section*{Supplementary Material and Data Availability}
All simulation code and data, as well as those used for the real-data applications, are available as supplementary material and in the public repository
\url{https://github.com/pdubey96/most-powerful-mht-fwer-supplementary} for reproducibility.
A detailed description of the files are also provided in Appendix~\ref{app:experiments}.

\section*{Acknowledgments}
Dubey gratefully acknowledges partial financial support from the Stewart Topper Fellowship award at the Georgia Institute of Technology. 
Huo was partially sponsored by a subcontract of NSF grant 2229876, the A. Russell Chandler III Professorship at the Georgia Institute of Technology, and the NIH-sponsored Georgia Clinical \& Translational Science Alliance.
We used generative AI tools (ChatGPT and Google Gemini) only for language editing and code formatting support; all data, results, and mathematical derivations are the authors’ own work.

\bibliographystyle{plainnat}
\bibliography{cite}

\newpage

\appendix

\appendix

\section{Supplementary Background for the Problem Formulation}
\label{app:pf_supp}

This appendix collects background material and notation that is standard in the development of our framework but was omitted from the main text for concision. In particular, it records (i) additional permutation notation and its relationship to symmetry; (ii) intuition for the arrangement-increasing condition; (iii) the full LR-ordering theorem and corollary from \citet{RHPA22}; and (iv) notation used in the linear integral representations of power and error.

\subsection{Additional notation for permutations and likelihood ratios}
\label{app:perm_notation}

Let $S_K$ denote the set of all permutations of $[K]$. For $\sigma\in S_K$, we write $\sigma(\vec X)$ (resp.\ $\sigma(\vec h)$) for the vector obtained by permuting the coordinates of $\vec X=(\vec X_1,\dots,\vec X_K)$ (resp.\ $\vec h=(h_1,\dots,h_K)$) according to $\sigma$. We use $\sigma_{ij}\in S_K$ to denote the transposition that swaps only the $i$-th and $j$-th entries. For example,
\[
\sigma_{12}(\vec X)=(\vec X_2,\vec X_1,\vec X_3,\dots,\vec X_K),\qquad
\sigma_{12}(\vec h)=(h_2,h_1,h_3,\dots,h_K).
\]
When likelihood ratios are used, we also write $\vec\Lambda(\vec X)=(\Lambda_1(\vec X_1),\dots,\Lambda_K(\vec X_K))$ and note that
\[
\vec\Lambda(\sigma_{12}(\vec X))=(\Lambda_1(\vec X_2),\Lambda_2(\vec X_1),\Lambda_3(\vec X_3),\dots,\Lambda_K(\vec X_K)).
\]
These conventions are useful when stating symmetry and arrangement-based conditions.

\subsection{Interpreting the arrangement-increasing condition}
\label{app:arr_inc_intuition}

Assumption~\ref{as:assumption2} (arrangement-increasing) formalizes the idea that the joint likelihood should not decrease when stronger evidence is aligned with true alternatives. To illustrate, consider indices $i\neq j$ with $h_i=1$ (alternative true) and $h_j=0$ (null true). If the observed data yield $\Lambda_i(\vec X_i)<\Lambda_j(\vec X_j)$, then the evidence is ``misaligned'' (the null index looks more significant). In this case, the condition in \eqref{eq:assumption2_eq1} holds, and arrangement-increasing asserts that swapping the entries to correct the misalignment cannot decrease the likelihood:
\[
\mathcal L_{\vec h}(\vec X)\le \mathcal L_{\vec h}(\sigma_{ij}(\vec X)).
\]
This monotonicity is what ultimately justifies restricting attention to decision rules that reject hypotheses with larger likelihood ratios at least as often as those with smaller likelihood ratios.

\subsection{LR-ordered decision policies}
\label{app:lrorder_full}

For completeness, we record the main LR-ordering improvement result and its corollary from \citet{RHPA22}. These results are invoked in the main text to justify restricting the search in \eqref{eq:mtpintitial1} to the class of symmetric LR-ordered policies.

\begin{theorem}\label{app:thm:lrorder}
(\cite[Theorem 1]{RHPA22})
Assume:
\begin{enumerate}
    \item Assumptions~\ref{as:assumption1} and~\ref{as:assumption2} hold;
    \item The power criterion $\Pi$ is $\Pi_{\text{any}}$ or $\Pi_l$ for $1\le l\le K$, and the error criterion $\mathrm{Err}$ is either $\mathrm{FDR}$ or $\mathrm{FWER}$.
\end{enumerate}
Assume also that we are given a symmetric decision policy $\vec D$. Then there exists a symmetric policy $\vec E$ such that:
\begin{enumerate}
    \item $\vec E$ is LR-ordered;
    \item $\mathrm{Err}_{\vec h_l}(\vec E)\le \mathrm{Err}_{\vec h_l}(\vec D)$ for all $0\le l<K$;
    \item $\Pi(\vec E)\ge \Pi(\vec D)$.
\end{enumerate}
If $\vec D$ is already LR-ordered, one may take $\vec E=\vec D$.
\end{theorem}

\begin{corollary}\label{app:cor:lrorder}
(\cite[Corollary 1]{RHPA22})
Under the assumptions of Theorem~\ref{app:thm:lrorder}, if an optimal solution is found within the class of LR-ordered symmetric policies, it is globally optimal for \eqref{eq:mtpintitial1}.
\end{corollary}

\subsection{Notation behind the linear integral representations}
\label{app:linear_rep_notation}

This subsection unpacks the notation appearing in Theorem~\ref{thm:powererrorfn}.

\paragraph{Indexing configurations.}
For $1\le l\le K$, let $\binom{K}{l}$ denote the collection of all subsets $i\subseteq [K]$ with $|i|=l$. A subset $i\in\binom{K}{l}$ indexes the configuration in which the hypotheses in $i$ are false nulls and the hypotheses in the complement $i^c$ are true nulls. Under the p-value formulation with common alternative density $g(\cdot)$ (and independence), the joint density of $\vec u=(u_1,\dots,u_K)$ under configuration $i$ can be written as
\[
f_i(\vec u)=\prod_{k\in i} g(u_k)\prod_{k\in i^c} 1,
\]
where the factors ``$1$'' correspond to the $U(0,1)$ null density. (More general factorizations—e.g.\ under heterogeneous alternatives—fit the same indexing convention; see \citet{RHPA22} for the general setup.)

\paragraph{The role of $\bar i_{\min}$ in the FWER representation.}
In the linear representation of $\mathrm{FWER}_l$, the summation over configurations is restricted using the index
\[
\bar i_{\min} \;:=\; \min(i^c),
\]
i.e., the smallest index among the true null hypotheses under configuration $i$. The constraint $\bar i_{\min}=k$ selects, for each $k$, exactly those configurations in which $k$ is the ``first'' true null (in the ordered coordinate system on $Q$), which is the appropriate bookkeeping needed to express $\mathbb P_{\vec h_l}(V>0)$ as a linear functional of the (ordered) coordinate-wise decisions $D_k(\vec u)$; see \citet[Equations (4) and (7)]{RHPA22} for the derivation and Appendix~\ref{app:linear_rep_notation} for this notation.

The above indexing is what allows both $\Pi_l(\vec D)$ and $\mathrm{FWER}_l(\vec D)$ to be written in the linear forms used in \eqref{eq:power_general}--\eqref{eq:error_general}, which in turn yields the canonical primal program \eqref{eq:mtpopt1} and its Lagrangian dual in Section~\ref{sec:lagdualprob}.

\section{Structural Proofs for Characterising the Lagrangian Minimiser}
The goal of this appendix is to develop the technical ingredients needed to characterise the minimiser of the Lagrangian $L(\vec{D}^\mu,\mu)$ in \eqref{eq:lagrangian5_general} and, ultimately, to prove the optimality conditions stated in Theorem~\ref{thm:mu_optimality_combined}. 
We proceed in several steps.
We begin in Subsection~\ref{app:proof_fwer_implies_fdr} by proving Lemma~\ref{lemma:fwer_implies_fdr}, which shows that control of the family-wise error rate (FWER) at level $\alpha$ automatically implies control of the false discovery rate (FDR) at the same level. This justifies focusing on exact FWER control in the optimisation problem. 
Subsection~\ref{app:proof-g-monotone} then establishes Lemma~\ref{lem:gdot}, showing that under our likelihood-ratio–based construction the alternative $p$-value density $g$ coincides with the threshold function $\lambda$ and is non-increasing. 
This structural property of $g$ is used repeatedly in the later analysis.
Subsection~\ref{app:dual_problem} proves Lemma~\ref{lem:dual_problem}, which introduces the Lagrangian formulation of the multiple testing problem and derives the dual problem in terms of the Lagrange multipliers $\mu$. 
Building on this, Subsection~\ref{app:proof_optimization_problem_k_3} proves Lemma~\ref{lemma:optimization_problem_k_3} by specialising the general formulation to the case $K=3$, expressing both the power objective and the FWER constraints explicitly in terms of the ordered $p$-values and the density $g$. 
Subsection~\ref{app:opt_decision_mu} then proves Lemma~\ref{lem:opt_decision_mu}, providing an explicit description of the optimal decision policy $\vec{D}^\mu$ for a fixed $\mu$ via the indicator functions $\alpha_i^{\mu}$ and $\beta_i^{\mu}$ built from partial sums of the quantities $R_i(\mu,\vec{u})$. 
Finally, Subsection~\ref{app:lagrangian_intermsof_mu0} proves Lemma~\ref{lemma:lagrangian_intermsof_mu0}, rewriting $L(\vec{D}^\mu,\mu)$ in a form that cleanly separates its dependence on $(\mu_0,\mu_1,\mu_2)$ and introduces the auxiliary functions $f_1(\mu,\vec{u})$ and $f_2(\mu,\vec{u})$ used in our optimality analysis. 

Taken together, these lemmas prepare the ground for the next section, Section~\ref{app:mu_optimality}, where we combine the above ingredients to minimise $L(\vec{D}^\mu,\mu)$ with respect to $\mu$ and complete the proof of Theorem~\ref{thm:mu_optimality_combined}.

\subsection{FWER Control Implies FDR Control}
\label{app:proof_fwer_implies_fdr}
\begin{lemma}
\label{lemma:fwer_implies_fdr}
Let $\vec{D} \in \{0,1\}^K$ be any decision rule and consider any configuration $\vec{h}_l$ with $0 \leq l < K$. Then, the false discovery rate (FDR) and family-wise error rate (FWER), as defined in \eqref{eq:fdr_final} and \eqref{eq:fwer_final}, satisfy:
\begin{equation*}
\label{eq:fwer_fdr_relationship}
\mathrm{FDR}_l(\vec{D}) \leq \mathrm{FWER}_l(\vec{D}).
\end{equation*}
In particular, if a multiple testing procedure controls $\mathrm{FWER}_l(\vec{D}) \leq \alpha$, then it also controls $\mathrm{FDR}_l(\vec{D}) \leq \alpha$.
\end{lemma}
\begin{proof}
From Table \ref{tab:outcome}, we summarize the results of testing $K$ multiple hypotheses as follows: $V$ denotes the total number of false discoveries, $S$ represents the number of true discoveries, and the total number of discoveries is given by $R = V + S$. 
Consequently, we have that if $R > 0$:  
\begin{equation}\label{eq:vbyrlessthan1}
\frac{V}{R} \le 1 \quad \text{with probability } 1.    
\end{equation}

We begin by applying the law of total expectation, splitting the domain of the expectation based on the event $\{V = 0\}$ and its complement $\{V > 0\}$.
For any $l$ with $0 \leq l < K$, we have that
\begin{equation} \label{eq:fdr_total_expectation}
\begin{aligned}
\mathrm{FDR}_l(\vec{D}) 
& \stackrel{\mbox{\eqref{eq:fdr_final}}}{=} \mathbb{E}_{\vec{h}_l}\left[\frac{V}{R} ; R > 0\right] \\
&= \mathbb{E}_{\vec{h}_l}\left[ \frac{V}{R} \cdot \mathbf{1}_{\{R > 0\}} \right] \\
&= \mathbb{E}_{\vec{h}_l}\left[ \frac{V}{R} \cdot \mathbf{1}_{\{R > 0\}} \mid V = 0 \right] \cdot \mathbb{P}_{\vec{h}_l}(V = 0) + \mathbb{E}_{\vec{h}_l}\left[ \frac{V}{R} \cdot \mathbf{1}_{\{R > 0\}} \mid V > 0 \right] \cdot \mathbb{P}_{\vec{h}_l}(V > 0).
\end{aligned}
\end{equation}

We now analyze each case separately.
\begin{enumerate}
\item If $V = 0$, then the quantity $\frac{V}{R} \cdot \mathbf{1}_{\{R > 0\}}=0$.
Hence:
\begin{equation} \label{eq:fdr_case1}
\mathbb{E}_{\vec{h}_l}\left[ \frac{V}{R} \cdot \mathbf{1}_{\{R > 0\}} \mid V = 0 \right] = 0.
\end{equation}

\item If $V > 0$, then necessarily $R = V + S > 0$, so the event $\{R > 0\}$ holds with probability 1 under this condition. 
Therefore:
\begin{equation} \label{eq:fdr_case2}
\begin{aligned}
\mathbb{E}_{\vec{h}_l}\left[ \frac{V}{R} \cdot \mathbf{1}_{\{R > 0\}} \mid V > 0 \right] 
&= \mathbb{E}_{\vec{h}_l}\left[ \frac{V}{R} \mid V > 0 \right] \\
& \stackrel{\mbox{\eqref{eq:vbyrlessthan1}}}{\le} \mathbb{E}_{\vec{h}_l}\left[ 1 \mid V > 0 \right] = 1.
\end{aligned}
\end{equation}
\end{enumerate}

Substituting \eqref{eq:fdr_case1} and \eqref{eq:fdr_case2} into the total expectation decomposition \eqref{eq:fdr_total_expectation}, we obtain:
\begin{equation}\label{eq:fdr<fwer_proved}
\begin{aligned}
\mathrm{FDR}_l(\vec{D}) 
& = \mathbb{E}_{\vec{h}_l}\left[ \frac{V}{R} \cdot \mathbf{1}_{\{R > 0\}} \mid V = 0 \right] \cdot \mathbb{P}_{\vec{h}_l}(V = 0) + \mathbb{E}_{\vec{h}_l}\left[ \frac{V}{R} \cdot \mathbf{1}_{\{R > 0\}} \mid V > 0 \right] \cdot \mathbb{P}_{\vec{h}_l}(V > 0) \\
&= 0 \cdot \mathbb{P}_{\vec{h}_l}(V = 0) 
+ \mathbb{E}_{\vec{h}_l}\left[ \frac{V}{R} \mid V > 0 \right] \cdot \mathbb{P}_{\vec{h}_l}(V > 0) \\
&\le 1 \cdot \mathbb{P}_{\vec{h}_l}(V > 0) \stackrel{\mbox{\eqref{eq:fwer_final}}}{=} \mathrm{FWER}_l(\vec{D}).
\end{aligned}
\end{equation}
Therefore, using \eqref{eq:fdr<fwer_proved} we have that if $\mathrm{FWER}_l(\vec{D}) \leq \alpha$ then it follows that $\mathrm{FDR}_l(\vec{D}) \leq \alpha$.
\end{proof}

\subsection{Proof of Lemma \ref{lem:gdot}}
\label{app:proof-g-monotone}

\begin{proof}[Proof of Lemma \ref{lem:gdot}]
Let $g(u)$ be the density of the random variable $u$ under $H_a$. For any $0 \leq \tau \leq 1$ we have,
\begin{equation*}
    \begin{aligned}
        \int_0^\tau g(u)du &= \mathbb{P}_{H_a}\{u \leq \tau\} \\ 
        &= \int \mathbb{I}\left\{\frac{f_1(x)}{f_0(x)}\geq \lambda(\tau)\right\}\frac{f_1(x)}{f_0(x)} f_0(x) dx
    \end{aligned}
\end{equation*}

We assume that $g(u)$ is continuous on $(0,1)$ and differentiable almost everywhere. 
From the mean value theorem for integrals, which states that for a continuous function $g(u)$ on an interval $[a,b]$, there exists some $u^* \in (a,b)$ such that  
\begin{equation*}
\int_{a}^{b} g(u) du = g(u^*) (b-a).
\end{equation*}
Applying this to the interval $[\tau, \tau + \Delta]$, we obtain  
\begin{equation}\label{eq:gu_mvt}
\int_{\tau}^{\tau+\Delta} g(u) du = g(u^*) \Delta, \quad \text{for some } u^* \in (\tau, \tau + \Delta).
\end{equation}
Since $g(u)$ is continuous at $\tau$, it follows that taking the limit $\Delta \to 0$, we have  
\begin{equation}\label{eq:gu_gt}
g(u^*) \to g(\tau).
\end{equation}

On the other hand, we have,
\begin{eqnarray}
\int_\tau^{\tau+\Delta} g(u)du &=& \int_0^{\tau+\Delta} g(u)du- \int_0^{\tau} g(u)du \nonumber \\
&=& \int \mathbb{I}\left\{\frac{f_1(x)}{f_0(x)}\geq \lambda(\tau+\Delta)\right\}\frac{f_1(x)}{f_0(x)} f_0(x) dx - 
\int \mathbb{I}\left\{\frac{f_1(x)}{f_0(x)}\geq \lambda(\tau)\right\}\frac{f_1(x)}{f_0(x)} f_0(x) dx \nonumber \\
&=& \int \mathbb{I}\left\{\lambda(\tau) \geq \frac{f_1(x)}{f_0(x)} \geq \lambda(\tau+\Delta)\right\}\frac{f_1(x)}{f_0(x)} f_0(x) dx \label{eq:gu_expand}
\end{eqnarray}
From \eqref{eq:pvalue} we have,
\begin{eqnarray}
&& \int \mathbb{I}\left\{\lambda(\tau) \geq \frac{f_1(x)}{f_0(x)} \geq \lambda(\tau+\Delta)\right\}f_0(x) dx \nonumber \\
&=& \int \mathbb{I}\left\{\frac{f_1(x)}{f_0(x)} \geq \lambda(\tau+\Delta)\right\}f_0(x) dx - \int \mathbb{I}\left\{\frac{f_1(x)}{f_0(x)} \geq \lambda(\tau)\right\}f_0(x) dx  \nonumber \\
&=& (\tau+\Delta)-\tau  \nonumber \\
&=& \Delta \label{eq:delta}
\end{eqnarray}
Combining \eqref{eq:gu_mvt}, \eqref{eq:gu_expand} and \eqref{eq:delta} we have 
\begin{equation*}
    \Delta \cdot \lambda(\tau+\Delta) \leq g(u^*) \cdot \Delta \leq \Delta \cdot \lambda(\tau).
\end{equation*}
The above can be simplified as 
\begin{equation}\label{eq:gu_sandwich}
    \lambda(\tau+\Delta) \leq g(u^*) \leq  \lambda(\tau). 
\end{equation}
Applying the limit $\Delta \to 0$ to \eqref{eq:gu_sandwich} and using \eqref{eq:gu_gt} we have that,
\begin{equation*}
    g(\tau)=\lambda(\tau).
\end{equation*}
Since $\lambda(\tau)$ is non-increasing in $\tau$, it follows that $g(\tau)$ is also non-increasing. 
Thus, for all $u \in (0,1)$, we conclude 
\begin{equation*}
   g(u)=\lambda(u).
\end{equation*}
\end{proof}

\subsection{Proof of Lemma \ref{lem:dual_problem}}
\label{app:dual_problem}
\begin{proof}
The Lagrangian associated with this problem \eqref{eq:mtpopt1} is given by \eqref{eq:lagrangian1} as:
\begin{equation*}
\begin{aligned}
L(\vec{D}, \mu) & =\int_{Q}\left(\sum_{i=1}^{K} a_{i}(\vec{u}) D_{i}(\vec{u})\right) d \vec{u}+\sum_{l=0}^{K-1} \mu_{l} \times\left\{\alpha-\int_{Q}(\sum_{i=1}^{K} b_{l, i}(\vec{u}) D_{i}(\vec{u})) d \vec{u}\right\} \\
& =\sum_{l=0}^{K-1} \mu_{l} \times \alpha+\int_{Q}\left\{\sum_{i=1}^{K} D_{i}(\vec{u})\left(a_{i}(\vec{u})-\sum_{l=0}^{K-1} \mu_{l} \times b_{l, i}(\vec{u})\right)\right\} d \vec{u} \\
& =\sum_{l=0}^{K-1} \mu_{l} \times \alpha+ \int_{Q} \sum_{i=1}^{K} D_{i}(\vec{u}) R_{i}(\mu,\vec{u}) d \vec{u} ,
\end{aligned}
\end{equation*} 
where, $R_{i}(\mu,\vec{u})$ is defined as in \eqref{eq:rimu1} and $\mu=(\mu_{0}, \ldots, \mu_{K-1})\geq0$ denote the Lagrange multipliers for the maximization problem \eqref{eq:mtpopt1}. 
In \eqref{eq:lagrangian1}, the decision policies $\vec{D} = \{D_1,\cdots,D_K\}$ represent the primal variables, while $\mu=(\mu_{0}, \ldots, \mu_{K-1})$ stand as the dual variables. 

Hence, we can formulate the dual problem of \eqref{eq:mtpopt1} as follows:
\begin{equation} \label{eq:dual1}
\begin{aligned}
  \min_{\mu\geq 0} \max_{\vec{D}} L(\vec{D}, \mu).
\end{aligned}
\end{equation} 
Looking at equation \eqref{eq:dual1}, our first step involves maximizing $L(\vec{D}, \mu)$. 
We can break this down into two parts, with our focus primarily on maximizing the second term in equation \eqref{eq:maxlagrangian1} below:
\begin{equation}\label{eq:maxlagrangian1}
\begin{aligned}
& \max_{\vec{D}} L(\vec{D}, \mu) \\
& \stackrel{\mbox{\eqref{eq:lagrangian1}}}{=} \max_{\vec{D}} \sum_{l=0}^{K-1} \mu_{l} \times \alpha
+ \int_{Q} \sum_{i=1}^{K} D_{i}(\vec{u}) R_{i}(\mu,\vec{u}) d \vec{u}
\end{aligned}
\end{equation} 

Note that the first term above is not a function of $\vec{D}$.
So, only the second term needs to be maximized. 
Recall, we only need to consider 
$$
\vec{D} \in \mathcal{D}=\{(0,0,\cdots,0),(1,0,\cdots,0),(1,1,\cdots,0),\cdots,(1,1,\cdots,1)\},
$$
which means that the sequence $\vec{D}$ starts with $1$ and decrease to $0$ and will never increase. 
The maximum value of the objective in \eqref{eq:maxlagrangian1} will be either $0$ or corresponds to an $l$, $1\leq l\leq K$, such that $\sum_{i=1}^{l} R_{i}(\mu,\vec{u})$ is the largest among all the similar partial sums. 
In summary, we can have the following: 
\begin{equation}\label{eq:maxdi}
\underset{\vec{D}}{\operatorname{max}} \int_{Q} \sum_{i=1}^{K} D_{i}(\vec{u}) R_{i}(\mu,\vec{u}) d \vec{u} 
= \int_{Q} \max \left\{\underset{1\leq l \leq K}{\operatorname{max}}\sum_{i=1}^{l} R_{i}(\mu,\vec{u}),0 \right\} d \vec{u}.
\end{equation}

The following decision policy accomplishes this. 
Let $l^\ast$ denote the value of $l$, $1 \le l \le K$, such that, for $\mu \ge 0$ and $\vec{u} \in Q$, we have 
\begin{equation*}
    \sum_{i=1}^{l^\ast} R_{i}(\mu,\vec{u}) = \underset{1\leq l \leq K}{\operatorname{max}}\sum_{i=1}^{l} R_{i}(\mu,\vec{u}),
\end{equation*}
and $l^\ast = 0$ if the right-hand side above is less than $0$. 
We then have 
\begin{equation*}
    D_{l}^{\mu}(\vec{u}) = \left\{\begin{array}{ll}
1, & \mbox{ if } l \le l^\ast \\
0, & \mbox{ otherwise.}
\end{array} \right.
\end{equation*}
Thus, by obtaining $\vec{D}^{\mu}$ as derived in the above equation \eqref{eq:optdecision1}, we solve the binary program \eqref{eq:maxlagrangian1} optimally. 
We then plugin the value of $\vec{D}$ as $\vec{D}^{\mu}$ as obtained in \eqref{eq:optdecision1} in equation \eqref{eq:maxlagrangian1}. Thus, we have:
\begin{equation*}
\begin{aligned}
    \max_{\vec{D}} L(\vec{D}, \mu) & \stackrel{\mbox{\eqref{eq:maxlagrangian1}}}{=} \max_{\vec{D}} \sum_{l=0}^{K-1} \mu_{l} \times \alpha
+ \int_{Q} \sum_{i=1}^{K} D_{i}(\vec{u}) R_{i}(\mu,\vec{u}) d \vec{u} \\ 
    &=\sum_{l=0}^{K-1} \mu_l \alpha + \underset{\vec{D}}{\operatorname{max}} \int_{Q} \sum_{i=1}^{K} D_{i}(\vec{u}) R_{i}(\mu,\vec{u}) d \vec{u} \\
    &\stackrel{\mbox{\eqref{eq:maxdi}}}{=}  \sum_{l=0}^{K-1} \mu_l \alpha + \int_Q \max \left\{ \max_{1 \leq l \leq K} \sum_{i=1}^{l} R_i(\mu, u), 0 \right\} du \\
    &\stackrel{\mbox{\eqref{eq:l*}}}{=}  \sum_{l=0}^{K-1} \mu_l \alpha + \int_Q \sum_{i=1}^{l^*(\mu, u)} R_i(\mu, u) du \\
    &= L(\vec{D}^{\mu}, \mu)
\end{aligned}
\end{equation*}

Consequently, the dual problem \eqref{eq:dual1} transforms into:
\begin{equation*}
\min_{\mu\geq 0} \max_{\vec{D}} L(\vec{D}, \mu) = \min_{\mu\geq 0}  L(\vec{D}^{\mu}, \mu),
\end{equation*}
This completes the proof.
\end{proof}

\subsection{Proof of Lemma \ref{lemma:optimization_problem_k_3}}
\label{app:proof_optimization_problem_k_3}

\begin{proof}
First, the objective function is given by \eqref{eq:pillinear1} as:
\begin{equation*}\label{eq:objective_pil3}
\begin{aligned}
    \Pi_3(\vec{D}) &= (3-1)!(3-3)! \int_{Q} \sum_{i \in {\binom{3}{3}}} f_{i}(\vec{u}) \sum_{k \in i} D_{k}(\vec{u}) d\vec{u} \\
    &= 2 \int_{Q} \sum_{i \in \{\{1,2,3\}\}} f_{i}(\vec{u}) \sum_{k \in i} D_{k}(\vec{u}) d\vec{u} \\
    &= 2 \int_{Q} f_{\{1,1,1\}}(\vec{u}) \sum_{k \in \{1,2,3\}} D_{k}(\vec{u}) d\vec{u} \\
    &\stackrel{\mbox{\eqref{eq:true_joint_dist}}}{=} 2 \int_{Q} f_{1,1}(u_1) f_{2,1}(u_2) f_{3,1}(u_3) \sum_{k \in \{1,2,3\}} D_{k}(\vec{u}) d\vec{u} \\
    &= 2 \int_{Q} g(u_1) g(u_2) g(u_3) (D_{1}(\vec{u}) + D_{2}(\vec{u}) + D_{3}(\vec{u}))d\vec{u} 
\end{aligned}
\end{equation*}
Second, the constraints are given by \eqref{eq:fwer_linear} as:
\begin{equation}\label{eq:constraints_fwer_0}
\begin{aligned}
    FWER_0(\vec{D}) &= 0!(3 - 0)! \int_Q \sum_{k=1}^3 D_k(\vec{u}) \sum_{\substack{i \in \binom{3}{0} \\ \Bar{i}_{min}=k}} f_i(\vec{u}) du \\
    &= 6 \int_Q \sum_{k=1}^3 D_k(\vec{u}) \sum_{\substack{i \in \{\phi\} \\ \Bar{i}_{min}=k}} f_i(\vec{u}) du \\
    &= 6 \int_Q \Bigg(D_1(\vec{u}) \sum_{\substack{i \in \{\phi\} \\ \Bar{i}_{min}=1}} f_i(\vec{u}) + D_2(\vec{u}) \sum_{\substack{i \in \{\phi\} \\ \Bar{i}_{min}=2}} f_i(\vec{u}) + D_3(\vec{u}) \sum_{\substack{i \in \{\phi\} \\ \Bar{i}_{min}=3}} f_i(\vec{u}) \Bigg) du
\end{aligned}
\end{equation}

\begin{equation}\label{eq:constraints_fwer_1}
\begin{aligned}
    FWER_1(\vec{D}) &= 1!(3 - 1)! \int_Q \sum_{k=1}^3 D_k(\vec{u}) \sum_{\substack{i \in \binom{3}{1} \\ \Bar{i}_{min}=k}} f_i(\vec{u}) du \\
    &= 2 \int_Q \sum_{k=1}^3 D_k(\vec{u}) \sum_{\substack{i \in \{\{1\},\{2\},\{3\}\} \\ \Bar{i}_{min}=k}} f_i(\vec{u}) du \\
    &= 2 \int_Q \Bigg( D_1(\vec{u}) \sum_{\substack{i \in \{\{1\},\{2\},\{3\}\} \\ \Bar{i}_{min}=1}} f_i(\vec{u}) + D_2(\vec{u}) \sum_{\substack{i \in \{\{1\},\{2\},\{3\}\} \\ \Bar{i}_{min}=2}} f_i(\vec{u}) + D_3(\vec{u}) \sum_{\substack{i \in \{\{1\},\{2\},\{3\}\} \\ \Bar{i}_{min}=3}} f_i(\vec{u}) \Bigg) du
\end{aligned}
\end{equation}

\begin{equation}\label{eq:constraints_fwer_2}
\begin{aligned}
    FWER_2(\vec{D}) &= 2!(3 - 2)! \int_Q \sum_{k=1}^3 D_k(\vec{u}) \sum_{\substack{i \in \binom{3}{2} \\ \Bar{i}_{min}=k}} f_i(\vec{u}) du \\
    &= 2 \int_Q \sum_{k=1}^3 D_k(\vec{u}) \sum_{\substack{i \in \{\{1,2\},\{1,3\},\{2,3\}\} \\ \Bar{i}_{min}=k}} f_i(\vec{u}) du \\
    &= 2 \int_Q \Bigg(D_1(\vec{u}) \sum_{\substack{i \in \{\{1,2\},\{1,3\},\{2,3\}\} \\ \Bar{i}_{min}=1}} f_i(\vec{u}) &+ D_2(\vec{u}) \sum_{\substack{i \in \{\{1,2\},\{1,3\},\{2,3\}\} \\ \Bar{i}_{min}=2}} f_i(\vec{u}) \\ &+ D_3(\vec{u}) \sum_{\substack{i \in \{\{1,2\},\{1,3\},\{2,3\}\} \\ \Bar{i}_{min}=3}} f_i(\vec{u}) \Bigg) du
\end{aligned}
\end{equation}

Now we have that, 
\begin{equation}\label{eq:i_min_bar}
    \Bar{i}_{min}=\left\{\begin{array}{ll}
1, & \mbox{ if } i=\phi \\
2, & \mbox{ if } i=\{1\} \\
1, & \mbox{ if } i=\{2\} \\
1, & \mbox{ if } i=\{3\} \\
3, & \mbox{ if } i=\{1,2\} \\
1, & \mbox{ if } i=\{2,3\} \\
2, & \mbox{ if } i=\{1,3\} \\
\phi, & \mbox{ if } i=\{1,2,3\} \\
\end{array} \right.
\end{equation}
Thus using \eqref{eq:i_min_bar} in  \eqref{eq:constraints_fwer_0}, \eqref{eq:constraints_fwer_1} and \eqref{eq:constraints_fwer_2} we have,

\begin{equation*}\label{eq:constraints_fwer_0_final}
\begin{aligned}
    FWER_0(\vec{D}) &= 6 \int_Q D_1(\vec{u}) \sum_{i \in \{\phi\}} f_i(\vec{u})  du \\ 
    &= 6 \int_Q D_1(\vec{u}) f_{\{0,0,0\}}  du \\ 
    &\stackrel{\mbox{\eqref{eq:true_joint_dist}}}{=}  6 \int_Q D_1(\vec{u}) du 
\end{aligned}
\end{equation*}

\begin{equation*}\label{eq:constraints_fwer_1_final}
\begin{aligned}
    FWER_1(\vec{D}) &= 2 \int_Q \Big( D_1(\vec{u}) \sum_{i \in \{\{2\},\{3\}\}} f_i(\vec{u}) + D_2(\vec{u}) \sum_{i \in \{\{1\}\}} f_i(\vec{u}) \Big) du \\
    &= 2 \int_Q D_1(\vec{u}) (f_{\{0,1,0\}}(\vec{u})+f_{\{0,0,1\}}(\vec{u}))  + D_2(\vec{u}) f_{\{1,0,0\}}(\vec{u}) du \\ 
    &\stackrel{\mbox{\eqref{eq:true_joint_dist}}}{=}  2 \int_Q D_1(\vec{u}) (g(u_2)+g(u_3))  + D_2(\vec{u}) g(u_1) du \\ 
\end{aligned}
\end{equation*}

\begin{equation*}\label{eq:constraints_fwer_2_final}
\begin{aligned}
    FWER_2(\vec{D}) &= 2 \int_Q \Big( D_1(\vec{u}) \sum_{i \in \{\{2,3\}\}} f_i(\vec{u}) + D_2(\vec{u}) \sum_{i \in \{\{1,3\}\}} f_i(\vec{u}) + D_3(\vec{u}) \sum_{i \in \{\{1,2\}\}} f_i(\vec{u}) \Big) du \\
    &= 2 \int_Q D_1(\vec{u}) f_{\{0,1,1\}}(\vec{u})+ D_2(\vec{u})f_{\{1,0,1\}}(\vec{u}))  + D_3(\vec{u}) f_{\{1,1,0\}}(\vec{u}) du \\ 
    &\stackrel{\mbox{\eqref{eq:true_joint_dist}}}{=}  2 \int_Q D_1(\vec{u})g(u_2)g(u_3))+D_2(\vec{u})g(u_1)g(u_3))+D_3(\vec{u})g(u_1)g(u_2)) du \\ 
\end{aligned}
\end{equation*}

Hence the optimization problem \eqref{eq:mtpopt1} becomes:
\begin{equation*}
\begin{aligned}
\max _{\vec{D}: Q \rightarrow[0,1]^{3}} & 2 \int_{Q}(D_{1}(\vec{u})+D_{2}(\vec{u})+D_{3}(\vec{u})) g(u_1) g(u_2) g(u_3) d \vec{u} \\
\text { s.t. } & FWER_{0}(\vec{D})=6 \int_{Q} D_{1}(\vec{u}) d \vec{u} \leq \alpha \text {, } \\
& FWER_{1}(\vec{D})=2 \int_{Q}\left[D_{1}(\vec{u})(g(u_2)+g(u_3))+D_{2}(\vec{u}) g(u_1) \right] d \vec{u} \leq \alpha \text {, } \\
& FWER_{2}(\vec{D})=2 \int_{Q}\left[D_{1}(\vec{u}) g(u_2) g(u_3)+ D_{2}(\vec{u}) g(u_1) g(u_3) + D_{3}(\vec{u}) g(u_1) g(u_2)\right] d \vec{u} \leq \alpha, \\
& 0 \leq D_{3}(\vec{u}) \leq D_{2}(\vec{u}) \leq D_{1}(\vec{u}) \leq 1, \ \forall \vec{u} \in Q.
\end{aligned}
\end{equation*} 
This completes the proof.
\end{proof}

\subsection{Proof of Lemma \ref{lem:opt_decision_mu}}
\label{app:opt_decision_mu}
\begin{proof}
We will solve the dual of \eqref{eq:objconst}. 
Using the notations as defined in Section \ref{sec:lagdualprob} and using $\mu$=$(\mu_0,\mu_1,\mu_2)$ to denote the Lagrange multipliers we derive $R_{i}(\mu,\vec{u})$ as in \eqref{eq:rimu1}:
\begin{equation}
    R_{i}(\mu,\vec{u})=a_{i}(\vec{u})-\sum_{l=0}^{2} \mu_{l} \times b_{l, i}(\vec{u}), \ 1\le i \le 3.
\end{equation}
Here $a_{i}(\vec{u})$ and $b_{l, i}(\vec{u})$ are obtained by comparing \eqref{eq:objconst} with \eqref{eq:power_general} and \eqref{eq:error_general}. Thus, we have,
\begin{equation}
\begin{array}{lll}
a_1(\vec{u})= & a_2(\vec{u})= &a_3(\vec{u})= 2g(u_1) g(u_2) g(u_3), \\
b_{0, 1}(\vec{u})= 6, & b_{0, 2}(\vec{u})=0, & b_{0, 3}(\vec{u})=0, \\
b_{1, 1}(\vec{u})=2(g(u_2)+g(u_3)), & b_{1, 2}(\vec{u})=2g(u_1), & b_{1, 3}(\vec{u})=0, \\
b_{2, 1}(\vec{u})=2g(u_2)g(u_3), & b_{2, 2}(\vec{u})=2g(u_1)g(u_3), & b_{2, 3}(\vec{u})=2g(u_1)g(u_2).
\end{array}
\end{equation}
Hence $R_{i}(\mu,\vec{u})$'s are given by:
\begin{equation*}
\begin{aligned}
R_{1}(\mu,\vec{u})= & 2g(u_1) g(u_2) g(u_3)- 6\mu_{0}- &2\mu_{1}(g(u_2)+g(u_3)) - & 2\mu_{2}g(u_2)g(u_3)\\
R_{2}(\mu,\vec{u})= & 2g(u_1) g(u_2) g(u_3) & -2\mu_{1} \ g(u_1) - & 2\mu_{2} \ g(u_1) g(u_3) \\
R_{3}(\mu,\vec{u})= & 2g(u_1) g(u_2) g(u_3) & -& 2\mu_2\ g(u_1)g(u_2)
\end{aligned}
\end{equation*} 
Recall that we only consider the domain $Q=\left\{\vec{u}: 0 \leq u_{1} \leq u_{2} \leq \cdots \leq u_{K} \leq 1\right\}$.
Here, we have $K=3$. 
Per Lemma \ref{lem:gdot}, we have $g(u_{1}) \geq g(u_{2}) \geq g(u_{3})$.

From Section \ref{sec:lagdualprob} we have the dual problem as \eqref{eq:ld_mu} :
\begin{equation} 
\label{eq:dual2}
\begin{aligned}
\min_{\mu\geq 0} L(\vec{D}^\mu, \mu).
\end{aligned}
\end{equation} where $L(\vec{D}, \mu)$ is the Lagrangian given by using \eqref{eq:lagrangian1},

\begin{equation}
\begin{aligned}
 L(\vec{D}^\mu, \mu) = \sum_{l=0}^{2} \mu_{l} \times \alpha+ \int_{Q} \sum_{i=1}^{3} D_{i}^{\mu}(\vec{u}) R_{i}(\mu,\vec{u})d \vec{u}
\end{aligned}
\end{equation} 
We then compute the \textit{optimal} decision policy $D$ using \eqref{eq:maxdi}. 

\begin{equation}
\max _{\vec{D}(\vec{u}) \in \mathcal{D}} \sum_{k=1}^{3} D_{k}(\vec{u}) R_{k}(\mu,\vec{u})=\max \left\{\max _{l=1, \ldots, 3} \sum_{k=1}^{l} R_{l}(\mu,\vec{u}), 0\right\}
\end{equation} which is given by using \eqref{eq:optdecision1} as,
\begin{equation}\label{eq:dimus1}
\begin{aligned}
D_{1}^{\mu}(\vec{u})= & \mathbbm{1}\left\{R_{1}(\mu,\vec{u})>0 \cup R_{1}(\mu,\vec{u})+R_{2}(\mu,\vec{u})>0 \cup R_{1}(\mu,\vec{u})+R_{2}(\mu,\vec{u})+R_{3}(\mu,\vec{u})>0 \right\} \\
D_{2}^{\mu}(\vec{u})= & D_{1}^{\mu}(\vec{u}) \times \mathbbm{1}\left\{R_{2}(\mu,\vec{u})>0 \cup R_{2}(\mu,\vec{u})+R_{3}(\mu,\vec{u})>0 \right\} \\
D_{3}^{\mu}(\vec{u})= & D_{2}^{\mu}(\vec{u}) \times \mathbbm{1}\left\{R_{3}(\mu,\vec{u})>0 \right\}.
\end{aligned}
\end{equation} 
Denote 
\begin{equation}
\begin{aligned}
\alpha_1^{\mu}(\vec{u}) &= \mathbbm{1}\left\{R_{1}(\mu,\vec{u}) > 0 \cup R_{1}(\mu,\vec{u})+R_{2}(\mu,\vec{u}) > 0 \cup R_{1}(\mu,\vec{u})+R_{2}(\mu,\vec{u})+R_{3}(\mu,\vec{u}) > 0 \right\}, \\
\alpha_2^{\mu}(\vec{u}) &= \mathbbm{1}\left\{R_{2}(\mu,\vec{u}) > 0 \cup R_{2}(\mu,\vec{u})+R_{3}(\mu,\vec{u}) > 0 \right\}, \\
\alpha_3^{\mu}(\vec{u}) &= \mathbbm{1}\left\{R_{3}(\mu,\vec{u}) > 0 \right\}.
\end{aligned}
\end{equation}

We have 

\begin{equation}
\begin{aligned}
D_{1}^{\mu}(\vec{u})= & \alpha_1^{\mu}(\vec{u}), \\
D_{2}^{\mu}(\vec{u})= & \alpha_{1}^{\mu}(\vec{u}) \alpha_2^{\mu}(\vec{u}), \\
D_{3}^{\mu}(\vec{u})= & \alpha_{1}^{\mu}(\vec{u}) \alpha_{2}^{\mu}(\vec{u}) \alpha_3^{\mu}(\vec{u}).
\end{aligned}
\end{equation} 

We have from \eqref{eq:alpha_mu_u}:
\begin{equation}\label{eq:1-alpha_mu_u}
\begin{aligned}
1-\alpha_1^{\mu}(\vec{u}) &= \mathbbm{1}\left\{R_{1}(\mu,\vec{u}) \le 0 \cap R_{1}(\mu,\vec{u})+R_{2}(\mu,\vec{u})\le0 \cap R_{1}(\mu,\vec{u})+R_{2}(\mu,\vec{u})+R_{3}(\mu,\vec{u})\le0 \right\}, \\
1-\alpha_2^{\mu}(\vec{u}) &= \mathbbm{1}\left\{R_{2}(\mu,\vec{u})\le0 \cap R_{2}(\mu,\vec{u})+R_{3}(\mu,\vec{u})\le0 \right\}, \\
1-\alpha_3^{\mu}(\vec{u}) &= \mathbbm{1}\left\{R_{3}(\mu,\vec{u})\le0 \right\}.
\end{aligned}
\end{equation}
This completes the proof.

Additionally, we denote:
{\footnotesize
\begin{equation}\label{eq:beta_mu_u}
\begin{aligned}
\beta_1^{\mu}(\vec{u})
&=1-\alpha_1^{\mu}(\vec{u}) = \mathbbm{1}\biggl\{\mu,\vec{u}:\begin{pmatrix}
R_{1}(\mu,\vec{u}) \\ 
R_{1}(\mu,\vec{u})+R_{2}(\mu,\vec{u}) \\
R_{1}(\mu,\vec{u})+R_{2}(\mu,\vec{u})+R_{3}(\mu,\vec{u})
\end{pmatrix} \le 0 \biggr\} \\
&=\mathbbm{1}\left\{\begin{array}{c}
2g(u_1) g(u_2) g(u_3) \le 6\mu_0 + 2(g(u_2)+g(u_3))\mu_1 + 2g(u_2)g(u_3)\mu_2\\
4g(u_1) g(u_2) g(u_3) \le 6\mu_0 + 2(g(u_1)+g(u_2)+g(u_3))\mu_1 + 2(g(u_1)g(u_3)+g(u_2)g(u_3))\mu_2\\
6g(u_1) g(u_2) g(u_3) \le 6\mu_0 + 2(g(u_1)+g(u_2)+g(u_3))\mu_1 + 2(g(u_1)g(u_2)+g(u_1)g(u_3)+g(u_2)g(u_3))\mu_2
\end{array}\right\} \\
\\
\beta_2^{\mu}(\vec{u}) & =1-\alpha_2^{\mu}(\vec{u}) =\mathbbm{1}\biggl\{\mu,\vec{u}:
\begin{pmatrix}
1&0\\
1&1
\end{pmatrix}
\begin{pmatrix}
2g(u_1) g(u_2) g(u_3) \\
2g(u_1) g(u_2) g(u_3)
\end{pmatrix}
\le 
\begin{pmatrix}
1&0\\
1&1
\end{pmatrix}
\begin{pmatrix}
2g(u_1)          & 2g(u_1) g(u_3) \\
0                & 2g(u_1)g(u_2)
\end{pmatrix}
\begin{pmatrix}
\mu_1 \\
\mu_2
\end{pmatrix}
\biggr\} \\
&= \mathbbm{1}\biggl\{\mu,\vec{u}:
\left\{
\begin{array}{c}
2g(u_1) g(u_2) g(u_3) - 2g(u_1) \mu_1 - 2g(u_1)g(u_3)\mu_2 \le 0 \\
4g(u_1) g(u_2) g(u_3) - 2g(u_1) \mu_1 - 2(g(u_1)g(u_3) + g(u_1)g(u_2))\mu_2 \le 0
\end{array}
\right\}
\biggr\}\\
\beta_3^{\mu}(\vec{u})&=1-\alpha_3^{\mu}(\vec{u})=\mathbbm{1}\biggl\{\mu,\vec{u}:2g(u_1) g(u_2) g(u_3) \le 2g(u_1)g(u_2)\mu_2 \biggr\}
\end{aligned}
\end{equation}}

\end{proof}

\subsection{Proof of Lemma \ref{lemma:lagrangian_intermsof_mu0}}
\label{app:lagrangian_intermsof_mu0}
\begin{proof}
We consider the MHT with $K=3$.
Focusing on the second term of the Lagrangian \eqref{eq:lagrangian1} in our example, we have 

{\small
\begin{eqnarray}
\int_{Q} \sum_{i=1}^{3} D_{i}^{\mu}(\vec{u}) R_{i}(\mu,\vec{u})d \vec{u} 
&\stackrel{\mbox{\eqref{eq:dimus2}}}{=}& \int_{Q} 
\alpha_1^{\mu}(\vec{u}) R_{1}(\mu,\vec{u}) + 
\alpha_1^{\mu}(\vec{u}) \alpha_2^{\mu}(\vec{u}) R_{2}(\mu,\vec{u}) + \alpha_1^{\mu}(\vec{u}) \alpha_2^{\mu}(\vec{u}) \alpha_3^{\mu}(\vec{u}) R_{3}(\mu,\vec{u}) d\vec{u} \nonumber \\
&\stackrel{\mbox{\eqref{eq:beta_mu_u}}}{=}& \int_{Q} 
(1-\beta_1^{\mu}(\vec{u})) R_{1}(\mu,\vec{u}) + 
(1-\beta_1^{\mu}(\vec{u})) (1-\beta_2^{\mu}(\vec{u})) R_{2}(\mu,\vec{u})\nonumber\\
&& + (1-\beta_1^{\mu}(\vec{u}))(1-\beta_2^{\mu}(\vec{u}))(1-\beta_3^{\mu}(\vec{u})) R_{3}(\mu,\vec{u}) d\vec{u} \nonumber\\
&=& \int_{Q} R_{1}(\mu,\vec{u})+R_{2}(\mu,\vec{u})+R_{3}(\mu,\vec{u})d\vec{u} \nonumber\\
&& + \int_Q \biggl(-\beta_1^{\mu}(\vec{u})R_{1}(\mu,\vec{u})-\beta_1^{\mu}(\vec{u})R_{2}(\mu,\vec{u})-\beta_2^{\mu}(\vec{u})R_{2}(\mu,\vec{u})-\beta_1^{\mu}(\vec{u})R_{3}(\mu,\vec{u})\nonumber\\
&& -\beta_2^{\mu}(\vec{u})R_{3}(\mu,\vec{u})-\beta_3^{\mu}(\vec{u})R_{3}(\mu,\vec{u})\biggr) d\vec{u}\nonumber\\
&& + \int_Q \beta_1^{\mu}(\vec{u})\beta_2^{\mu}(\vec{u})R_{2}(\mu,\vec{u}) + \biggl(\beta_1^{\mu}(\vec{u})\beta_2^{\mu}(\vec{u})+\beta_1^{\mu}(\vec{u})\beta_3^{\mu}(\vec{u})\\
&& +\beta_2^{\mu}(\vec{u})\beta_3^{\mu}(\vec{u}) \biggr)R_{3}(\mu,\vec{u}) d\vec{u} \nonumber\\
&& - \int_Q \beta_1^{\mu}(\vec{u})\beta_2^{\mu}(\vec{u}) \beta_3^{\mu}(\vec{u})R_{3}(\mu,\vec{u}) d\vec{u} \nonumber\\
&=& \int_{Q} R_{1}(\mu,\vec{u})+R_{2}(\mu,\vec{u})+R_{3}(\mu,\vec{u})d\vec{u}\nonumber\\ 
&&  - \int_Q \biggl( \beta_1^{\mu}(\vec{u})(R_{1}(\mu,\vec{u})+R_{2}(\mu,\vec{u})+R_{3}(\mu,\vec{u})) \nonumber\\
&&  + \beta_2^{\mu}(\vec{u})(R_{2}(\mu,\vec{u})+R_{3}(\mu,\vec{u})) + \beta_3^{\mu}(\vec{u})R_{3}(\mu,\vec{u})\biggr) d\vec{u} \nonumber\\
&&  + \int_Q \beta_1^{\mu}(\vec{u})\beta_2^{\mu}(\vec{u})(R_{2}(\mu,\vec{u})+R_{3}(\mu,\vec{u})) + \beta_3^{\mu}(\vec{u})(\beta_1^{\mu}(\vec{u})+\beta_2^{\mu}(\vec{u}))R_{3}(\mu,\vec{u}) d\vec{u} \nonumber\\
&&  -\int_Q \beta_1^{\mu}(\vec{u})\beta_2^{\mu}(\vec{u})\beta_3^{\mu}(\vec{u})R_{3}(\mu,\vec{u}) d\vec{u} \nonumber\\
&=& \int_{Q} (1-\beta_1^{\mu}(\vec{u}))(R_{1}(\mu,\vec{u})+R_{2}(\mu,\vec{u})+R_{3}(\mu,\vec{u}))d\vec{u} \nonumber\\
&&  -\int_Q \beta_2^{\mu}(\vec{u})(R_{2}(\mu,\vec{u})+R_{3}(\mu,\vec{u}))(1-\beta_1^{\mu}(\vec{u})) d\vec{u} \nonumber\\
&&  - \int_Q \beta_3^{\mu}(\vec{u})R_{3}(\mu,\vec{u})(1-\beta_1^{\mu}(\vec{u})-\beta_2^{\mu}(\vec{u})+\beta_1^{\mu}(\vec{u})\beta_2^{\mu}(\vec{u}))d\vec{u}\nonumber\\
&\stackrel{\mbox{\eqref{eq:beta_mu_u}}}{=}& \int_{Q} \alpha_1^{\mu}(\vec{u})(R_{1}(\mu,\vec{u})+R_{2}(\mu,\vec{u})+R_{3}(\mu,\vec{u}))d\vec{u} \nonumber\\
&&  -\int_Q \alpha_1^{\mu}(\vec{u})\beta_2^{\mu}(\vec{u})(R_{2}(\mu,\vec{u})+R_{3}(\mu,\vec{u})) d\vec{u} \nonumber\\
&&  - \int_Q \beta_3^{\mu}(\vec{u})R_{3}(\mu,\vec{u})(\alpha_1^{\mu}(\vec{u})\alpha_2^{\mu}(\vec{u}))d\vec{u}\nonumber\\
&=& \int_{Q} (R_{1}(\mu,\vec{u})+R_{2}(\mu,\vec{u})+R_{3}(\mu,\vec{u}))d\vec{u} \label{eq:lagrangian_2nd_term} \\
&& - \int_{Q} \beta_1^{\mu}(\vec{u})(R_{1}(\mu,\vec{u})+R_{2}(\mu,\vec{u})+R_{3}(\mu,\vec{u})) d\vec{u}\nonumber\\
&&  -\int_Q \alpha_1^{\mu}(\vec{u})\beta_2^{\mu}(\vec{u})(R_{2}(\mu,\vec{u})+R_{3}(\mu,\vec{u})) d\vec{u} \nonumber\\
&&  - \int_Q \alpha_1^{\mu}(\vec{u})\alpha_2^{\mu}(\vec{u}) \beta_3^{\mu}(\vec{u}) R_{3}(\mu,\vec{u}) d\vec{u}  \nonumber
\end{eqnarray}}

Thus, from \eqref{eq:lagrangian1} with $K=3$, and \eqref{eq:lagrangian_2nd_term} we have the objective function as,
\begin{equation}\label{eq:lagrangian3}
\begin{aligned}
L(\vec{D}^\mu, \mu) &= \sum_{l=0}^{2} \mu_{l} \times \alpha+ \int_{Q} \sum_{i=1}^{3} D_{i}^{\mu}(\vec{u}) R_{i}(\mu,\vec{u})d \vec{u} \\
&= \sum_{l=0}^{2} \mu_{l} \times \alpha+ \int_{Q} (R_{1}(\mu,\vec{u})+R_{2}(\mu,\vec{u})+R_{3}(\mu,\vec{u}))d\vec{u} \\
&- \int_{Q} \beta_1^{\mu}(\vec{u})(R_{1}(\mu,\vec{u})+R_{2}(\mu,\vec{u})+R_{3}(\mu,\vec{u})) d\vec{u}\\
&  -\int_Q \alpha_1^{\mu}(\vec{u})\beta_2^{\mu}(\vec{u})(R_{2}(\mu,\vec{u})+R_{3}(\mu,\vec{u})) d\vec{u} \\
& - \int_Q \alpha_1^{\mu}(\vec{u})\alpha_2^{\mu}(\vec{u})\beta_3^{\mu}(\vec{u})R_{3}(\mu,\vec{u})d\vec{u}
\end{aligned}
\end{equation}

Define,
\begin{equation*}\label{eq:bi}
b_3(\vec{u})=\begin{pmatrix}
2g(u_1) g(u_2) g(u_3) \\
2g(u_1) g(u_2) g(u_3) \\
2g(u_1) g(u_2) g(u_3)
\end{pmatrix}, \ 
b_2(\vec{u})=\begin{pmatrix}
2g(u_1) g(u_2) g(u_3) \\
2g(u_1) g(u_2) g(u_3)
\end{pmatrix}, \
b_1(\vec{u})= 2g(u_1) g(u_2) g(u_3)
\end{equation*}
We also define,
\begin{equation}\label{eq:Ai}
A_1(\vec{u})=\begin{pmatrix}
6 & 2(g(u_2)+g(u_3))  & 2g(u_2)g(u_3)\\
0  & 2g(u_1)          & 2g(u_1) g(u_3) \\
0  & 0                & 2g(u_1)g(u_2)
\end{pmatrix}, \ 
A_2(\vec{u})=\begin{pmatrix}
2g(u_1)          & 2g(u_1) g(u_3) \\
0                & 2g(u_1)g(u_2)
\end{pmatrix}, \ 
A_3(\vec{u})=2g(u_1)g(u_2)
\end{equation}
Then from \eqref{eq:beta_mu_u} we have,
\begin{equation*}\label{eq:beta_mu_u_2}
\begin{aligned}
\beta_1^{\mu}(\vec{u})&=\mathbbm{1}\biggl\{\mu,\vec{u}:
\begin{pmatrix}
1&0&0\\
1&1&0\\
1&1&1
\end{pmatrix} b_3(\vec{u})
\le 
\begin{pmatrix}
1&0&0\\
1&1&0\\
1&1&1
\end{pmatrix} A_1(\vec{u}) \mu
\biggr\} \\
\beta_2^{\mu}(\vec{u})&=\mathbbm{1}\biggl\{\mu,\vec{u}:
\begin{pmatrix}
1&0\\
1&1
\end{pmatrix} b_2(\vec{u})
\le 
\begin{pmatrix}
1&0\\
1&1
\end{pmatrix} A_2(\vec{u}) \mu
\biggr\}\\
\beta_3^{\mu}(\vec{u})&=\mathbbm{1}\biggl\{\mu,\vec{u}: b_1(\vec{u})\le A_3(\vec{u}) \mu
\biggr\}\\
\end{aligned}
\end{equation*}
Additionally, we can express the $R_{i}(\mu,\vec{u})$'s as,
\begin{equation*}\label{eq:ri_mu_2}
\begin{aligned}
R_{1}(\mu,\vec{u})+R_{2}(\mu,\vec{u})+R_{3}(\mu,\vec{u})&=1^T \begin{pmatrix}
R_{1}(\mu,\vec{u})\\
R_{2}(\mu,\vec{u})\\
R_{3}(\mu,\vec{u})
\end{pmatrix}=1^T(b_3(\vec{u})-A_1(\vec{u})\mu)\\
R_{2}(\mu,\vec{u})+R_{3}(\mu,\vec{u})&=
\begin{pmatrix}
0&1&1
\end{pmatrix}\begin{pmatrix}
R_{1}(\mu,\vec{u})\\
R_{2}(\mu,\vec{u})\\
R_{3}(\mu,\vec{u})
\end{pmatrix}=\begin{pmatrix}
0&1&1
\end{pmatrix}(b_3(\vec{u})-A_1(\vec{u})\mu)\\
R_{3}(\mu,\vec{u})&=
\begin{pmatrix}
0&0&1
\end{pmatrix}\begin{pmatrix}
R_{1}(\mu,\vec{u})\\
R_{2}(\mu,\vec{u})\\
R_{3}(\mu,\vec{u})
\end{pmatrix}=\begin{pmatrix}
0&0&1
\end{pmatrix}(b_3(\vec{u})-A_1(\vec{u})\mu)
\end{aligned}
\end{equation*}
We can now reformulate the objective function \eqref{eq:lagrangian3} as:
{\footnotesize
\begin{equation}\label{eq:lagrangian4}
\begin{aligned}
L(\vec{D}^\mu, \mu) =& \alpha \times 1^T\mu+ \int_Q 1^T(b_3(\vec{u})-A_1(\vec{u})\mu) d\vec{u}\\
& -\int_Q \beta_1^{\mu}(\vec{u}) 1^T(b_3(\vec{u})-A_1(\vec{u})\mu) d\vec{u} \\
& -\int_Q \alpha_1^{\mu}(\vec{u})\beta_2^{\mu}(\vec{u})\begin{pmatrix}
0&1&1
\end{pmatrix}(b_3(\vec{u})-A_1(\vec{u})\mu) d\vec{u} \\
&- \int_Q \alpha_1^{\mu}(\vec{u})\alpha_2^{\mu}(\vec{u})\beta_3^{\mu}(\vec{u})\begin{pmatrix}
0&0&1
\end{pmatrix}(b_3(\vec{u})-A_1(\vec{u})\mu)d\vec{u}\\
=& \alpha \times 1^T\mu \\
&+ \int_Q \biggl(1^Tb_3(\vec{u})- \beta_1^{\mu}(\vec{u}) 1^Tb_3(\vec{u}) - \alpha_1^{\mu}(\vec{u})\beta_2^{\mu}(\vec{u})\begin{pmatrix}
0&1&1
\end{pmatrix}b_3(\vec{u})-\alpha_1^{\mu}(\vec{u})\alpha_2^{\mu}(\vec{u})\beta_3^{\mu}(\vec{u})\begin{pmatrix}
0&0&1
\end{pmatrix}b_3(\vec{u}) \biggr)d\vec{u} \\
&+ \int_Q \biggl(-1^TA_1(\vec{u})+\beta_1^{\mu}(\vec{u}) 1^TA_1(\vec{u})+ \alpha_1^{\mu}(\vec{u})\beta_2^{\mu}(\vec{u})\begin{pmatrix}
0&1&1
\end{pmatrix}A_1(\vec{u})\\
&+\alpha_1^{\mu}(\vec{u})\alpha_2^{\mu}(\vec{u})\beta_3^{\mu}(\vec{u})\begin{pmatrix}
0&0&1
\end{pmatrix}A_1(\vec{u})\biggr)\mu d\vec{u}\\
=& \alpha \times 1^T\mu \\
&+ \int_Q 2g(u_1) g(u_2) g(u_3) \biggl(3- 3\beta_1^{\mu}(\vec{u}) - 2\alpha_1^{\mu}(\vec{u})\beta_2^{\mu}(\vec{u})-\alpha_1^{\mu}(\vec{u})\alpha_2^{\mu}(\vec{u})\beta_3^{\mu}(\vec{u}) \biggr)d\vec{u} \\
&+ \int_Q \biggl(-1^TA_1(\vec{u})+\beta_1^{\mu}(\vec{u}) 1^TA_1(\vec{u})+ \alpha_1^{\mu}(\vec{u})\beta_2^{\mu}(\vec{u})\begin{pmatrix}
0&1&1
\end{pmatrix}A_1(\vec{u})\\
&+\alpha_1^{\mu}(\vec{u})\alpha_2^{\mu}(\vec{u})\beta_3^{\mu}(\vec{u})\begin{pmatrix}
0&0&1
\end{pmatrix}A_1(\vec{u})\biggr)\mu d\vec{u}
\end{aligned}
\end{equation}}
Note here $Q=\left\{\vec{u}: 0 \leq u_{1} \leq u_{2} \leq u_{3} \leq 1\right\}$.
The objective is to minimize the  $L(\vec{D}^\mu, \mu)$ w.r.t. $\mu$.
We denote by $f_1(\mu,\vec{u})$ and $f_2(\mu,\vec{u})$ the second and third terms of the summation $L(\vec{D}^\mu, \mu)$ in \eqref{eq:lagrangian4} respectively. 
That is,
\begin{equation}\label{eq:f1}
\begin{aligned}
f_1(\mu,\vec{u})&=  \int_Q 2g(u_1) g(u_2) g(u_3) \biggl(3- 3\beta_1^{\mu}(\vec{u}) - 2\alpha_1^{\mu}(\vec{u})\beta_2^{\mu}(\vec{u})-\alpha_1^{\mu}(\vec{u})\alpha_2^{\mu}(\vec{u})\beta_3^{\mu}(\vec{u}) \biggr)d\vec{u}\\
&= \int_Q 2g(u_1) g(u_2) g(u_3) \biggl(3\alpha_1^{\mu}(\vec{u}) - 2\alpha_1^{\mu}(\vec{u})\beta_2^{\mu}(\vec{u})-\alpha_1^{\mu}(\vec{u})\alpha_2^{\mu}(\vec{u})\beta_3^{\mu}(\vec{u}) \biggr)d\vec{u}\\
&= \int_Q 2g(u_1) g(u_2) g(u_3) \biggl(3 - 2\beta_2^{\mu}(\vec{u})-\alpha_2^{\mu}(\vec{u})\beta_3^{\mu}(\vec{u}) \biggr)\alpha_1^{\mu}(\vec{u})d\vec{u}\\
&= \int_Q 2g(u_1) g(u_2) g(u_3) \biggl(1 + 2\alpha_2^{\mu}(\vec{u})-\alpha_2^{\mu}(\vec{u})\beta_3^{\mu}(\vec{u}) \biggr)\alpha_1^{\mu}(\vec{u})d\vec{u}\\
&= \int_Q 2g(u_1) g(u_2) g(u_3) \biggl(1 + \alpha_2^{\mu}(\vec{u})+\alpha_2^{\mu}(\vec{u})\alpha_3^{\mu}(\vec{u}) \biggr)\alpha_1^{\mu}(\vec{u})d\vec{u}\\
\end{aligned}
\end{equation}

{\footnotesize
\begin{equation}\label{eq:f2}
\begin{aligned}
f_2(\mu,\vec{u})=& \int_Q \biggl(-1^TA_1(\vec{u})+\beta_1^{\mu}(\vec{u}) 1^TA_1(\vec{u})+ \alpha_1^{\mu}(\vec{u})\beta_2^{\mu}(\vec{u})\begin{pmatrix}
0&1&1
\end{pmatrix}A_1(\vec{u})\\
&+\alpha_1^{\mu}(\vec{u})\alpha_2^{\mu}(\vec{u})\beta_3^{\mu}(\vec{u})\begin{pmatrix}
0&0&1
\end{pmatrix}A_1(\vec{u})\biggr)\mu d\vec{u}    \\
=&  \int_Q \biggl(-1^T+\beta_1^{\mu}(\vec{u}) 1^T+ \alpha_1^{\mu}(\vec{u})\beta_2^{\mu}(\vec{u})\begin{pmatrix}
0&1&1
\end{pmatrix}+\alpha_1^{\mu}(\vec{u})\alpha_2^{\mu}(\vec{u})\beta_3^{\mu}(\vec{u})\begin{pmatrix}
0&0&1
\end{pmatrix}\biggr)A_1(\vec{u})\mu d\vec{u} \\
=& \int_Q \begin{pmatrix}
-\alpha_1^{\mu}(\vec{u})\\
-\alpha_1^{\mu}(\vec{u})+\alpha_1^{\mu}(\vec{u})\beta_2^{\mu}(\vec{u})\\
-\alpha_1^{\mu}(\vec{u})+\alpha_1^{\mu}(\vec{u})\beta_2^{\mu}(\vec{u})+\alpha_1^{\mu}(\vec{u})\alpha_2^{\mu}(\vec{u})\beta_3^{\mu}(\vec{u})
\end{pmatrix}^TA_1(\vec{u})\mu d\vec{u}\\
=& \int_Q \biggl(-1^T+\beta_1^{\mu}(\vec{u}) 1^T+ \alpha_1^{\mu}(\vec{u})\beta_2^{\mu}(\vec{u})\begin{pmatrix}
0&1&1
\end{pmatrix}+\alpha_1^{\mu}(\vec{u})\alpha_2^{\mu}(\vec{u})\beta_3^{\mu}(\vec{u})\begin{pmatrix}
0&0&1
\end{pmatrix}\biggr)A_1(\vec{u})\mu d\vec{u} \\
&\stackrel{\mbox{\eqref{eq:Ai}}}{=} \int_Q \begin{pmatrix}
-\alpha_1^{\mu}(\vec{u})\\
-\alpha_1^{\mu}(\vec{u})+\alpha_1^{\mu}(\vec{u})\beta_2^{\mu}(\vec{u})\\
-\alpha_1^{\mu}(\vec{u})+\alpha_1^{\mu}(\vec{u})\beta_2^{\mu}(\vec{u})+\alpha_1^{\mu}(\vec{u})\alpha_2^{\mu}(\vec{u})\beta_3^{\mu}(\vec{u})
\end{pmatrix}^T \begin{pmatrix}
6 & 2(g(u_2)+g(u_3))  & 2g(u_2)g(u_3)\\
0  & 2g(u_1)          & 2g(u_1) g(u_3) \\
0  & 0                & 2g(u_1)g(u_2)
\end{pmatrix} \mu d\vec{u} \\
=& \int_Q \begin{pmatrix}
-6\alpha_1^\mu(\vec{u})\\
-2\alpha_1^\mu(\vec{u})\bigl[g(u_2) + g(u_3)\bigr]+2g(u_1)\Bigl[-\alpha_1^\mu(\vec{u})+\alpha_1^\mu(\vec{u})\beta_2^\mu(\vec{u})\Bigr]\\
-2\alpha_1^\mu(\vec{u})g(u_2)g(u_3)+2\bigl[-\alpha_1^\mu + \alpha_1^\mu\beta_2^\mu\bigr]g(u_1)g(u_3)+ 2g(u_1)g(u_2)\bigl[-\alpha_1^\mu + \alpha_1^\mu\beta_2^\mu +\alpha_1^\mu\alpha_2^\mu\beta_3^\mu\bigr] 
\end{pmatrix}^T \mu d\vec{u}\\
=& \int_Q \begin{pmatrix}
-6\alpha_1^\mu(\vec{u})\\
-2\alpha_1^\mu(\vec{u})\bigl[g(u_1)+g(u_2) + g(u_3)\bigr]+2g(u_1)\alpha_1^\mu(\vec{u})\beta_2^\mu(\vec{u})\\
2\alpha_{1}^{\mu}(\vec{u})\Bigl[-g(u_{2})g(u_{3})+\bigl(-1 + \beta_{2}^{\mu}(\vec{u})\bigr)g(u_{1})g(u_{3})+\bigl(-1 + \beta_{2}^{\mu}(\vec{u}) + \alpha_{2}^{\mu}(\vec{u})\beta_{3}^{\mu}(\vec{u})\bigr)
g(u_{1})g(u_{2})\Bigr] 
\end{pmatrix}^T \mu d\vec{u}\\
=& \int_Q \begin{pmatrix}
-6\alpha_1^\mu(\vec{u})\\
-2\alpha_1^\mu(\vec{u})\bigl[g(u_1)+g(u_2) + g(u_3)-g(u_1)\beta_2^\mu(\vec{u})\bigr]\\
2\alpha_{1}^{\mu}(\vec{u})\Bigl(-g(u_{2})g(u_{3})-\alpha_{2}^{\mu}(\vec{u})g(u_{1})g(u_{3})+\bigl(-\alpha_{2}^{\mu}(\vec{u}) + \alpha_{2}^{\mu}(\vec{u})\beta_{3}^{\mu}(\vec{u})\bigr)
g(u_{1})g(u_{2})\Bigr)
\end{pmatrix}^T \mu d\vec{u}\\
=& \int_Q 
\begin{pmatrix}
M_{1}^{\mu}(\vec{u})\\
M_{2}^{\mu}(\vec{u})\\
M_{3}^{\mu}(\vec{u})
\end{pmatrix}^T
\mu
d\vec{u}=\int_Q\bigl[
M_{1}^{\mu}(\vec{u})\mu_0
+
M_{2}^{\mu}(\vec{u})\mu_1
+
M_{3}^{\mu}(\vec{u})\mu_2
\bigr]d\vec{u}
\end{aligned}
\end{equation}}
where,
\begin{equation*}
\begin{aligned}
M_{1}^{\mu}(\vec{u})&=-6\alpha_1^\mu(\vec{u})\\
M_{2}^{\mu}(\vec{u})&=-2\alpha_1^\mu(\vec{u})\bigl[g(u_1)+g(u_2) + g(u_3)-g(u_1)\beta_2^\mu(\vec{u})\bigr]\\
M_{3}^{\mu}(\vec{u})&=2\alpha_{1}^{\mu}(\vec{u})\Bigl(-g(u_{2})g(u_{3})-\alpha_{2}^{\mu}(\vec{u})g(u_{1})g(u_{3})+\bigl(-\alpha_{2}^{\mu}(\vec{u}) + \alpha_{2}^{\mu}(\vec{u})\beta_{3}^{\mu}(\vec{u})\bigr)
g(u_{1})g(u_{2})\Bigr)
\end{aligned}
\end{equation*}

Using \eqref{eq:f1} and \eqref{eq:f2}, we have $L(\vec{D}^\mu, \mu)$ from \eqref{eq:lagrangian4}  as,
\begin{equation}
\begin{aligned}
L(\vec{D}^\mu, \mu)&= \sum_{l=0}^{2} \mu_{l} \times \alpha+f_1(\mu,\vec{u})+f_2(\mu,\vec{u}) \\
=&\alpha(\mu_0+\mu_1+\mu_2)\\
&+\int_Q 2g(u_1) g(u_2) g(u_3) \biggl(1 + \alpha_2^{\mu}(\vec{u})+\alpha_2^{\mu}(\vec{u})\alpha_3^{\mu}(\vec{u}) \biggr)\alpha_1^{\mu}(\vec{u})d\vec{u}\\
&+\int_Q\bigl[M_{1}^{\mu}(\vec{u})\mu_0+M_{2}^{\mu}(\vec{u})\mu_1+M_{3}^{\mu}(\vec{u})\mu_2\bigr]d\vec{u}\\
=&\alpha(\mu_0+\mu_1+\mu_2)\\
&+\int_Q 2g(u_1) g(u_2) g(u_3) \biggl(1 + \alpha_2^{\mu}(\vec{u})+\alpha_2^{\mu}(\vec{u})\alpha_3^{\mu}(\vec{u}) \biggr)\alpha_1^{\mu}(\vec{u})d\vec{u}\\
&+\int_Q\biggl[-6\alpha_1^\mu(\vec{u})\mu_0-2\alpha_1^\mu(\vec{u})\Bigl(g(u_1)+g(u_2) + g(u_3)-g(u_1)\beta_2^\mu(\vec{u})\Bigr)\mu_1\\
&\quad +2\alpha_{1}^{\mu}(\vec{u})\Bigl(-g(u_{2})g(u_{3})-\alpha_{2}^{\mu}(\vec{u})g(u_{1})g(u_{3})+\bigl(-\alpha_{2}^{\mu}(\vec{u}) + \alpha_{2}^{\mu}(\vec{u})\beta_{3}^{\mu}(\vec{u})\bigr)g(u_{1})g(u_{2})\Bigr)\mu_2\biggr]d\vec{u}\\
=&\alpha(\mu_0+\mu_1+\mu_2)\\
&+2\int_Q g(u_1) g(u_2) g(u_3)\alpha_1^{\mu}(\vec{u}) \biggl(1 + \alpha_2^{\mu}(\vec{u})+\alpha_2^{\mu}(\vec{u})\alpha_3^{\mu}(\vec{u}) \biggr)d\vec{u}\\
&-6\int_Q \alpha_1^\mu(\vec{u})\mu_0d\vec{u} \\
&-2\int_Q \alpha_1^\mu(\vec{u})\Bigl(g(u_1)\alpha_2^{\mu}(\vec{u})+g(u_2) + g(u_3)\Bigr)\mu_1d\vec{u}\\
& -2\int_Q\alpha_{1}^{\mu}(\vec{u})\Bigl(g(u_{2})g(u_{3})+\alpha_{2}^{\mu}(\vec{u})g(u_{1})g(u_{3})+\alpha_{2}^{\mu}(\vec{u})\alpha_{3}^{\mu}(\vec{u})g(u_{1})g(u_{2})\Bigr)\mu_2 d\vec{u}.
\end{aligned}
\end{equation}
This completes the proof.

Note that the Lagrangian expression above can be further rewritten as
\begin{equation}
\begin{aligned}
L(\vec{D}^\mu, \mu) =& \alpha(\mu_1+\mu_2)+\alpha\mu_0-6\int_Q\alpha_1^\mu(\vec{u})\mu_0d\vec{u}\\
&+ \int_Q 2\alpha_1^\mu(\vec{u})\Biggl[\biggl(g(u_1) g(u_2) g(u_3) \bigl(1 + \alpha_2^{\mu}(\vec{u})+\alpha_2^{\mu}(\vec{u})\alpha_3^{\mu}(\vec{u}) \bigr) \biggr)\\
&- \biggl(g(u_1)+g(u_2) + g(u_3)-g(u_1)\beta_2^\mu(\vec{u}) \biggr)\mu_1\\
&\quad+ \biggl(-g(u_{2})g(u_{3})-\alpha_{2}^{\mu}(\vec{u})g(u_{1})g(u_{3})- \alpha_{2}^{\mu}(\vec{u})\alpha_{3}^{\mu}(\vec{u})g(u_{1})g(u_{2}) \biggr)\mu_2\Biggr]d\vec{u}\\
=&c_1(\mu_1,\mu_2)+\alpha\mu_0-6\int_Q\alpha_1^\mu(\vec{u})\mu_0d\vec{u}+\int_Q 2\alpha_1^\mu(\vec{u})c_2(\vec{u},\mu_1,\mu_2)d\vec{u}
\end{aligned}
\end{equation}
where,
{\small
\begin{equation}
\begin{aligned}
c_1(\mu_1,\mu_2)&=\alpha(\mu_1+\mu_2)\\
c_2(\vec{u},\mu_1,\mu_2)&=\biggl(g(u_1) g(u_2) g(u_3) \Bigl(1 + \alpha_2^{\mu}(\vec{u})+\alpha_2^{\mu}(\vec{u})\alpha_3^{\mu}(\vec{u}) \Bigr) \biggr)- \biggl(g(u_1)+g(u_2) + g(u_3)-g(u_1)\beta_2^\mu(\vec{u}) \biggr)\mu_1\\
& + \biggl(-g(u_{2})g(u_{3})-\alpha_{2}^{\mu}(\vec{u})g(u_{1})g(u_{3})- \alpha_{2}^{\mu}(\vec{u})\alpha_{3}^{\mu}(\vec{u})g(u_{1})g(u_{2}) \biggr)\mu_2
\end{aligned}
\end{equation}}
are constants w.r.t. $\mu_0$.
\end{proof}

\section{\texorpdfstring{Proof of Theorem \ref{thm:mu_optimality_combined}: Minimizing the Lagrangian w.r.t. $\mu$}{Proof of Theorem: Minimizing the Lagrangian w.r.t. mu}}
\label{app:mu_optimality}
In this section, we complete the dual analysis by characterising the minimiser $\mu^\ast = (\mu_0^\ast,\mu_1^\ast,\mu_2^\ast)$ of the Lagrangian $L(\vec{D}^\mu,\mu)$ in \eqref{eq:lagrangian5_general}.
Our starting point is the dual formulation \eqref{eq:dual2} of the primal problem \eqref{eq:objconst}, together with the explicit optimal decision rule $\vec{D}^\mu$ from Lemma~\ref{lem:opt_decision_mu} and the decomposition of $L(\vec{D}^\mu,\mu)$ in Lemma~\ref{lemma:lagrangian_intermsof_mu0}.
We then treat the Lagrange multipliers coordinate-wise: fixing two coordinates at a time, we analyze the remaining one and the resulting change in the Lagrangian.
Sections~\ref{app:mu0_optimality}, \ref{app:mu1_optimality}, and~\ref{app:mu2_optimality} carry out this program for $\mu_0$, $\mu_1$, and $\mu_2$, respectively and yield the coordinate-wise characterisations of any minimiser $\mu^\ast$ as provided in Theorem~\ref{thm:mu_optimality_combined}.
Finally, Section~\ref{sec:d_mu_optimality} establishes that the decision policy parameterised by the optimal multipliers satisfies strong duality, confirming it as the globally optimal solution to the primal problem.

\begin{proof}
We focus on the dual formulation \eqref{eq:dual2} corresponding to the multiple hypotheses testing problem \eqref{eq:objconst}, and analyze the associated Lagrangian $L(\vec{D}^\mu, \mu)$. 
Focusing on one coordinate at a time, and for fixed values of the remaining coordinates, we subsequently derive the necessary and sufficient conditions for optimality.
The proof of Theorem \ref{thm:mu_optimality_combined} follows from combining Theorems \ref{thm:mu0_optimality}, \ref{thm:mu1_optimality} and \ref{thm:mu2_optimality} which are detailed separately in the following sections \ref{app:mu0_optimality}, \ref{app:mu1_optimality} and \ref{app:mu2_optimality} respectively.
\end{proof}

\subsection{\texorpdfstring{Minimizing the Lagrangian w.r.t. $\mu_0$}{Minimizing the Lagrangian w.r.t. mu0}}
\label{app:mu0_optimality}
We first begin with minimizing $L(\vec{D}^\mu, \mu)$ from \eqref{eq:lagrangian5_general} w.r.t. $\mu_0$.
We begin by examining the structure of the Lagrangian and isolating the terms that depend on \(\mu_0\) from those that do not. 
Specifically, Lemma \ref{lemma:lagrangian_intermsof_mu0} shows that the Lagrangian expression in \eqref{eq:lagrangian5_general} can be rewritten as
\begin{equation}\label{eq:lagrangian5}
\begin{aligned}
L(\vec{D}^\mu, \mu) \stackrel{\mbox{\eqref{eq:lagrangian5_general}}}{=}& \alpha(\mu_1+\mu_2)+\alpha\mu_0\\
&-6\int_Q\alpha_1^\mu(\vec{u})\mu_0d\vec{u}\\
&+ \int_Q 2\alpha_1^\mu(\vec{u})\Biggl[\biggl(g(u_1) g(u_2) g(u_3) \bigl(1 + \alpha_2^{\mu}(\vec{u})+\alpha_2^{\mu}(\vec{u})\alpha_3^{\mu}(\vec{u}) \bigr) \biggr) \\
&- \biggl(g(u_1)+g(u_2) + g(u_3)-g(u_1)\beta_2^\mu(\vec{u}) \biggr)\mu_1\\
&\quad+ \biggl(-g(u_{2})g(u_{3})-\alpha_{2}^{\mu}(\vec{u})g(u_{1})g(u_{3})- \alpha_{2}^{\mu}(\vec{u})\alpha_{3}^{\mu}(\vec{u})g(u_{1})g(u_{2}) \biggr)\mu_2\Biggr]d\vec{u}\\
=&c_1(\mu_1,\mu_2)+\alpha\mu_0-6\int_Q\alpha_1^\mu(\vec{u})\mu_0d\vec{u}+\int_Q 2\alpha_1^\mu(\vec{u})c_2(\vec{u},\mu_1,\mu_2)d\vec{u}
\end{aligned}
\end{equation}
where,
\begin{equation}\label{eq:c1_c2}
\begin{aligned}
c_1(\mu_1,\mu_2)&=\alpha(\mu_1+\mu_2)\\
c_2(\vec{u},\mu_1,\mu_2)&=g(u_1) g(u_2) g(u_3) \Bigl(1 + \alpha_2^{\mu}(\vec{u})+\alpha_2^{\mu}(\vec{u})\alpha_3^{\mu}(\vec{u}) \Bigr)- \biggl(g(u_1)+g(u_2) + g(u_3)-g(u_1)\beta_2^\mu(\vec{u}) \biggr)\mu_1\\
& + \biggl(-g(u_{2})g(u_{3})-\alpha_{2}^{\mu}(\vec{u})g(u_{1})g(u_{3})- \alpha_{2}^{\mu}(\vec{u})\alpha_{3}^{\mu}(\vec{u})g(u_{1})g(u_{2}) \biggr)\mu_2
\end{aligned}
\end{equation}
are constants w.r.t. $\mu_0$.

Since we minimize $L(\vec{D}^\mu, \mu)$ w.r.t. $\mu_0$ we denote it by $L(\mu_0)$. 
Similarly, we denote $\alpha_1^{\mu}(\vec{u})$ by $\alpha_1^{\mu_0}(\vec{u})$.
Thus, if $\mu_0^\ast$ is the minimizer, then for any $\delta\neq0$ we have $$L(\mu_0^\ast) \le L(\mu_0^\ast+\delta).$$
That is,
\begin{equation}\label{eq:mu0_min}
\begin{aligned}
0 &\le L(\mu_0^\ast+\delta)- L(\mu_0^\ast) \\
 &\stackrel{\mbox{\eqref{eq:lagrangian5}}}{=} \alpha.\delta-6\int_Q\alpha_1^{\mu_0^\ast+\delta}(\vec{u}).(\mu_0^\ast+\delta)d\vec{u}+6\int_Q\alpha_1^{\mu_0^\ast}(\vec{u})\mu_0^\ast d\vec{u}+\int_Q 2[\alpha_1^{\mu_0^\ast+\delta}(\vec{u})-\alpha_1^{\mu_0^\ast}(\vec{u})]c_2(\vec{u},\mu_1,\mu_2)d\vec{u}\\
&=\alpha.\delta-6\delta\int_Q\alpha_1^{\mu_0^\ast+\delta}(\vec{u})d\vec{u}+6\int_Q[\alpha_1^{\mu_0^\ast}(\vec{u})-\alpha_1^{\mu_0^\ast+\delta}(\vec{u})]\mu_0^\ast d\vec{u}+\int_Q 2[\alpha_1^{\mu_0^\ast+\delta}(\vec{u})-\alpha_1^{\mu_0^\ast}(\vec{u})]c_2(\vec{u},\mu_1,\mu_2)d\vec{u}\\
&=\delta\biggl[\alpha-6\int_Q\alpha_1^{\mu_0^\ast+\delta}(\vec{u})d\vec{u} \biggr]+\int_Q\biggl[\alpha_1^{\mu_0^\ast}(\vec{u})-\alpha_1^{\mu_0^\ast+\delta}(\vec{u})\biggr]\biggl[6\mu_0^\ast-2c_2(\vec{u},\mu_1,\mu_2)\biggr]d\vec{u}
\end{aligned}
\end{equation}
where $c_2(\vec{u},\mu_1,\mu_2)$ is defined as in \eqref{eq:c1_c2}. 

To characterize the minimizer $\mu_0^\ast$ of the Lagrangian $L(\vec{D}^\mu, \mu)$ with respect to $\mu_0$, we analyze the behavior of $L(\vec{D}^\mu, \mu)$ under small perturbations of $\mu_0$ as in \eqref{eq:mu0_min}. 
We now state the necessary and sufficient condition for $\mu_0^\ast$ to be the minimizer of the Lagrangian $L(\vec{D}^\mu, \mu)$ with respect to $\mu_0$.
\begin{theorem}
\label{thm:mu0_optimality}
Suppose Assumptions~\ref{as:assumption3} and~\ref{as:assumption4} hold.
Consider the dual formulation \eqref{eq:dual2} of the multiple hypotheses testing problem \eqref{eq:objconst} and the associated Lagrangian $L(\vec{D}^\mu, \mu)$ of the dual problem as given by \eqref{eq:lagrangian5}.
Let $\mu_0^\ast \in \mathbb{R}_+$ and fix $\mu_1, \mu_2$.

Then $\mu_0^\ast \text{ is a minimizer of } L(\vec{D}^\mu, \mu) \text{ with respect to } \mu_0$ only if
\begin{equation}\label{eq:mu_0_optimality}
\alpha = 6 \int_Q \alpha_1^{\mu_0^\ast}(\vec{u}) d\vec{u}.
\end{equation}
\end{theorem}

\begin{proof}
$\forall \delta $ we have from \eqref{eq:mu0_min},
{\footnotesize
\begin{equation}\label{eq:mu0_min2}
\begin{aligned}
0 \le & L(\mu_0^\ast+\delta)- L(\mu_0^\ast)\\
=&\delta\biggl[\alpha-6\int_Q\alpha_1^{\mu_0^\ast+\delta}(\vec{u})d\vec{u}\biggr]+\int_Q\biggl[\alpha_1^{\mu_0^\ast}(\vec{u})-\alpha_1^{\mu_0^\ast+\delta}(\vec{u})\biggr]\biggl[6\mu_0^\ast-2c_2(\vec{u},\mu_1,\mu_2)\biggr]d\vec{u}\\
=&\delta\biggl[\alpha-6\int_Q\alpha_1^{\mu_0^\ast+\delta}(\vec{u})d\vec{u}\biggr]+\int_Q\biggl[\alpha_1^{\mu_0^\ast}(\vec{u})-\alpha_1^{\mu_0^\ast+\delta}(\vec{u})\biggr]\biggl[6\mu_0^\ast-2\biggl(g(u_1) g(u_2) g(u_3) \Bigl(1 + \alpha_2^{\mu}(\vec{u})+\alpha_2^{\mu}(\vec{u})\alpha_3^{\mu}(\vec{u}) \Bigr) \biggr)\\
&  +2\biggl(g(u_1)+g(u_2) + g(u_3)-g(u_1)\beta_2^\mu(\vec{u}) \biggr)\mu_1  \\
&+2\biggl(g(u_{2})g(u_{3})+\alpha_{2}^{\mu}(\vec{u})g(u_{1})g(u_{3})+ \alpha_{2}^{\mu}(\vec{u})\alpha_{3}^{\mu}(\vec{u})g(u_{1})g(u_{2}) \biggr)\mu_2\biggr]d\vec{u}\\
=&\delta\biggl[\alpha-6\int_Q\alpha_1^{\mu_0^\ast}(\vec{u})d\vec{u}\biggr]+\delta\biggl[\int_Q 6\alpha_1^{\mu_0^\ast}(\vec{u})d\vec{u}-\int_Q 6\alpha_1^{\mu_0^\ast+\delta}(\vec{u})d\vec{u}\biggr]\\
&  +\int_Q\biggl[\alpha_1^{\mu_0^\ast}(\vec{u})-\alpha_1^{\mu_0^\ast+\delta}(\vec{u})\biggr]\biggl[6\mu_0^\ast-2\biggl(g(u_1) g(u_2) g(u_3) \Bigl(1 + \alpha_2^{\mu}(\vec{u})+\alpha_2^{\mu}(\vec{u})\alpha_3^{\mu}(\vec{u}) \Bigr) \biggr)\\
&  +2\biggl(g(u_1)+g(u_2) + g(u_3)-g(u_1)\beta_2^\mu(\vec{u}) \biggr)\mu_1 \\
&+2\biggl(g(u_{2})g(u_{3})+\alpha_{2}^{\mu}(\vec{u})g(u_{1})g(u_{3})+ \alpha_{2}^{\mu}(\vec{u})\alpha_{3}^{\mu}(\vec{u})g(u_{1})g(u_{2}) \biggr)\mu_2\biggr]d\vec{u}\\
=&\delta\biggl[\alpha-6\int_Q\alpha_1^{\mu_0^\ast}(\vec{u})d\vec{u}\biggr]+\int_Q \biggl(\alpha_1^{\mu_0^\ast}(\vec{u})-\alpha_1^{\mu_0^\ast+\delta}(\vec{u})\biggr)\biggl[6\delta+6\mu_0^\ast -2\biggl(g(u_1) g(u_2) g(u_3) \Bigl(1 + \alpha_2^{\mu}(\vec{u})+\alpha_2^{\mu}(\vec{u})\alpha_3^{\mu}(\vec{u}) \Bigr) \biggr)\\
&  +2\biggl(g(u_1)+g(u_2) + g(u_3)-g(u_1)\beta_2^\mu(\vec{u}) \biggr)\mu_1 \\
&+2\biggl(g(u_{2})g(u_{3})+\alpha_{2}^{\mu}(\vec{u})g(u_{1})g(u_{3})+ \alpha_{2}^{\mu}(\vec{u})\alpha_{3}^{\mu}(\vec{u})g(u_{1})g(u_{2}) \biggr)\mu_2 \biggr]d\vec{u}
\end{aligned}
\end{equation}}

Analyzing \eqref{eq:mu0_min2}, we need to check if the second term of the summation,
\begin{equation}\label{eq:o_delta_conjecture_mu0}
\int_Q\biggl[\alpha_1^{\mu_0^\ast}(\vec{u})-\alpha_1^{\mu_0^\ast+\delta}(\vec{u})\biggr]\biggl[6\delta + 6\mu_0^\ast-2c_2(\vec{u},\mu_1,\mu_2)\biggr]d\vec{u}=o(\delta) \\
\end{equation}

Now from \eqref{eq:beta_mu_u} we have:
{\footnotesize
\begin{equation}\label{eq:beta1}
\begin{aligned}
\beta_1^{\mu}(\vec{u})&= \mathbbm{1}\biggl\{\mu,\vec{u}:
\begin{pmatrix}
1&0&0\\
1&1&0\\
1&1&1
\end{pmatrix}
\begin{pmatrix}
2g(u_1) g(u_2) g(u_3) \\
2g(u_1) g(u_2) g(u_3) \\
2g(u_1) g(u_2) g(u_3)
\end{pmatrix}
\le 
\begin{pmatrix}
1&0&0\\
1&1&0\\
1&1&1
\end{pmatrix}
\begin{pmatrix}
6 & 2(g(u_2)+g(u_3))  & 2g(u_2)g(u_3)\\
0  & 2g(u_1)          & 2g(u_1) g(u_3) \\
0  & 0                & 2g(u_1)g(u_2)
\end{pmatrix}
\begin{pmatrix}
\mu_0 \\
\mu_1 \\
\mu_2
\end{pmatrix}
\biggr\} \\
&= \mathbbm{1}\biggl\{\mu,\vec{u}:
\begin{pmatrix}
2g(u_1) g(u_2) g(u_3) \\
4g(u_1) g(u_2) g(u_3) \\
6g(u_1) g(u_2) g(u_3)
\end{pmatrix}
\le \\
& \quad
\begin{pmatrix}
6 & 2(g(u_2)+g(u_3))  & 2g(u_2)g(u_3)\\
6 & 2(g(u_1)+g(u_2)+g(u_3)) & 2((g(u_1) g(u_3)+g(u_2)g(u_3)) \\
6 & 2(g(u_1)+g(u_2)+g(u_3)) & 2((g(u_1) g(u_3)+g(u_2)g(u_3)+g(u_1)g(u_2))  
\end{pmatrix}
\begin{pmatrix}
\mu_0 \\
\mu_1 \\
\mu_2
\end{pmatrix}
\biggr\}\\
&=\mathbbm{1}\left\{\begin{array}{c}
2g(u_1) g(u_2) g(u_3) \le 6\mu_0 + 2(g(u_2)+g(u_3))\mu_1 + 2g(u_2)g(u_3)\mu_2\\
4g(u_1) g(u_2) g(u_3) \le 6\mu_0 + 2(g(u_1)+g(u_2)+g(u_3))\mu_1 + 2(g(u_1)g(u_3)+g(u_2)g(u_3))\mu_2\\
6g(u_1) g(u_2) g(u_3) \le 6\mu_0 + 2(g(u_1)+g(u_2)+g(u_3))\mu_1 + 2(g(u_1)g(u_2)+g(u_1)g(u_3)+g(u_2)g(u_3))\mu_2
\end{array}\right\}
\end{aligned}
\end{equation}}

Thus we have for $\delta>0$ using \eqref{eq:beta1},
{\footnotesize
\begin{equation}\label{eq:b1_diff}
\begin{aligned}
&\beta_1^{\mu_0^\ast+\delta}(\vec{u})-\beta_1^{\mu_0^\ast}(\vec{u})\\
&=\mathbbm{1} \left\{\begin{array}{c}
6\mu_0^\ast \le 2g(u_1) g(u_2) g(u_3) - 2(g(u_2) + g(u_3))\mu_1 - 2g(u_2)g(u_3)\mu_2 \le 6(\mu_0^\ast+\delta)\\
4g(u_1) g(u_2) g(u_3) - 2(g(u_1)+g(u_2)+g(u_3))\mu_1 - 2(g(u_1)g(u_3)+g(u_2)g(u_3))\mu_2 \le 6(\mu_0^\ast+\delta)\\
6g(u_1) g(u_2) g(u_3) - 2(g(u_1)+g(u_2)+g(u_3))\mu_1 - 2(g(u_1)g(u_2)+g(u_1)g(u_3)+g(u_2)g(u_3))\mu_2 \le 6(\mu_0^\ast+\delta)
\end{array}\right\}\\
&+ \mathbbm{1} \left\{\begin{array}{c}
2g(u_1) g(u_2) g(u_3) - 2(g(u_2) + g(u_3))\mu_1 - 2g(u_2)g(u_3)\mu_2 \le 6 \mu_0^\ast \\
6\mu_0^\ast \le 4g(u_1) g(u_2) g(u_3) - 2(g(u_1)+g(u_2)+g(u_3))\mu_1 - 2(g(u_1)g(u_3)+g(u_2)g(u_3))\mu_2 \le 6(\mu_0^\ast+\delta)\\
6g(u_1) g(u_2) g(u_3) - 2(g(u_1)+g(u_2)+g(u_3))\mu_1 - 2(g(u_1)g(u_2)+g(u_1)g(u_3)+g(u_2)g(u_3))\mu_2 \le 6(\mu_0^\ast+\delta)
\end{array}\right\}\\
&+ \mathbbm{1} \left\{\begin{array}{c}
2g(u_1) g(u_2) g(u_3) - 2(g(u_2) + g(u_3))\mu_1 - 2g(u_2)g(u_3)\mu_2 \le 6\mu_0^\ast\\
4g(u_1) g(u_2) g(u_3) - 2(g(u_1)+g(u_2)+g(u_3))\mu_1 - 2(g(u_1)g(u_3)+g(u_2)g(u_3))\mu_2 \le 6\mu_0^\ast\\
6\mu_0^\ast \le 6g(u_1) g(u_2) g(u_3) - 2(g(u_1)+g(u_2)+g(u_3))\mu_1 - 2(g(u_1)g(u_2)+g(u_1)g(u_3)+g(u_2)g(u_3))\mu_2 \le 6(\mu_0^\ast+\delta)
\end{array}\right\}\\
&\le\mathbbm{1}\left\{\begin{array}{c}
6\mu_0^\ast \le 2g(u_1) g(u_2) g(u_3) - 2(g(u_2)+g(u_3))\mu_1 - 2g(u_2)g(u_3)\mu_2 \le 6(\mu_0^\ast+\delta)\\
2g(u_1) g(u_2) g(u_3) - 2g(u_1)\mu_1 - 2g(u_1)g(u_3)\mu_2 \le 6\delta\\
4g(u_1) g(u_2) g(u_3) - 2g(u_1)\mu_1 - 2(g(u_1)g(u_2) + g(u_1)g(u_3))\mu_2 \le 6\delta\\
\end{array}\right\}\\
&+ \mathbbm{1} \left\{\begin{array}{c}
0 \le 2g(u_1) g(u_2) g(u_3) - 2g(u_1)\mu_1 - 2g(u_1)g(u_3)\mu_2  \\
6\mu_0^\ast \le 4g(u_1) g(u_2) g(u_3) - 2(g(u_1)+g(u_2)+g(u_3))\mu_1 - 2(g(u_1)g(u_3)+g(u_2)g(u_3))\mu_2 \le 6(\mu_0^\ast+\delta)\\
2g(u_1) g(u_2) g(u_3) - 2g(u_1)g(u_2)\mu_2 \le 6\delta
\end{array}\right\}\\
&+ \mathbbm{1} \left\{\begin{array}{c}
0 \le 4g(u_1) g(u_2) g(u_3) - 2g(u_1)\mu_1 - 2(g(u_1)g(u_2)+g(u_1)g(u_3))\mu_2 \\
0 \le 2g(u_1) g(u_2) g(u_3) - 2g(u_1)g(u_2)\mu_2 \\
6\mu_0^\ast \le 6g(u_1) g(u_2) g(u_3) - 2(g(u_1)+g(u_2)+g(u_3))\mu_1 - 2(g(u_1)g(u_2)+g(u_1)g(u_3)+g(u_2)g(u_3))\mu_2 \le 6(\mu_0^\ast+\delta)
\end{array}\right\}
\end{aligned}
\end{equation}}

Note that, for a fixed value of \(\mu\) and \(\vec{u}\), only one of the indicator functions on the right-hand side of \eqref{eq:b1_diff} will be active. 
We analyze each case individually to establish \eqref{eq:o_delta_conjecture_mu0}.
\paragraph{Case 1: The first indicator function of the right-hand side of \eqref{eq:b1_diff} is true:}
Suppose the first inequality in the first term on the right-hand side of the summation in \eqref{eq:b1_diff} is true. 
Fix $u_2$,$u_3$,$\mu_1$ and $\mu_2$. 
Denote by $u_{1,\text{min}}$ and $u_{1,\text{max}}$ the minimum and maximum value of $u_1$. Then we have,
\begin{equation}\label{eq:b1_1st_inequality}
\begin{aligned}
&6\mu_0^\ast \le 2g(u_1) g(u_2) g(u_3) - 2(g(u_2)+g(u_3))\mu_1 - 2g(u_2)g(u_3)\mu_2 \le 6(\mu_0^\ast+\delta)\\
\Rightarrow& |2(g(u_{1,\text{min}})-g(u_{1,\text{max}}))g(u_2) g(u_3)|\le 6\delta\\
\Rightarrow& |g(u_{1,\text{min}})-g(u_{1,\text{max}})|\le \frac{3\delta}{g(u_2) g(u_3)} 
\end{aligned}
\end{equation}


Then we have for $u_1$,
\begin{equation}\label{eq:u1_bound}
|u_1-u_1'|  \stackrel{\mbox{\eqref{as:assumption3}}}{\le} \frac{1}{c_3} |g(u_1)-g(u_1')| \stackrel{\mbox{\eqref{eq:b1_1st_inequality}}}{\le} \frac{1}{c_3} \frac{3\delta}{g(u_2) g(u_3)} \stackrel{\mbox{\eqref{as:assumption4}}}{\le} \frac{1}{c_3} \frac{3\delta}{c_4^2}
\end{equation}

We have $Q= \{\vec{u}\colon 0 \le u_1 \le u_2 \le u_3 \le 1\}$.  
Define the set  
\begin{equation}
A_{\delta}
=
\bigl\{
\vec{u}\in Q: \beta_1^{\mu_0^\ast+\delta}(\vec{u})\neq 
\beta_1^{\mu_0^\ast}(\vec{u})
\bigr\}.
\end{equation}
From \eqref{eq:u1_bound}, for any fixed $u_2$ and $u_3$ values in the interval $[0,1]$, the possible $u_1$ value lies within a one-dimensional region whose measure is of order $\delta$. 
Since both $u_2$ and $u_3$ lie in $[0,1]$, the measure of the three-dimensional set $A_{\delta}$ is of order $\delta$.  
The difference  
\begin{equation}
\beta_1^{\mu_0^\ast+\delta}(\vec{u}) 
-
\beta_1^{\mu_0^\ast}(\vec{u}) 
\end{equation}
takes values in $\{-1, 0,1\}$. Hence, 
\begin{equation}
\Bigl\lvert
\beta_1^{\mu_0^\ast+\delta}(\vec{u})
-
\beta_1^{\mu_0^\ast}(\vec{u})
\Bigr\rvert
\le1
\quad
\text{for all } \vec{u}\in Q.
\end{equation}
Therefore, by definition of $A_{\delta}$,  
\begin{equation}
\left\lvert
\int_Q
\Bigl(
\beta_1^{\mu_0^\ast+\delta}(\vec{u})
-
\beta_1^{\mu_0^\ast}(\vec{u})
\Bigr)
d\vec{u}
\right\rvert
\le
\int_{A_{\delta}} 1d\vec{u}
=
\mathrm{measure}(A_{\delta})
=
O(\delta).
\end{equation}
Thus, we have that,
\begin{equation}\label{eq:b1_Odelta}
\int_Q
\Bigl(
\alpha_1^{\mu_0^\ast}(\vec{u})
-
\alpha_1^{\mu_0^\ast+\delta}(\vec{u})
\Bigr)
d\vec{u}
=
O(\delta)
\quad\Longleftrightarrow\quad
\int_Q
\Bigl(
\beta_1^{\mu_0^\ast+\delta}(\vec{u})
-
\beta_1^{\mu_0^\ast}(\vec{u})
\Bigr)
d\vec{u}
=
O(\delta).
\end{equation}
Consider the second factor of the integral of \eqref{eq:o_delta_conjecture_mu0}: 
\begin{equation}
\label{eq:2ndterm_0_delta_conjecture}
\int_Q\biggl[\alpha_1^{\mu_0^\ast}(\vec{u})-\alpha_1^{\mu_0^\ast+\delta}(\vec{u})\biggr]\biggl[6\delta+6\mu_0^\ast-2c_2(\vec{u},\mu_1,\mu_2)\biggr]d\vec{u}.
\end{equation}
From \eqref{eq:b1_diff}, we already know that  
\begin{equation}
\int_Q
\Bigl(
\alpha_1^{\mu_0^\ast+\delta}(\vec{u})
-
\alpha_1^{\mu_0^\ast}(\vec{u})
\Bigr)
d\vec{u}
=
O(\delta).
\end{equation}
If we can establish that the second factor in \eqref{eq:2ndterm_0_delta_conjecture} is also $O(\delta)$, it follows that the entire integral is $o(\delta)$. 

To show this, we begin with \eqref{eq:mu0_min2}. 
Our focus is on the second term within the second integral of the summation in \eqref{eq:mu0_min2}, which is precisely the term of interest, and we can rewrite it as:
{\small \begin{equation}\label{eq:secondterm_lagrangian5}
\begin{aligned}
&6\delta+6\mu_0^\ast-2c_2(\vec{u},\mu_1,\mu_2)\\
=& 6\delta+6\mu_0^\ast -2\biggl(g(u_1) g(u_2) g(u_3) \Bigl(1 + \alpha_2^{\mu}(\vec{u})+\alpha_2^{\mu}(\vec{u})\alpha_3^{\mu}(\vec{u}) \Bigr) \biggr) +2\biggl(g(u_1)+g(u_2) + g(u_3)-g(u_1)\beta_2^\mu(\vec{u}) \biggr)\mu_1 \\
&+2\biggl(g(u_{2})g(u_{3})+\alpha_{2}^{\mu}(\vec{u})g(u_{1})g(u_{3})+ \alpha_{2}^{\mu}(\vec{u})\alpha_{3}^{\mu}(\vec{u})g(u_{1})g(u_{2})\biggr)\mu_2 \\
=& 6\delta+6\mu_0^\ast -2g(u_1) g(u_2) g(u_3) - 2g(u_1) g(u_2) g(u_3)\alpha_2^{\mu}(\vec{u}) -2g(u_1) g(u_2) g(u_3)\alpha_2^{\mu}(\vec{u})\alpha_3^{\mu}(\vec{u})\\
&+2\biggl(g(u_2) + g(u_3) \biggr)\mu_1+2g(u_1)\alpha_2^{\mu}(\vec{u})\mu_1+ 2g(u_{2})g(u_{3})\mu_2 + 2g(u_{1})g(u_{3})\alpha_{2}^{\mu}(\vec{u})\mu_2 +2g(u_{1})g(u_{2})\alpha_{2}^{\mu}(\vec{u})\alpha_{3}^{\mu}(\vec{u})\mu_2\\
=& 6\delta + 6\mu_0^\ast -2g(u_1) g(u_2) g(u_3) + 2\bigg(g(u_2) + g(u_3)\bigg)\mu_1 + 2g(u_{2})g(u_{3})\mu_2 \\
&+ 2\alpha_{2}^{\mu}(\vec{u})\bigg(-g(u_1)g(u_2)g(u_3) + g(u_1)\mu_1+g(u_1)g(u_3)\mu_2\bigg)\\
&+ 2\alpha_{2}^{\mu}(\vec{u})\alpha_{3}^{\mu}(\vec{u})\bigg(-g(u_1)g(u_2)g(u_3)+g(u_1)g(u_2)\mu_2\bigg)
\end{aligned}
\end{equation}}

Now we compare this term with the first indicator function in \eqref{eq:b1_diff},
\begin{equation}\label{eq:b1_diff_1st_ineq}
\begin{aligned}
&\mathbbm{1}\left\{\begin{array}{c}
6\mu_0^\ast \le 2g(u_1) g(u_2) g(u_3) - 2(g(u_2)+g(u_3))\mu_1 - 2g(u_2)g(u_3)\mu_2 \le 6(\mu_0^\ast+\delta)\\
2g(u_1) g(u_2) g(u_3) - 2g(u_1)\mu_1 - 2g(u_1)g(u_3)\mu_2 \le 6\delta\\
4g(u_1) g(u_2) g(u_3) - 2g(u_1)\mu_1 - 2(g(u_1)g(u_2) + g(u_1)g(u_3))\mu_2 \le 6\delta
\end{array}\right\}\\
=& \mathbbm{1}\left\{\begin{array}{c}
-6\delta \le 6\mu_0^\ast - 2g(u_1) g(u_2) g(u_3) + 2(g(u_2)+g(u_3))\mu_1 + 2g(u_2)g(u_3)\mu_2 \le 0\\
2g(u_1) g(u_2) g(u_3) - 2g(u_1)\mu_1 - 2g(u_1)g(u_3)\mu_2 \le 6\delta\\
4g(u_1) g(u_2) g(u_3) - 2g(u_1)\mu_1 - 2(g(u_1)g(u_2) + g(u_1)g(u_3))\mu_2 \le 6\delta
\end{array}\right\}\\
=&\mathbbm{1}\left\{\begin{array}{c}
-6\delta \le F_1(\vec{u},\mu_0^\ast,\mu_1,\mu_2) \le 0\\
F_2(\vec{u},\mu_1,\mu_2) \le 6\delta\\
F_3(\vec{u},\mu_1,\mu_2) \le 6\delta
\end{array}\right\}\\
=&\mathbbm{1}\left\{\begin{array}{c}
0 \le F_1(\vec{u},\mu_0^\ast,\mu_1,\mu_2)+6\delta \le 6\delta\\
F_2(\vec{u},\mu_1,\mu_2) \le 6\delta\\
F_3(\vec{u},\mu_1,\mu_2) \le 6\delta
\end{array}\right\}
\end{aligned}
\end{equation}
where 
\begin{equation}\label{eq:f123}
\begin{aligned}
F_1(\vec{u},\mu_0^\ast,\mu_1,\mu_2)&=6\mu_0^\ast-2g(u_1) g(u_2) g(u_3) + 2(g(u_2)+g(u_3))\mu_1 + 2g(u_2)g(u_3)\mu_2\\
F_2(\vec{u},\mu_1,\mu_2)&=2g(u_1) g(u_2) g(u_3) - 2g(u_1)\mu_1 - 2g(u_1)g(u_3)\mu_2\\
F_3(\vec{u},\mu_1,\mu_2)&=4g(u_1) g(u_2) g(u_3) - 2g(u_1)\mu_1 - 2(g(u_1)g(u_2) + g(u_1)g(u_3))\mu_2.
\end{aligned}
\end{equation}

Now using \eqref{eq:b1_diff_1st_ineq}, \eqref{eq:secondterm_lagrangian5} can be rewritten as,
\begin{equation}\label{eq:secondterm_lagrangian5_f123}
\begin{aligned}
&6\delta+F_1(\vec{u},\mu_0^\ast,\mu_1,\mu_2)+2\alpha_2^{\mu}(\vec{u})\bigg(-\frac{1}{2}F_2(\vec{u},\mu_1,\mu_2) \bigg) \\
&+2\alpha_2^{\mu}(\vec{u})\alpha_3^{\mu}(\vec{u})\bigg(-\frac{1}{4}F_3(\vec{u},\mu_1,\mu_2)-\frac{1}{2}g(u_1)\mu_1+\frac{1}{2}g(u_1)g(u_2)\mu_2-\frac{1}{2}g(u_1)g(u_3)\mu_2 \bigg)\\
=& 6\delta+F_1(\vec{u},\mu_0^\ast,\mu_1,\mu_2)-\alpha_2^{\mu}(\vec{u})F_2(\vec{u},\mu_1,\mu_2)-\frac{1}{2}\alpha_2^{\mu}(\vec{u})\alpha_3^{\mu}(\vec{u})F_3(\vec{u},\mu_1,\mu_2)\\
&-\alpha_2^{\mu}(\vec{u})\alpha_3^{\mu}(\vec{u})g(u_1)\mu_1 + \alpha_2^{\mu}(\vec{u})\alpha_3^{\mu}(\vec{u})g(u_1)g(u_2)\mu_2-g(u_1)g(u_3)\alpha_2^{\mu}(\vec{u})\alpha_3^{\mu}(\vec{u})\mu_2 \\
=& \bigg(6\delta+F_1(\vec{u},\mu_0^\ast,\mu_1,\mu_2)\bigg) - \alpha_2^{\mu}(\vec{u})F_2(\vec{u},\mu_1,\mu_2)-\frac{1}{2}\alpha_2^{\mu}(\vec{u})\alpha_3^{\mu}(\vec{u})F_3(\vec{u},\mu_1,\mu_2)+\alpha_2^{\mu}(\vec{u})\alpha_3^{\mu}(\vec{u})g(u_1)\mu_1\\
&+\alpha_2^{\mu}(\vec{u})\alpha_3^{\mu}(\vec{u})g(u_1)\mu_2\bigg(g(u_2)-g(u_3)\bigg)
\end{aligned}
\end{equation}
Now consider various combinations of $\alpha_2^{\mu}(\vec{u})$.

\begin{enumerate}
\item If $\alpha_2^{\mu}(\vec{u})=0$, \eqref{eq:secondterm_lagrangian5_f123} becomes,
\begin{equation}\label{eq:6dplusf1}
\bigg(6\delta+F_1(\vec{u},\mu_0^\ast,\mu_1,\mu_2)\bigg)
\end{equation}
Also using \eqref{eq:b1_diff_1st_ineq}, we have 
\begin{equation}\label{eq:6dplusf1_Od}
\bigg(6\delta+F_1(\vec{u},\mu_0^\ast,\mu_1,\mu_2)\bigg)=O(\delta).
\end{equation}
Recall \eqref{eq:o_delta_conjecture_mu0}. We have,
\begin{equation}\label{eq:case1_o_delta_conjecture_proven}
\begin{aligned}
&\int_Q\biggl[\alpha_1^{\mu_0^\ast}(\vec{u})-\alpha_1^{\mu_0^\ast+\delta}(\vec{u})\biggr]\biggl[6\delta+6\mu_0^\ast-2c_2(\vec{u},\mu_1,\mu_2)\biggr]d\vec{u}\\
\stackrel{\mbox{\eqref{eq:6dplusf1}}}{=}& \int_Q\biggl[\alpha_1^{\mu_0^\ast}(\vec{u})-\alpha_1^{\mu_0^\ast+\delta}(\vec{u})\biggr]\biggl[6\delta+F_1(\vec{u},\mu_0^\ast,\mu_1,\mu_2)\biggr]d\vec{u}\\
\le & \|6\delta+F_1(\vec{u},\mu_0^\ast,\mu_1,\mu_2)\|_{\infty} \int_Q\bigg(\alpha_1^{\mu_0^\ast}(\vec{u})-\alpha_1^{\mu_0^\ast+\delta}(\vec{u})\bigg)d\vec{u}\\
\stackrel{\mbox{\eqref{eq:b1_Odelta},\eqref{eq:6dplusf1_Od}}}{=}& O(\delta^2)\\
=& o(\delta) \text{ in the limit } \delta \rightarrow 0.
\end{aligned}
\end{equation}
Thus, equation \eqref{eq:o_delta_conjecture_mu0} holds.
The proof then follows from proposition \ref{prop:o_delta_condition_mu_0}.

\item If $\alpha_2^{\mu}(\vec{u})=1$, we have $\beta_2^{\mu}(\vec{u})=0$.
Recall the expression of $\beta_2^{\mu}(\vec{u})$ in \eqref{eq:beta_mu_u}.
\begin{equation}\label{eq:b2_f23}
\begin{aligned}
\beta_2^{\mu}(\vec{u}) &= \mathbbm{1}\biggl\{\mu,\vec{u}:
\left\{
\begin{array}{c}
2g(u_1) g(u_2) g(u_3) - 2g(u_1) \mu_1 - 2g(u_1)g(u_3)\mu_2 \le 0 \\
4g(u_1) g(u_2) g(u_3) - 2g(u_1) \mu_1 - 2(g(u_1)g(u_3) + g(u_1)g(u_2))\mu_2 \le 0
\end{array}
\right\}
\biggr\}\\
& \stackrel{\mbox{\eqref{eq:f123}}}{=}\mathbbm{1}\biggl\{\mu,\vec{u}:
\left\{
\begin{array}{c}
F_2(\vec{u},\mu_1,\mu_2) \le 0 \\
F_3(\vec{u},\mu_1,\mu_2) \le 0
\end{array}
\right\}
\biggr\}.
\end{aligned}
\end{equation}
If $\beta_2^{\mu}(\vec{u})=0$, we have either $F_2(\vec{u},\mu_1,\mu_2) > 0$ or $F_3(\vec{u},\mu_1,\mu_2) > 0$ where $F_2(\vec{u},\mu_1,\mu_2)$ and $F_3(\vec{u},\mu_1,\mu_2)$ are defined in \eqref{eq:f123}. 

\begin{enumerate}
    \item If $F_2(\vec{u},\mu_1,\mu_2)>0$,
fix $u_1$,$u_3$,$\mu_1$ and $\mu_2$. 
Denote by $u_{2,\text{min}}$ and $u_{2,\text{max}}$ the minimum and maximum value of $u_2$. 
Then we have,
\begin{equation}
\begin{aligned}
& 0 < F_2(\vec{u},\mu_1,\mu_2) \stackrel{\mbox{\eqref{eq:b1_diff}}}{\le} 6\delta \\
\Rightarrow & 0 < 2g(u_1) g(u_2) g(u_3) - 2g(u_1) \mu_1 - 2g(u_1)g(u_3)\mu_2 \le 6 \delta \\ 
\Rightarrow &  |2(g(u_{2,\text{min}})-g(u_{2,\text{max}})) g(u_1) g(u_3)| < 6\delta \\ 
\Rightarrow & |g(u_{2,\text{min}})-g(u_{2,\text{max}})| <\frac{3\delta}{g(u_1) g(u_3)}
\end{aligned}
\end{equation}
Then we have for $u_2$,
\begin{equation}\label{eq:u2_bound}
|u_2-u_2'|  \stackrel{\mbox{\eqref{as:assumption3}}}{\le} \frac{1}{c_3} |g(u_2)-g(u_2')| \stackrel{\mbox{\eqref{eq:b1_1st_inequality}}}{\le} \frac{1}{c_3} \frac{3\delta}{g(u_1) g(u_3)} \stackrel{\mbox{\eqref{as:assumption4}}}{\le} \frac{1}{c_3} \frac{3\delta}{c_4^2}
\end{equation}

So we have that both $u_1$ and $u_2$ are $O(\delta)$ from \eqref{eq:u1_bound} and \eqref{eq:u2_bound} respectively. 
Recall \eqref{eq:o_delta_conjecture_mu0}
\begin{equation}
\int_Q\biggl[\alpha_1^{\mu_0^\ast}(\vec{u})-\alpha_1^{\mu_0^\ast+\delta}(\vec{u})\biggr]\biggl[6\delta + 6\mu_0^\ast-2c_2(\vec{u},\mu_1,\mu_2)\biggr]d\vec{u}=o(\delta).
\end{equation}
Define
\[
A_\delta \;:=\; \bigl\{\vec{u}\in Q:\,\alpha_1^{\mu_0^\ast}(\vec{u})\neq \alpha_1^{\mu_0^\ast+\delta}(\vec{u})\bigr\}.
\]
Since $\alpha_1^{\mu_0^\ast}$ and $\alpha_1^{\mu_0^\ast+\delta}$ are indicator functions, the integrand in \eqref{eq:o_delta_conjecture_mu0} vanishes outside $A_\delta$, and we can rewrite
\begin{equation}\label{eq:int_on_A_delta}
\begin{aligned}
&\int_Q\biggl[\alpha_1^{\mu_0^\ast}(\vec{u})-\alpha_1^{\mu_0^\ast+\delta}(\vec{u})\biggr]\biggl[6\delta + 6\mu_0^\ast-2c_2(\vec{u},\mu_1,\mu_2)\biggr]d\vec{u}\\
=&\int_{A_\delta}\biggl[\alpha_1^{\mu_0^\ast}(\vec{u})-\alpha_1^{\mu_0^\ast+\delta}(\vec{u})\biggr]\biggl[6\delta + 6\mu_0^\ast-2c_2(\vec{u},\mu_1,\mu_2)\biggr]d\vec{u}.
\end{aligned}
\end{equation}
From the inequalities established in \eqref{eq:u1_bound} and \eqref{eq:u2_bound}, we have that for $\vec{u}\in A_\delta$ the coordinates $u_1$ and $u_2$ vary over an interval of length $O(\delta)$:
there exist constants $K_1,K_2>0$ such that, for all $\vec{u},\vec{u}'\in A_\delta$ with the same $u_3$,
\begin{equation}\label{eq:u1u2_Odelta}
|u_1-u_1'|\le K_1\delta,
\qquad
|u_2-u_2'|\le K_2\delta.
\end{equation}
Since $Q\subset[0,1]^3$, the $u_3$-coordinate is always in $[0,1]$, so the Lebesgue measure of $A_\delta$ satisfies
\begin{equation}\label{eq:measure_Adelta}
|A_\delta| \;\le\; K\,\delta^2,
\qquad\text{for some constant }K>0 \text{ independent of }\delta.
\end{equation}

Next, we bound the factor $6\delta + 6\mu_0^\ast-2c_2(\vec{u},\mu_1,\mu_2)$ uniformly in $\vec{u}\in Q$. 
By Assumption~\ref{as:assumption3}, $g$ is continuous on $[0,1]$, hence bounded. 
Since $c_2(\vec{u},\mu_1,\mu_2)$ is a finite linear combination of products of $g(u_1),g(u_2),g(u_3)$ with fixed coefficients, it is also continuous in $\vec{u}$ and therefore bounded on the compact set $Q$. 
Thus there exists $B>0$ such that
\[
|c_2(\vec{u},\mu_1,\mu_2)|\le B 
\qquad\text{for all }\vec{u}\in Q.
\]
Fix $\delta_0>0$ and consider $0<\delta\le\delta_0$. Then for all $\vec{u}\in Q$,
\begin{equation}\label{eq:integrand_bound}
\begin{aligned}
\bigl|6\delta + 6\mu_0^\ast - 2c_2(\vec{u},\mu_1,\mu_2)\bigr|
&\le |6\delta| + 6|\mu_0^\ast| + 2|c_2(\vec{u},\mu_1,\mu_2)| \\
&\le 6\delta_0 + 6|\mu_0^\ast| + 2B \;=:\; M.
\end{aligned}
\end{equation}
In particular, $M$ is a finite constant that does not depend on $\delta$ or $\vec{u}$.

Combining \eqref{eq:int_on_A_delta}, \eqref{eq:measure_Adelta}, and \eqref{eq:integrand_bound}, and using that $\bigl|\alpha_1^{\mu_0^\ast}(\vec{u})-\alpha_1^{\mu_0^\ast+\delta}(\vec{u})\bigr|\le 1$, we obtain
\begin{equation*}
\begin{aligned}
&\biggl|\int_Q\bigl[\alpha_1^{\mu_0^\ast}(\vec{u})-\alpha_1^{\mu_0^\ast+\delta}(\vec{u})\bigr]\bigl[6\delta + 6\mu_0^\ast-2c_2(\vec{u},\mu_1,\mu_2)\bigr]d\vec{u}\biggr|\\
&= \biggl|\int_{A_\delta}\bigl[\alpha_1^{\mu_0^\ast}(\vec{u})-\alpha_1^{\mu_0^\ast+\delta}(\vec{u})\bigr]\bigl[6\delta + 6\mu_0^\ast-2c_2(\vec{u},\mu_1,\mu_2)\bigr]d\vec{u}\biggr| \\
&\le \int_{A_\delta} M\,d\vec{u}
\;\le\; M\,|A_\delta|
\;\le\; MK\,\delta^2.
\end{aligned}
\end{equation*}
Hence the integral in \eqref{eq:o_delta_conjecture_mu0} is $O(\delta^2)$ and therefore $o(\delta)$ as $\delta\to 0$.
This verifies \eqref{eq:o_delta_conjecture_mu0}, and by Proposition~\ref{prop:o_delta_condition_mu_0} we conclude that \eqref{eq:mu_0_optimality} holds.

\item If $F_3(\vec{u},\mu_1,\mu_2)>0$, we have
\begin{equation}\label{eq:case1_b}
\begin{aligned}
& 0 < F_3(\vec{u},\mu_1,\mu_2) \stackrel{\mbox{\eqref{eq:b1_diff}}}{\le} 6\delta \\
\Leftrightarrow & 0 < 4g(u_1) g(u_2) g(u_3) - 2g(u_1) \mu_1 - 2(g(u_1)g(u_3) + g(u_1)g(u_2))\mu_2 \le 6\delta \\
\Leftrightarrow &  0 < g(u_2) g(u_3) - \frac{\mu_1}{2} -\frac{(g(u_3) + g(u_2))\mu_2}{2} \le \frac{3\delta}{2g(u_1)}\\
\Leftrightarrow & 0 < \bigg(g(u_2)-\frac{\mu_2}{2}\bigg)\bigg(g(u_3)-\frac{\mu_2}{2}\bigg) - \frac{\mu_1}{2} - \frac{\mu_2^2}{4} \le \frac{3\delta}{2g(u_1)}\\
\end{aligned}
\end{equation}
We have from Lemma \ref{lem:gdot} that $g(u_2)\ge g(u_3)$ as $u_2\le u_3$.
Also, from above we have,
\begin{equation}
\bigg(g(u_2)-\frac{\mu_2}{2}\bigg)\bigg(g(u_3)-\frac{\mu_2}{2}\bigg) \ge 0
\end{equation}
Without loss of generality, we assume, 
\begin{equation}\label{eq:case1b_g1>g2}
g(u_2) \ge g(u_3) > \frac{\mu_2}{2}.
\end{equation}
Thu,s we have,
\begin{equation}\label{eq:case1b_g2>}
\begin{aligned}
&\frac{\mu_1}{2} + \frac{\mu_2^2}{4}  \stackrel{\mbox{\eqref{eq:case1_b}}}{\le} \bigg(g(u_2)-\frac{\mu_2}{2}\bigg)\bigg(g(u_3)-\frac{\mu_2}{2}\bigg) \stackrel{\mbox{\eqref{eq:case1b_g1>g2}}}{\le} \bigg(g(u_2)-\frac{\mu_2}{2}\bigg)^2 \\
\Rightarrow & \bigg(g(u_2)-\frac{\mu_2}{2}\bigg) \ge \sqrt{\frac{\mu_1}{2} + \frac{\mu_2^2}{4}}
\end{aligned}
\end{equation}
Fix $u_1$,$u_2$,$\mu_1$ and $\mu_2$. 
Denote by $u_{3,\text{min}}$ and $u_{3,\text{max}}$ the minimum and maximum value of $u_3$. 
Then we have using \eqref{eq:case1_b},
\begin{equation}\label{eq:g3_bound}
\begin{aligned}
& \bigg(g(u_2)-\frac{\mu_2}{2}\bigg)|g(u_{3,\text{min}})-g(u_{3,\text{max}})| <\frac{3\delta}{2g(u_1)} \\
\Rightarrow & |g(u_{3,\text{min}})-g(u_{3,\text{max}})| <\frac{3\delta}{2g(u_1)} \frac{1}{\bigg(g(u_2)-\frac{\mu_2}{2}\bigg)} \stackrel{\mbox{\eqref{eq:case1b_g2>}}}{<} \frac{3\delta}{2g(u_1)} \frac{1}{\sqrt{\frac{\mu_1}{2} + \frac{\mu_2^2}{4}}}
\end{aligned}
\end{equation}
Then we have for $u_3$,
\begin{equation}\label{eq:u3_bound}
|u_3-u_3'|  \stackrel{\mbox{\eqref{as:assumption3}}}{\le} \frac{1}{c_3} |g(u_3)-g(u_3')| \stackrel{\mbox{\eqref{eq:g3_bound}}}{\le} \frac{1}{c_3} \frac{3\delta}{2g(u_1)} \frac{1}{\sqrt{\frac{\mu_1}{2} + \frac{\mu_2^2}{4}}} \stackrel{\mbox{\eqref{as:assumption4}}}{\le} \frac{1}{c_3} \frac{3\delta}{2c_4} \frac{1}{\sqrt{\frac{\mu_1}{2} + \frac{\mu_2^2}{4}}}
\end{equation}

So we have that both $u_1$ and $u_3$ are $O(\delta)$ from \eqref{eq:u1_bound} and \eqref{eq:u3_bound} respectively. 
Thus it follows that \eqref{eq:o_delta_conjecture_mu0} is true and, consequently, \eqref{eq:mu_0_optimality} is true using proposition \ref{prop:o_delta_condition_mu_0}.
We now verify condition~\eqref{eq:o_delta_conjecture_mu0}.
Recall that
\begin{equation}
\int_Q\biggl[\alpha_1^{\mu_0^\ast}(\vec{u})-\alpha_1^{\mu_0^\ast+\delta}(\vec{u})\biggr]\biggl[6\delta + 6\mu_0^\ast-2c_2(\vec{u},\mu_1,\mu_2)\biggr]d\vec{u}=o(\delta).
\end{equation}
Define
\[
A_\delta \;:=\; \bigl\{\vec{u}\in Q:\,\alpha_1^{\mu_0^\ast}(\vec{u})\neq \alpha_1^{\mu_0^\ast+\delta}(\vec{u})\bigr\}.
\]
Since $\alpha_1^{\mu_0^\ast}$ and $\alpha_1^{\mu_0^\ast+\delta}$ are indicator functions, the integrand in \eqref{eq:o_delta_conjecture_mu0} vanishes outside $A_\delta$, and we can rewrite
\begin{equation}\label{eq:int_on_A_delta_2nditem}
\begin{aligned}
&\int_Q\biggl[\alpha_1^{\mu_0^\ast}(\vec{u})-\alpha_1^{\mu_0^\ast+\delta}(\vec{u})\biggr]\biggl[6\delta + 6\mu_0^\ast-2c_2(\vec{u},\mu_1,\mu_2)\biggr]d\vec{u}\\
=&\int_{A_\delta}\biggl[\alpha_1^{\mu_0^\ast}(\vec{u})-\alpha_1^{\mu_0^\ast+\delta}(\vec{u})\biggr]\biggl[6\delta + 6\mu_0^\ast-2c_2(\vec{u},\mu_1,\mu_2)\biggr]d\vec{u}.
\end{aligned}
\end{equation}

From the inequalities established in \eqref{eq:u1_bound} and \eqref{eq:u3_bound}, we have that for $\vec{u}\in A_\delta$ and fixed $u_2$ the coordinates $u_1$ and $u_3$ each vary over an interval of length $O(\delta)$:
there exist constants $K_1,K_3>0$ such that, for all $\vec{u},\vec{u}'\in A_\delta$ with the same $u_2$,
\begin{equation}\label{eq:u1u3_Odelta}
|u_1-u_1'|\le K_1\delta,
\qquad
|u_3-u_3'|\le K_3\delta.
\end{equation}
Since $Q\subset[0,1]^3$, the $u_2$-coordinate is always in $[0,1]$, so the Lebesgue measure of $A_\delta$ satisfies
\begin{equation}\label{eq:measure_Adelta_2nditem}
|A_\delta| \;\le\; K\,\delta^2,
\qquad\text{for some constant }K>0 \text{ independent of }\delta.
\end{equation}

Next, we bound the factor $6\delta + 6\mu_0^\ast-2c_2(\vec{u},\mu_1,\mu_2)$ uniformly in $\vec{u}\in Q$. 
By Assumption~\ref{as:assumption3}, $g$ is continuous on $[0,1]$, hence bounded. 
Since $c_2(\vec{u},\mu_1,\mu_2)$ is a finite linear combination of products of $g(u_1),g(u_2),g(u_3)$ with fixed coefficients, it is also continuous in $\vec{u}$ and therefore bounded on the compact set $Q$. 
Thus there exists $B>0$ such that
\[
|c_2(\vec{u},\mu_1,\mu_2)|\le B 
\qquad\text{for all }\vec{u}\in Q.
\]
Fix $\delta_0>0$ and consider $0<\delta\le\delta_0$. Then for all $\vec{u}\in Q$,
\begin{equation}\label{eq:integrand_bound_2nditem}
\begin{aligned}
\bigl|6\delta + 6\mu_0^\ast - 2c_2(\vec{u},\mu_1,\mu_2)\bigr|
&\le |6\delta| + 6|\mu_0^\ast| + 2|c_2(\vec{u},\mu_1,\mu_2)| \\
&\le 6\delta_0 + 6|\mu_0^\ast| + 2B \;=:\; M.
\end{aligned}
\end{equation}
In particular, $M$ is a finite constant that does not depend on $\delta$ or $\vec{u}$.

Combining \eqref{eq:int_on_A_delta_2nditem}, \eqref{eq:measure_Adelta_2nditem}, and \eqref{eq:integrand_bound_2nditem}, and using that $\bigl|\alpha_1^{\mu_0^\ast}(\vec{u})-\alpha_1^{\mu_0^\ast+\delta}(\vec{u})\bigr|\le 1$, we obtain
\begin{equation*}
\begin{aligned}
&\biggl|\int_Q\bigl[\alpha_1^{\mu_0^\ast}(\vec{u})-\alpha_1^{\mu_0^\ast+\delta}(\vec{u})\bigr]\bigl[6\delta + 6\mu_0^\ast-2c_2(\vec{u},\mu_1,\mu_2)\bigr]d\vec{u}\biggr|\\
&= \biggl|\int_{A_\delta}\bigl[\alpha_1^{\mu_0^\ast}(\vec{u})-\alpha_1^{\mu_0^\ast+\delta}(\vec{u})\bigr]\bigl[6\delta + 6\mu_0^\ast-2c_2(\vec{u},\mu_1,\mu_2)\bigr]d\vec{u}\biggr| \\
&\le \int_{A_\delta} M\,d\vec{u}
\;\le\; M\,|A_\delta|
\;\le\; MK\,\delta^2.
\end{aligned}
\end{equation*}
Hence the integral in \eqref{eq:o_delta_conjecture_mu0} is $O(\delta^2)$ and therefore $o(\delta)$ as $\delta\to 0$.
This verifies \eqref{eq:o_delta_conjecture_mu0}, and by Proposition~\ref{prop:o_delta_condition_mu_0} we conclude that \eqref{eq:mu_0_optimality} holds.

\end{enumerate}
\end{enumerate}

\paragraph{Case 2: The second indicator function of the right-hand side of \eqref{eq:b1_diff} is true:} In this case, we have the second indicator function of \eqref{eq:b1_diff} as true. This means the following inequalities are true:
\begin{equation}\label{eq:b1_diff_2nd_indicator}
\begin{aligned}
& 0 \le 2g(u_1) g(u_2) g(u_3) - 2g(u_1)\mu_1 - 2g(u_1)g(u_3)\mu_2  \\
& 6\mu_0^\ast \le 4g(u_1) g(u_2) g(u_3) - 2(g(u_1)+g(u_2)+g(u_3))\mu_1 - 2(g(u_1)g(u_3)+g(u_2)g(u_3))\mu_2 \le 6(\mu_0^\ast+\delta)\\
& 2g(u_1) g(u_2) g(u_3) - 2g(u_1)g(u_2)\mu_2 \le 6\delta
\end{aligned}
\end{equation}
From \eqref{eq:beta_mu_u}, we have $\beta_2^{\mu}(\vec{u})$ as,
\begin{equation}
\beta_2^{\mu}(\vec{u})=\mathbbm{1}\biggl\{\mu,\vec{u}:
\left\{
\begin{array}{c}
2g(u_1) g(u_2) g(u_3) - 2g(u_1) \mu_1 - 2g(u_1)g(u_3)\mu_2 \le 0 \\
4g(u_1) g(u_2) g(u_3) - 2g(u_1) \mu_1 - 2(g(u_1)g(u_3) + g(u_1)g(u_2))\mu_2 \le 0
\end{array}
\right\}
\biggr\}
\end{equation}
Due to the first inequality of \eqref{eq:b1_diff_2nd_indicator}, it can be seen that the first inequality of the above indicator function does not hold, which indicates that for this case, we have,
\begin{equation}\label{eq:beta_2_0_case2}
\beta_2^{\mu}(\vec{u})=0 \Leftrightarrow \alpha_2^{\mu}(\vec{u})=1
\end{equation}
From \eqref{eq:b1_diff_2nd_indicator} we have the second inequality as,
\begin{equation}\label{eq:b1_diff_2nd_ineq_case2}
\begin{aligned}
& 6\mu_0^\ast \le 4g(u_1) g(u_2) g(u_3) - 2(g(u_1)+g(u_2)+g(u_3))\mu_1 - 2(g(u_1)g(u_3)+g(u_2)g(u_3))\mu_2 \le 6(\mu_0^\ast+\delta) \\
\Rightarrow & \frac{6\mu_0^\ast}{4g(u_3)} \le g(u_1) g(u_2)-  \frac{\mu_1}{2g(u_3)}(g(u_1)+g(u_2))-\frac{\mu_1}{2}-\frac{1}{2}(g(u_1)+g(u_2))\mu_2 \le \frac{3(\mu_0^\ast+\delta)}{2g(u_3)}\\
\Rightarrow & \frac{6\mu_0^\ast}{4g(u_3)} \le g(u_1) g(u_2)-  (g(u_1)+g(u_2))\bigg(\frac{\mu_1}{2g(u_3)}+\frac{\mu_2}{2}\bigg)-\frac{\mu_1}{2} \le \frac{3(\mu_0^\ast+\delta)}{2g(u_3)}\\
\Rightarrow & \frac{3\mu_0^\ast}{2g(u_3)} \le \bigg(g(u_1)-\frac{\mu_1}{2g(u_3)}-\frac{\mu_2}{2}\bigg)\bigg(g(u_2)-\frac{\mu_1}{2g(u_3)}-\frac{\mu_2}{2}\bigg) -\frac{\mu_1}{2} - \bigg(\frac{\mu_1}{2g(u_3)}+\frac{\mu_2}{2}\bigg)^2   \le \frac{3(\mu_0^\ast+\delta)}{2g(u_3)}
\end{aligned}
\end{equation}
We must have,
\begin{equation}\label{eq:case2_g1>g2}
g(u_1) \ge g(u_2) \ge \frac{\mu_1}{2g(u_3)}+\frac{\mu_2}{2} \text{ or } \frac{\mu_1}{2g(u_3)}+\frac{\mu_2}{2} \ge g(u_1) \ge g(u_2) 
\end{equation}
Without loss of generality, assume the first inequality of \eqref{eq:case2_g1>g2} is true. Then we have from \eqref{eq:b1_diff_2nd_ineq_case2},
\begin{equation}\label{eq:case2_g1_sq}
\begin{aligned}
 &  \bigg(g(u_1)-\frac{\mu_1}{2g(u_3)}-\frac{\mu_2}{2}\bigg)\bigg(g(u_2)-\frac{\mu_1}{2g(u_3)}-\frac{\mu_2}{2}\bigg) \ge \frac{3\mu_0^\ast}{2g(u_3)} +\frac{\mu_1}{2} + \bigg(\frac{\mu_1}{2g(u_3)}+\frac{\mu_2}{2}\bigg)^2 \\
\Rightarrow & \bigg(g(u_1)-\frac{\mu_1}{2g(u_3)}-\frac{\mu_2}{2}\bigg)^2 \ge \frac{3\mu_0^\ast}{2g(u_3)} +\frac{\mu_1}{2} + \bigg(\frac{\mu_1}{2g(u_3)}+\frac{\mu_2}{2}\bigg)^2 \\
\Rightarrow & \bigg(g(u_1)-\frac{\mu_1}{2g(u_3)}-\frac{\mu_2}{2}\bigg) \ge \sqrt{\frac{3\mu_0^\ast}{2g(u_3)} +\frac{\mu_1}{2} + \bigg(\frac{\mu_1}{2g(u_3)}+\frac{\mu_2}{2}\bigg)^2} 
\end{aligned}
\end{equation}
Fix $u_1$,$u_3$,$\mu_1$ and $\mu_2$. 
Denote by $u_{2,\text{min}}$ and $u_{2,\text{max}}$ the minimum and maximum value of $u_2$. 
Then we have using \eqref{eq:b1_diff_2nd_ineq_case2},
{\small
\begin{equation}\label{eq:g2_bound_case2}
\begin{aligned}
& \bigg(g(u_1)-\frac{\mu_1}{2g(u_3)}-\frac{\mu_2}{2}\bigg)|g(u_{2,\text{min}})-g(u_{2,\text{max}})| <\frac{3\delta}{2g(u_3)} \\
\Rightarrow & |g(u_{2,\text{min}})-g(u_{2,\text{max}})| <\frac{3\delta}{2g(u_3)} \frac{1}{\bigg|\bigg(g(u_2)-\frac{\mu_1}{2g(u_3)}-\frac{\mu_2}{2}\bigg)\bigg|} \stackrel{\mbox{\eqref{eq:case2_g1_sq}}}{\le} \frac{3\delta}{2g(u_3)} \frac{1}{\sqrt{\frac{3\mu_0^\ast}{2g(u_3)} +\frac{\mu_1}{2} + \bigg(\frac{\mu_1}{2g(u_3)}+\frac{\mu_2}{2}\bigg)^2}}
\end{aligned}
\end{equation}}
Thus we have $u_2=O(\delta)$.

Next, we try to show that the second integrand in the second term of \eqref{eq:mu0_min2} is $O(\delta)$. This term can be rewritten as in \eqref{eq:secondterm_lagrangian5} as,
\begin{equation}\label{eq:secondterm_lagrangian5_case2}
\begin{aligned}
& 6\delta + 6\mu_0^\ast -2g(u_1) g(u_2) g(u_3) + 2\bigg(g(u_2) + g(u_3)\bigg)\mu_1 + 2g(u_{2})g(u_{3})\mu_2 \\
&+ 2\alpha_{2}^{\mu}(\vec{u})\bigg(-g(u_1)g(u_2)g(u_3) + g(u_1)\mu_1+g(u_1)g(u_3)\mu_2\bigg) \\
& + 2\alpha_{2}^{\mu}(\vec{u})\alpha_{3}^{\mu}(\vec{u})\bigg(-g(u_1)g(u_2)g(u_3)+g(u_1)g(u_2)\mu_2\bigg)\\
\stackrel{\mbox{\eqref{eq:beta_2_0_case2}}}{=}& 6\delta + 6\mu_0^\ast -2g(u_1) g(u_2) g(u_3) + 2\bigg(g(u_2) + g(u_3)\bigg)\mu_1 + 2g(u_{2})g(u_{3})\mu_2 \\
& -2g(u_1)g(u_2)g(u_3) + 2g(u_1)\mu_1 + 2g(u_1)g(u_3)\mu_2 + 2\alpha_{3}^{\mu}(\vec{u})\bigg(-g(u_1)g(u_2)g(u_3)+g(u_1)g(u_2)\mu_2\bigg)\\
=& 6\delta + 6\mu_0^\ast -4g(u_1) g(u_2) g(u_3) + 2\bigg(g(u_1)+g(u_2) + g(u_3)\bigg)\mu_1 + 2\bigg(g(u_1)g(u_3)+g(u_{2})g(u_{3})\bigg)\mu_2 \\
& + 2g(u_1)g(u_2)\alpha_{3}^{\mu}(\vec{u})\bigg(\mu_2-g(u_3)\bigg)\\
=& F_4(\vec{u},\mu_0^\ast,\mu_1,\mu_2), \text{ say.}
\end{aligned}
\end{equation}
We denote,
{\small
\begin{equation}\label{eq:f5}
F_5(\vec{u},\mu_0^\ast,\mu_1,\mu_2)= 6\delta + 6\mu_0^\ast -4g(u_1) g(u_2) g(u_3) + 2\bigg(g(u_1)+g(u_2) + g(u_3)\bigg)\mu_1 + 2\bigg(g(u_1)g(u_3)+g(u_{2})g(u_{3})\bigg)\mu_2
\end{equation}}
Note using upper inequality of the second equation of \eqref{eq:b1_diff_2nd_indicator}, we have 
\begin{equation}\label{eq:f5>0}
F_5(\vec{u},\mu_0^\ast,\mu_1,\mu_2)\ge 0. 
\end{equation}

Now consider various values of $\alpha_{3}^{\mu}(\vec{u})$.
\begin{enumerate}
\item If $\alpha_{3}^{\mu}(\vec{u})=0$, then we have
\begin{equation}\label{eq:f4=f5}
F_4(\vec{u},\mu_0^\ast,\mu_1,\mu_2)=F_5(\vec{u},\mu_0^\ast,\mu_1,\mu_2).
\end{equation}
Then using the lower inequality of the second equation of \eqref{eq:b1_diff_2nd_indicator},
\begin{equation}\label{eq:f4_bound1}
\begin{aligned}
& \begin{multlined}[t] 
    0 \stackrel{\mbox{\eqref{eq:f5>0}}}{\le} 6\delta + 6\mu_0^\ast - 4g(u_1) g(u_2) g(u_3) \\
    \qquad + 2(g(u_1)+g(u_2)+g(u_3))\mu_1 + 2(g(u_1)g(u_3)+g(u_2)g(u_3))\mu_2 \le 6\delta \end{multlined} \\
\Rightarrow & 0 \le F_5(\vec{u},\mu_0^\ast,\mu_1,\mu_2) \le 6\delta\\
\stackrel{\mbox{\eqref{eq:f4=f5}}}{\Rightarrow} & 0 \le F_4(\vec{u},\mu_0^\ast,\mu_1,\mu_2) \le 6\delta\\
\Rightarrow & F_4(\vec{u},\mu_0^\ast,\mu_1,\mu_2) = O(\delta).
\end{aligned}
\end{equation}

\item If $\alpha_{3}^{\mu}(\vec{u})=1$, using \eqref{eq:beta_mu_u} we have 
\begin{equation}\label{eq:mu2<g3}
\mu_2 < g(u_3)
\end{equation}
Using the third equation of \eqref{eq:b1_diff_2nd_indicator} we have,
\begin{equation}\label{eq:f4_2ndterm_bound}
0 \stackrel{\mbox{\eqref{eq:mu2<g3}}}{<} 2g(u_1)g(u_2) (g(u_3)-\mu_2) \le 6\delta
\end{equation}
Fix $u_1$,$u_2$ and $\mu_2$. 
Denote by $u_{3,\text{min}}$ and $u_{3,\text{max}}$ the minimum and maximum value of $u_3$. 
Then we have using \eqref{eq:f4_2ndterm_bound},
\begin{equation}\label{eq:g3_bound_case2}
\begin{aligned}
& 2g(u_1)g(u_2) |g(u_{3,\text{max}})-g(u_{3,\text{min}})| <6\delta \\
\Rightarrow & |g(u_{3,\text{max}})-g(u_{3,\text{min}})| <\frac{3\delta}{2g(u_1)g(u_2)} \stackrel{\mbox{\eqref{as:assumption4}}}{\le} \frac{3\delta}{c_4^2}
\end{aligned}
\end{equation}
Thus we have $u_3=O(\delta)$.
\end{enumerate}


Thus for Case 2 we have $u_2 = O(\delta)$ from \eqref{eq:g2_bound_case2} and, depending on whether $\alpha_{3}^{\mu}(\vec{u})=1$, either $F_4(\vec{u},\mu_0^\ast,\mu_1,\mu_2)=O(\delta)$ or $u_3 = O(\delta)$.
The verification of condition \eqref{eq:o_delta_conjecture_mu0} then follows by exactly the same argument as in Case~1, with the roles of the coordinates adjusted accordingly.
Hence \eqref{eq:o_delta_conjecture_mu0} holds in Case~2 as well, and Proposition~\ref{prop:o_delta_condition_mu_0} again yields \eqref{eq:mu_0_optimality}.

\paragraph{Case 3: The third indicator function of the right-hand side of \eqref{eq:b1_diff} is true:} In this case, we have the third indicator function of \eqref{eq:b1_diff} as true. This means the following inequalities are true:
{\footnotesize
\begin{equation}\label{eq:b1_diff_3rd_indicator}
\begin{aligned}
& 0 \le 4g(u_1) g(u_2) g(u_3) - 2g(u_1)\mu_1 - 2(g(u_1)g(u_2)+g(u_1)g(u_3))\mu_2 \\
& 0 \le 2g(u_1) g(u_2) g(u_3) - 2g(u_1)g(u_2)\mu_2 \\
& 6\mu_0^\ast \le 6g(u_1) g(u_2) g(u_3) - 2(g(u_1)+g(u_2)+g(u_3))\mu_1 - 2(g(u_1)g(u_2)+g(u_1)g(u_3)+g(u_2)g(u_3))\mu_2 \le 6(\mu_0^\ast+\delta)
\end{aligned}
\end{equation}}
From the first inequality of \eqref{eq:b1_diff_3rd_indicator}, based on $\beta_2^{\mu}(\vec{u})$ as in \eqref{eq:beta_mu_u}, we have
\begin{equation}\label{eq:beta_2_0_case3}
\beta_2^{\mu}(\vec{u})=0 \Leftrightarrow \alpha_2^{\mu}(\vec{u})=1
\end{equation}
Also, from the second inequality of \eqref{eq:b1_diff_3rd_indicator}, based on $\beta_3^{\mu}(\vec{u})$ as in \eqref{eq:beta_mu_u}, we have
\begin{equation}\label{eq:beta_3_0_case3}
\beta_3^{\mu}(\vec{u})=0 \Leftrightarrow \alpha_3^{\mu}(\vec{u})=1
\end{equation}
Using Lemma \ref{lem:gdot} and the second inequality of \eqref{eq:b1_diff_3rd_indicator}, we establish the following notation and inequalities:
\begin{equation}\label{eq:gi_prime_case3}
\begin{aligned}
& g(u_1)\ge g(u_2)\ge g(u_3) \ge \mu_2 \ge \frac{\mu_2}{3} \\
& g'(u_1) = g(u_1) - \frac{\mu_2}{3} \ge 0 \\
& g'(u_2) = g(u_2) - \frac{\mu_2}{3} \ge 0 \\
& g'(u_3) = g(u_3) - \frac{\mu_2}{3} \ge 0 \\
& g'(u_1)\ge g'(u_2)\ge g'(u_3) 
\end{aligned}
\end{equation}
Recall the third sequence of inequalities of \eqref{eq:b1_diff_3rd_indicator},
{\footnotesize
\begin{equation}\label{eq:case3_3rd_equation}
\begin{aligned}
& 6\mu_0^\ast \le 6g(u_1) g(u_2) g(u_3) - 2(g(u_1)+g(u_2)+g(u_3))\mu_1 - 2(g(u_1)g(u_2)+g(u_1)g(u_3)+g(u_2)g(u_3))\mu_2 \le 6(\mu_0^\ast+\delta) \\
\Leftrightarrow & \mu_0^\ast \le g(u_1) g(u_2) g(u_3) - \frac{(g(u_1)+g(u_2)+g(u_3))\mu_1}{3} - \frac{(g(u_1)g(u_2)+g(u_1)g(u_3)+g(u_2)g(u_3))\mu_2}{3} \le \mu_0^\ast+\delta
\end{aligned}
\end{equation}}
Analyzing the term in the middle of the above equation, we have,
\begin{equation}\label{eq:case3_3rd_equation_middleterm}
\begin{aligned}
& g(u_1) g(u_2) g(u_3) - \frac{(g(u_1)+g(u_2)+g(u_3))\mu_1}{3} - \frac{(g(u_1)g(u_2)+g(u_1)g(u_3)+g(u_2)g(u_3))\mu_2}{3}\\
=& \bigg(g(u_1) - \frac{\mu_2}{3}\bigg)\bigg(g(u_2) - \frac{\mu_2}{3}\bigg)\bigg(g(u_3) - \frac{\mu_2}{3}\bigg) - \bigg(\frac{\mu_2}{3}\bigg)^2 (g(u_1)+g(u_2)+g(u_3)) \\
& - \frac{\mu_1}{3} (g(u_1)+g(u_2)+g(u_3)) + \bigg(\frac{\mu_2}{3}\bigg)^3 \\
\stackrel{\mbox{\eqref{eq:gi_prime_case3}}}{=} & g'(u_1)g'(u_2)g'(u_3) - \bigg(\bigg(\frac{\mu_2}{3}\bigg)^2 + \frac{\mu_1}{3} \bigg) ( g'(u_1)+ g'(u_2)+ g'(u_3)+\mu_2) + \bigg(\frac{\mu_2}{3}\bigg)^3 \\
= & g'(u_1)g'(u_2)g'(u_3) - \bigg(\frac{\mu_2^2}{9} + \frac{\mu_1}{3} \bigg) ( g'(u_1)+ g'(u_2)+ g'(u_3)) - \frac{\mu_2^3}{9} -\frac{\mu_1\mu_2}{3} + \frac{\mu_2^3}{27} \\
= & g'(u_1)g'(u_2)g'(u_3) - \bigg(\frac{\mu_2^2}{9} + \frac{\mu_1}{3} \bigg) ( g'(u_1)+ g'(u_2)+ g'(u_3)) - \frac{\mu_1\mu_2}{3} - \frac{2\mu_2^3}{27} \\
= & g'(u_3)\bigg(g'(u_1)g'(u_2) - \frac{\mu_2^2}{9} - \frac{\mu_1}{3} \bigg) - c_3(\mu_1,\mu_2,u_1,u_2) \\
\end{aligned}
\end{equation}
where $c_3(\mu_1,\mu_2,u_1,u_2)$ is a constant term w.r.t. $\mu_1,\mu_2,u_1$ and $u_2$. 

From the first inequality of \eqref{eq:b1_diff_3rd_indicator} we have,
\begin{equation}\label{eq:b1_diff_3rd_ind_1st_eq}
\begin{aligned}
& 0 \le 4g(u_1) g(u_2) g(u_3) - 2g(u_1)\mu_1 - 2(g(u_1)g(u_2)+g(u_1)g(u_3))\mu_2 \\
\Leftrightarrow & 0 \le 2g(u_2) g(u_3) - \mu_1 - (g(u_2)+g(u_3))\mu_2 \\ 
\Leftrightarrow & \mu_1 + (g(u_2)+g(u_3))\mu_2 \le 2g(u_2) g(u_3) \\ 
\stackrel{\mbox{\eqref{eq:gi_prime_case3}}}{\Leftrightarrow} & \mu_1 + \bigg(g'(u_2)+g'(u_3)+\frac{2\mu_2}{3}\bigg)\mu_2 \le 2 \bigg(g'(u_2)+\frac{\mu_2}{3}\bigg)\bigg(g'(u_3)+\frac{\mu_2}{3}\bigg) \\ 
\Leftrightarrow & \mu_1 + \bigg(g'(u_2)+g'(u_3)\bigg)\mu_2 +\frac{2\mu_2^2}{3} \le 2g'(u_2)g'(u_3) + \frac{2\mu_2}{3}\bigg(g'(u_2)+g'(u_3)\bigg)+\frac{2\mu_2^2}{9} \\ 
\Leftrightarrow & \mu_1 +\frac{4\mu_2^2}{9} + \frac{\mu_2}{3}\bigg(g'(u_2)+g'(u_3)\bigg)  \le 2g'(u_2)g'(u_3) \\ 
\Leftrightarrow & \mu_1 +\frac{4\mu_2^2}{9} \le 2g'(u_2)g'(u_3) - \frac{\mu_2}{3}\bigg(g'(u_2)+g'(u_3)\bigg) \\ 
\Leftrightarrow & \frac{\mu_1}{2} +\frac{2\mu_2^2}{9}  \le  \bigg(g'(u_2)-\frac{\mu_2}{6}\bigg)\bigg(g'(u_3)-\frac{\mu_2}{6}\bigg)-\frac{\mu_2^2}{36}\\ 
\Leftrightarrow & \frac{\mu_1}{2} +\frac{\mu_2^2}{4}  \le  \bigg(g'(u_2)-\frac{\mu_2}{6}\bigg)\bigg(g'(u_3)-\frac{\mu_2}{6}\bigg)\\ 
\end{aligned}
\end{equation}
Recall using \eqref{eq:gi_prime_case3}, we have,
\begin{equation}
\begin{aligned}
& g(u_2) \ge \mu_2\\
\Rightarrow & g'(u_2) \stackrel{\mbox{\eqref{eq:gi_prime_case3}}}{=}  g(u_2)-\frac{\mu_2}{3}  > \frac{\mu_2}{6}  \\
\Rightarrow & g'(u_2)- \frac{\mu_2}{6} >0
\end{aligned}
\end{equation}
Similarly, from \eqref{eq:gi_prime_case3}, we have $g'(u_3)- \frac{\mu_2}{6} >0$.
Recall we have $g(u_1) \ge g(u_2)\ge g(u_3) \Rightarrow g'(u_1) \ge g'(u_2)\ge g'(u_3)$.
Thus it follows from \eqref{eq:b1_diff_3rd_ind_1st_eq},
\begin{equation}\label{eq:case3_g2_bound1}
\begin{aligned}
& \frac{\mu_1}{2} +\frac{\mu_2^2}{4}  \le  \bigg(g'(u_2)-\frac{\mu_2}{6}\bigg)^2 \\
\Leftrightarrow & \sqrt{\frac{\mu_1}{2} +\frac{\mu_2^2}{4}}  \le  \bigg(g'(u_2)-\frac{\mu_2}{6}\bigg)
\end{aligned}
\end{equation}
Since $g'(u_1) \ge g'(u_2)$, we also have that
\begin{equation}\label{eq:case3_g1_bound1}
\sqrt{\frac{\mu_1}{2} +\frac{\mu_2^2}{4}}  \le  \bigg(g'(u_1)-\frac{\mu_2}{6}\bigg)
\end{equation}
From \eqref{eq:case3_g2_bound1} and \eqref{eq:case3_g1_bound1} we have,
\begin{equation}\label{eq:case3_g1g2_bound1}
\begin{aligned}
 &\bigg(g'(u_1)-\frac{\mu_2}{6}\bigg) \ge \bigg(g'(u_2)-\frac{\mu_2}{6}\bigg) \ge \sqrt{\frac{\mu_1}{2} +\frac{\mu_2^2}{4}} \\
\Rightarrow & g'(u_1)g'(u_2) - \frac{\mu_2^2}{9} - \frac{\mu_1}{3} \\
\ge & \bigg(\frac{\mu_2}{6} + \sqrt{\frac{\mu_1}{2} +\frac{\mu_2^2}{4}} \bigg)^2 - \frac{\mu_2^2}{9} - \frac{\mu_1}{3} \\
= & \frac{\mu_2^2}{36} + \frac{\mu_1}{2} +\frac{\mu_2^2}{4} + \frac{\mu_2}{3}\sqrt{\frac{\mu_1}{2} +\frac{\mu_2^2}{4}} - \frac{\mu_2^2}{9} - \frac{\mu_1}{3} \\
= &  \frac{\mu_2^2}{6} + \frac{\mu_1}{6} + \frac{\mu_2}{3}\sqrt{\frac{\mu_1}{2} +\frac{\mu_2^2}{4}} \\
\ge & 0.
\end{aligned}
\end{equation}
Fix $u_1$,$u_2$,$\mu_1$ and $\mu_2$. 
Denote by $u_{3,\text{min}}$ and $u_{3,\text{max}}$ the minimum and maximum value of $u_3$. 
Then we have using  \eqref{eq:case3_3rd_equation} and \eqref{eq:case3_3rd_equation_middleterm},
\begin{equation}\label{eq:g3_bound_case3}
\begin{aligned}
& \bigg(g'(u_1)g'(u_2) - \frac{\mu_2^2}{9} - \frac{\mu_1}{3}\bigg)|g'(u_{3,\text{min}})-g'(u_{3,\text{max}})| <\delta \\
\stackrel{\mbox{\eqref{eq:case3_g1g2_bound1}}}{\Rightarrow} & |g'(u_{3,\text{min}})-g'(u_{3,\text{max}})|  \le \frac{\delta}{ \frac{\mu_2^2}{6} + \frac{\mu_1}{6} + \frac{\mu_2}{3}\sqrt{\frac{\mu_1}{2} +\frac{\mu_2^2}{4}}} = O(\delta)
\end{aligned}
\end{equation}
Thus we have $u_3=O(\delta)$.
Next, we show that the second integrand in the second term of \eqref{eq:mu0_min2} is $O(\delta)$. 
This term can be rewritten as in \eqref{eq:secondterm_lagrangian5} as,
\begin{equation}\label{eq:secondterm_lagrangian5_case3}
\begin{aligned}
& 6\delta + 6\mu_0^\ast -2g(u_1) g(u_2) g(u_3) + 2\bigg(g(u_2) + g(u_3)\bigg)\mu_1 + 2g(u_{2})g(u_{3})\mu_2 \\
&+ 2\alpha_{2}^{\mu}(\vec{u})\bigg(-g(u_1)g(u_2)g(u_3) + g(u_1)\mu_1+g(u_1)g(u_3)\mu_2\bigg) \\
& + 2\alpha_{2}^{\mu}(\vec{u})\alpha_{3}^{\mu}(\vec{u})\bigg(-g(u_1)g(u_2)g(u_3)+g(u_1)g(u_2)\mu_2\bigg)\\
\stackrel{\mbox{\eqref{eq:beta_2_0_case3},\eqref{eq:beta_3_0_case3}}}{=}& 6\delta + 6\mu_0^\ast -2g(u_1) g(u_2) g(u_3) + 2\bigg(g(u_2) + g(u_3)\bigg)\mu_1 + 2g(u_{2})g(u_{3})\mu_2 \\
& -2g(u_1)g(u_2)g(u_3) + 2g(u_1)\mu_1 + 2g(u_1)g(u_3)\mu_2  -2g(u_1)g(u_2)g(u_3)+2g(u_1)g(u_2)\mu_2\\
=& 6\delta + 6\mu_0^\ast -6g(u_1) g(u_2) g(u_3) + 2\bigg(g(u_1)+g(u_2) + g(u_3)\bigg)\mu_1 \\
&+ 2\bigg(g(u_1)g(u_2) + g(u_1)g(u_3)+g(u_{2})g(u_{3})\bigg)\mu_2 \\
\stackrel{\mbox{\eqref{eq:b1_diff_3rd_indicator}}}{=} & O(\delta) 
\end{aligned}
\end{equation}
Combining \eqref{eq:secondterm_lagrangian5_case3} and \eqref{eq:g3_bound_case3}, we have \eqref{eq:o_delta_conjecture_mu0}.
The proof then follows from proposition \ref{prop:o_delta_condition_mu_0}.
\end{proof}

\subsubsection{\texorpdfstring
  {Sufficient condition for characterizing the minimizer $\mu_0^\ast$}
  {Sufficient condition for characterizing the minimizer mu0*}}
The following proposition states a sufficient condition under which a simple integral identity fully characterizes the minimizer $\mu_0^\ast$ as derived in the above theorem and serves as a key stepping stone in that proof. 
This allows the proof of Theorem \ref{thm:mu0_optimality} to trace a single necessary equality in \eqref{eq:mu_0_optimality}.
\begin{prop}
\label{prop:o_delta_condition_mu_0}
Let $\mu_0^\ast \in \mathbb{R}_+$ and fix $\mu_1, \mu_2$. Suppose
\begin{equation}\label{eq:o_delta_condition_integral}
\int_Q\biggl[\alpha_1^{\mu_0^\ast}(\vec{u})-\alpha_1^{\mu_0^\ast+\delta}(\vec{u})\biggr]\biggl[6\delta+6\mu_0^\ast-2c_2(\vec{u},\mu_1,\mu_2)\biggr]d\vec{u}=o(\delta)  \\
\end{equation}
where $c_2(\vec{u}, \mu_1, \mu_2)$ is as defined in \eqref{eq:c1_c2}. Then \eqref{eq:mu_0_optimality} holds.
\end{prop}
\begin{proof}
Suppose \eqref{eq:o_delta_condition_integral} holds.
We then demonstrate that then $\mu_0^\ast$ is the minimizer of $L(\vec{D}^\mu, \mu)$ of the dual problem \eqref{eq:dual2} w.r.t. $\mu_0$ only if $\alpha=6\int_Q\alpha_1^{\mu_0^\ast}(\vec{u})d\vec{u} $.

From \eqref{eq:mu0_min} we have that if $\mu_0^\ast$ is the minimizer,
\begin{equation}\label{eq:mu_0_must_hold}
\begin{aligned}
0&\le\delta\biggl[\alpha-6\int_Q\alpha_1^{\mu_0^\ast}(\vec{u})d\vec{u} \biggr]+\int_Q\biggl[\alpha_1^{\mu_0^\ast}(\vec{u})-\alpha_1^{\mu_0^\ast+\delta}(\vec{u})\biggr]\biggl[6\delta+6\mu_0^\ast-2c_2(\vec{u},\mu_1,\mu_2)\biggr]d\vec{u}\\
&=\delta \tau + o(\delta),\ \forall \delta 
\end{aligned}
\end{equation}
where, 
\begin{equation*}
\tau=\alpha-6\int_Q\alpha_1^{\mu_0^\ast}(\vec{u})d\vec{u}
\end{equation*} 
and using \eqref{eq:o_delta_condition_integral}, 
\begin{equation*}
\int_Q\biggl[\alpha_1^{\mu_0^\ast}(\vec{u})-\alpha_1^{\mu_0^\ast+\delta}(\vec{u})\biggr]\biggl[6\delta+6\mu_0^\ast-2c_2(\vec{u},\mu_1,\mu_2)\biggr]d\vec{u}=o(\delta)
\end{equation*} 

We use a contradiction argument to prove that the condition $\alpha=6\int_Q\alpha_1^{\mu_0^\ast+\delta}(\vec{u})d\vec{u}$ must hold.
Suppose \begin{equation}\label{eq:tau>0}
\tau =\alpha-6\int_Q\alpha_1^{\mu_0^\ast}(\vec{u})d\vec{u} > 0.
\end{equation}
Then choosing $\delta$ small enough for some $\delta<0$, we have using \eqref{eq:o_delta_condition_integral} and \eqref{eq:tau>0},
\begin{equation*}
 \delta\tau + o(\delta)<0
\end{equation*}
which contradicts \eqref{eq:mu_0_must_hold} which states that $\forall \delta$, $$ \delta\tau + o(\delta)\ge 0. $$
Therefore, \eqref{eq:tau>0} cannot be true.
A similar argument holds if $\tau<0$, and it cannot be true either.
So we must have,
\begin{equation*}
\tau=\alpha-6\int_Q\alpha_1^{\mu_0^\ast}(\vec{u})d\vec{u} = 0.
\end{equation*}
Thus, we have proved that $\mu_0^\ast$ is the minimizer of $L(\vec{D}^\mu, \mu)$ of the dual problem \eqref{eq:dual2} w.r.t. $\mu_0$ only if $\alpha=6\int_Q\alpha_1^{\mu_0^\ast}(\vec{u})d\vec{u}$ i.e., \eqref{eq:mu_0_optimality} is true.
\end{proof}

\subsection{\texorpdfstring{Minimizing the Lagrangian w.r.t. $\mu_1$}{Minimizing the Lagrangian w.r.t. mu1}}
\label{app:mu1_optimality}
Next, we have the Lagrangian as in \eqref{eq:lagrangian5_general}, and our objective is to minimize it w.r.t. $\mu_1$.
Note that from \eqref{eq:beta_mu_u} we have that both $\beta_1^{\mu}(\vec{u})$ and $\beta_2^{\mu}(\vec{u})$ depend on $\mu_1$, whereas $\beta_3^{\mu}(\vec{u})$ does not. 
Consequently, we find that both $\alpha_1^{\mu}(\vec{u})$ and $\alpha_2^{\mu}(\vec{u})$ depend on $\mu_1$, whereas $\alpha_3^{\mu}(\vec{u})$ does not depend on $\mu_1$. 
Moreover, in this case, we only need to consider the case  
\begin{equation}\label{eq:alpha_1_equals_1_mu1}
\alpha_1^{\mu}(\vec{u})=1,
\end{equation}
because if $\alpha_1^{\mu}(\vec{u})=0$, then
\begin{equation*}
L(\vec{D}^\mu, \mu) \stackrel{\mbox{\eqref{eq:lagrangian5_general}}}{=} \alpha(\mu_0+\mu_1+\mu_2)
\end{equation*}
therefore the minimimizer of $L(\vec{D}^\mu, \mu)$ w.r.t. $\mu_1$ can't be optimal until it is zero.
Thus, we assume $\alpha_1^{\mu}(\vec{u})=1$.
Plugging this in the Lagrangian \eqref{eq:lagrangian5_general} we have,
\begin{equation}\label{eq:lagrangian5_mu1}
\begin{aligned}
L(\vec{D}^\mu, \mu) \stackrel{\mbox{\eqref{eq:lagrangian5_general},\eqref{eq:alpha_1_equals_1_mu1}}}{=}&\alpha(\mu_0+\mu_1+\mu_2)\\
&+2\int_Q g(u_1) g(u_2) g(u_3) \biggl(1 + \alpha_2^{\mu}(\vec{u})+\alpha_2^{\mu}(\vec{u})\alpha_3^{\mu}(\vec{u}) \biggr)d\vec{u}\\
&-6\int_Q \mu_0d\vec{u} \\
&-2\int_Q \Bigl(g(u_1)\alpha_2^{\mu}(\vec{u})+g(u_2) + g(u_3)\Bigr)\mu_1d\vec{u}\\
& -2\int_Q \Bigl(g(u_{2})g(u_{3})+\alpha_{2}^{\mu}(\vec{u})g(u_{1})g(u_{3})+\alpha_{2}^{\mu}(\vec{u})\alpha_{3}^{\mu}(\vec{u})g(u_{1})g(u_{2})\Bigr)\mu_2 d\vec{u}.
\end{aligned}
\end{equation}
Note that on the right-hand side of the above equation, only $\alpha_{2}^{\mu}(\vec{u})$ depends on $\mu_1$.
For the remainder of this section, we will work with this expression of $L(\vec{D}^\mu, \mu)$.

Since we minimize $L(\vec{D}^\mu, \mu)$ w.r.t. $\mu_1$ we denote it by $L(\mu_1)$. 
Also, since only $\alpha_{2}^{\mu}(\vec{u})$ depends on $\mu_1$ we denote $\alpha_2^{\mu}(\vec{u})$ only by $\alpha_2^{\mu_1}(\vec{u})$. 
Similarly, we denote $\beta_2^{\mu}(\vec{u})$ by $\beta_2^{\mu_1}(\vec{u})$.
Because $\alpha_{3}^{\mu}(\vec{u})$ depends only on $\mu_2$ and we are working to minimize $L(\vec{D}^\mu, \mu)$ w.r.t. $\mu_1$ keeping $\mu_0$ and $\mu_2$ fixed, we denote $\alpha_{3}^{\mu}(\vec{u})$ and $\beta_{3}^{\mu}(\vec{u})$ by $\alpha_{3}(\vec{u})$ and $\beta_{3}(\vec{u})$ respectively.

Thus, we have that if $\mu_1^\ast$ is the minimizer, then for any $\delta\neq0$ we have $$L(\mu_1^\ast) \le L(\mu_1^\ast+\delta).$$
That is,
{\footnotesize
\begin{equation}\label{eq:mu1_min}
\begin{aligned}
0 &\le L(\mu_1^\ast+\delta)- L(\mu_1^\ast)\\
&\stackrel{\mbox{\eqref{eq:lagrangian5_mu1}}}{=} \alpha\delta \\
&+2\int_Q g(u_1) g(u_2) g(u_3) \biggl(1 + \alpha_3(\vec{u}) \biggr) \bigg[\alpha_2^{\mu_1^\ast+\delta}(\vec{u})-\alpha_2^{\mu_1^\ast}(\vec{u}) \bigg]d\vec{u}\\
&-2\int_Q \Bigl(g(u_1)\alpha_2^{\mu_1^\ast+\delta}(\vec{u})+g(u_2) + g(u_3)\Bigr)(\mu_1^\ast+\delta)d\vec{u}\\
&+2\int_Q \Bigl(g(u_1)\alpha_2^{\mu_1^\ast}(\vec{u})+g(u_2) + g(u_3)\Bigr)\mu_1^\ast d\vec{u}\\
& -2\int_Q \bigg[g(u_{2})g(u_{3})+\bigg(g(u_{1})g(u_{3})+\alpha_{3}(\vec{u})g(u_{1})g(u_{2})\bigg) \alpha_{2}^{\mu_1^\ast+\delta}(\vec{u})\bigg]\mu_2 d\vec{u}\\
& +2\int_Q \bigg[g(u_{2})g(u_{3})+\bigg(g(u_{1})g(u_{3})+\alpha_{3}(\vec{u})g(u_{1})g(u_{2})\bigg) \alpha_{2}^{\mu_1^\ast}(\vec{u})\bigg]\mu_2 d\vec{u} \\
=& \alpha\delta \\
&+2\int_Q g(u_1) g(u_2) g(u_3) \biggl(1 + \alpha_3(\vec{u}) \biggr) \bigg[\alpha_2^{\mu_1^\ast+\delta}(\vec{u})-\alpha_2^{\mu_1^\ast}(\vec{u}) \bigg]d\vec{u}\\
&-2 \int_Q g(u_1)\alpha_2^{\mu_1^\ast+\delta}\mu_1^\ast d\vec{u} -2 \int_Q (g(u_2) + g(u_3)) \mu_1^\ast d\vec{u}  -2 \int_Q g(u_1)\alpha_2^{\mu_1^\ast+\delta}\delta d\vec{u} -2 \int_Q (g(u_2) + g(u_3)) \delta d\vec{u} \\
&+ 2 \int_Q g(u_1)\alpha_2^{\mu_1^\ast}\mu_1^\ast d\vec{u} + 2 \int_Q (g(u_2) + g(u_3)) \mu_1^\ast d\vec{u} \\
&+2\int_Q \bigg(g(u_{1})g(u_{3})+\alpha_{3}(\vec{u})g(u_{1})g(u_{2})\bigg) \bigg[\alpha_2^{\mu_1^\ast}(\vec{u})-\alpha_2^{\mu_1^\ast+\delta}(\vec{u}) \bigg]\mu_2 d\vec{u}\\
=&  \alpha\delta \\
& -2\delta \int_Q \bigg(g(u_1)\alpha_2^{\mu_1^\ast+\delta}+g(u_2)+ g(u_3) \bigg) d\vec{u} \\
& +2\int_Q g(u_1) g(u_2) g(u_3) \biggl(1 + \alpha_3(\vec{u}) \biggr) \bigg[\alpha_2^{\mu_1^\ast+\delta}(\vec{u})-\alpha_2^{\mu_1^\ast}(\vec{u}) \bigg]d\vec{u}\\
& -2\int_Q g(u_1)\mu_1^\ast \bigg[\alpha_2^{\mu_1^\ast+\delta}(\vec{u})-\alpha_2^{\mu_1^\ast}(\vec{u}) \bigg]d\vec{u} \\
& -2\int_Q \bigg(g(u_{1})g(u_{3})+\alpha_{3}(\vec{u})g(u_{1})g(u_{2})\bigg) \bigg[\alpha_2^{\mu_1^\ast}(\vec{u})-\alpha_2^{\mu_1^\ast+\delta}(\vec{u}) \bigg]\mu_2d\vec{u} \\
=& \delta \bigg[\alpha - 2\int_Q \bigg(g(u_1)\alpha_2^{\mu_1^\ast+\delta}(\vec{u})+g(u_2)+ g(u_3) \bigg)d\vec{u} \bigg] \\
&+2\int_Q \bigg[g(u_1) g(u_2) g(u_3) \biggl(1 + \alpha_3(\vec{u}) \biggr)- g(u_1)\mu_1^\ast -\bigg(g(u_{1})g(u_{3})+\alpha_{3}(\vec{u})g(u_{1})g(u_{2})\bigg)\mu_2 \bigg] \bigg[\alpha_2^{\mu_1^\ast+\delta}(\vec{u})-\alpha_2^{\mu_1^\ast}(\vec{u}) \bigg]d\vec{u}
\end{aligned}
\end{equation}}

\begin{theorem}
\label{thm:mu1_optimality}
Suppose Assumptions~\ref{as:assumption3},~\ref{as:assumption4}, and~\ref{as:assumption5} hold.  
Consider the dual formulation \eqref{eq:dual2} of the multiple hypotheses testing problem \eqref{eq:objconst}, and the associated Lagrangian $L(\vec{D}^\mu, \mu)$ as defined in \eqref{eq:lagrangian5_general}.  
Then, for fixed values of the other coordinates $\mu_0$ and $\mu_2$, the coordinate $\mu_1^\ast \in \mathbb{R}_+$ is a minimizer of $L(\vec{D}^\mu, \mu)$ with respect to $\mu_1$ only if
\begin{equation}
\label{eq:mu1_optimality}
\alpha = 1 - 2 \int_Q g(u_1) \beta_2^{\mu_1^\ast}(\vec{u}) d\vec{u}.
\end{equation}
\end{theorem}

\begin{proof}
From \eqref{eq:beta_mu_u} we have that,
\begin{equation}\label{eq:beta_mu_u_mu1}
\begin{aligned}
&\beta_2^{\mu_1^\ast+\delta}(\vec{u})-\beta_2^{\mu_1^\ast}(\vec{u})\\
=& \mathbbm{1}\biggl\{\mu,\vec{u}:
\left\{
\begin{array}{c}
2g(u_1) g(u_2) g(u_3) - 2g(u_1)(\mu_1^\ast+\delta) - 2g(u_1)g(u_3)\mu_2 \le 0 \\
4g(u_1) g(u_2) g(u_3) - 2g(u_1)(\mu_1^\ast+\delta) - 2(g(u_1)g(u_3) + g(u_1)g(u_2))\mu_2 \le 0
\end{array}\right\}\biggr\} \\
&- \mathbbm{1}\biggl\{\mu,\vec{u}:
\left\{
\begin{array}{c}
2g(u_1) g(u_2) g(u_3) - 2g(u_1)\mu_1^\ast - 2g(u_1)g(u_3)\mu_2 \le 0 \\
4g(u_1) g(u_2) g(u_3) - 2g(u_1)\mu_1^\ast - 2(g(u_1)g(u_3) + g(u_1)g(u_2))\mu_2 \le 0
\end{array}\right\}\biggr\} \\
=& \mathbbm{1}\biggl\{\mu,\vec{u}:
\left\{
\begin{array}{c}
2g(u_1)\mu_1^\ast \le 2g(u_1) g(u_2) g(u_3) - 2g(u_1)g(u_3)\mu_2 \le 2g(u_1)(\mu_1^\ast+\delta) \\
4g(u_1) g(u_2) g(u_3)  - 2(g(u_1)g(u_3) + g(u_1)g(u_2))\mu_2 \le  2g(u_1)(\mu_1^\ast+\delta)
\end{array}\right\}\biggr\} \\
&+ \mathbbm{1}\biggl\{\mu,\vec{u}:
\left\{
\begin{array}{c}
2g(u_1) g(u_2) g(u_3)  - 2g(u_1)g(u_3)\mu_2 \le 2g(u_1)\mu_1^\ast \\
2g(u_1)\mu_1^\ast \le 4g(u_1) g(u_2) g(u_3)  - 2(g(u_1)g(u_3) + g(u_1)g(u_2))\mu_2 \le - 2g(u_1)(\mu_1^\ast+\delta)
\end{array}\right\}\biggr\} \\
\le & \mathbbm{1}\biggl\{\mu,\vec{u}:
\left\{
\begin{array}{c}
2g(u_1)\mu_1^\ast \le 2g(u_1) g(u_2) g(u_3) - 2g(u_1)g(u_3)\mu_2 \le 2g(u_1)(\mu_1^\ast+\delta) \\
2g(u_1) g(u_2) g(u_3)  - 2 g(u_1)g(u_2)\mu_2 \le  2g(u_1)\delta
\end{array}\right\}\biggr\} \\
&+ \mathbbm{1}\biggl\{\mu,\vec{u}:
\left\{
\begin{array}{c}
0 \le 2g(u_1) g(u_2) g(u_3)  - 2g(u_1)g(u_2)\mu_2 \\
2g(u_1)\mu_1^\ast \le 4g(u_1) g(u_2) g(u_3)  - 2(g(u_1)g(u_3) + g(u_1)g(u_2))\mu_2 \le  2g(u_1)(\mu_1^\ast+\delta)
\end{array}\right\}\biggr\} \\
\end{aligned}
\end{equation}
Note that only one of the two indicator functions on the right-hand side of \eqref{eq:beta_mu_u_mu1} can be true.

\begin{enumerate}
\item Case 1, suppose the first indicator function on the right-hand side of \eqref{eq:beta_mu_u_mu1} is true, i.e.,
\begin{equation}
\label{eq:1st_ind_mu1}
\begin{aligned}
2g(u_1)\mu_1^\ast \le 2g(u_1) g(u_2) g(u_3) - 2g(u_1)g(u_3)\mu_2 &\le 2g(u_1)(\mu_1^\ast+\delta) \\
2g(u_1) g(u_2) g(u_3)  - 2 g(u_1)g(u_2)\mu_2 &\le  2g(u_1)\delta
\end{aligned}
\end{equation}
Fix $u_1$,$u_3$,$\mu_2$. 
Denote by $u_{2,\text{min}}$ and $u_{2,\text{max}}$ the minimum and maximum value of $u_2$. 
Then, we use the first inequality of \eqref{eq:1st_ind_mu1} to obtain the following.
\begin{equation}\label{eq:g2_bound_mu1}
\begin{aligned}
& 2g(u_1)\mu_1^\ast \le 2g(u_1) g(u_2) g(u_3) - 2g(u_1)g(u_3)\mu_2 &\le 2g(u_1)(\mu_1^\ast+\delta) \\
\Rightarrow & \mu_1^\ast \le g(u_2) g(u_3) - g(u_3)\mu_2 &\le (\mu_1^\ast+\delta) \\
\Rightarrow & |(g(u_{2,\text{min}})-g(u_{2,\text{max}}))g(u_3)| \le \delta \\
\Rightarrow & |(g(u_{2,\text{min}})-g(u_{2,\text{max}}))| \le \frac{\delta}{g(u_3)}
\end{aligned}
\end{equation}
Thus, we have for $u_2$,
\begin{equation}\label{eq:u2_bound_m1}
|u_2-u_2'|  \stackrel{\mbox{\eqref{as:assumption3}}}{\le} \frac{|g(u_2)-g(u_2')|}{c_3} \stackrel{\mbox{\eqref{eq:g2_bound_mu1}}}{\le} \frac{\delta}{c_3 g(u_3)} \ \stackrel{\mbox{\eqref{as:assumption4}}}{\le} \frac{\delta}{c_3 c_4}=O(\delta)
\end{equation}
We have $Q= \{\vec{u}\colon 0 \le u_1 \le u_2 \le u_3 \le 1\}$.  
Define the set  
\begin{equation}
A_{\delta}
=
\bigl\{
\vec{u}\in Q: \beta_2^{\mu_1^\ast+\delta}(\vec{u})\neq 
\beta_2^{\mu_1^\ast}(\vec{u})
\bigr\}.
\end{equation}
From \eqref{eq:u2_bound_m1}, for any fixed $u_1$ and $u_3$ values in the interval $[0,1]$, the possible $u_2$ value lies within a one-dimensional region whose measure is of order $\delta$. 
Since both $u_1$ and $u_3$ lie in $[0,1]$, the measure of the three-dimensional set $A_{\delta}$ is of order $\delta$.  
The difference  
\begin{equation}
\beta_2^{\mu_1^\ast+\delta}(\vec{u}) 
-
\beta_2^{\mu_1^\ast}(\vec{u}) 
\end{equation}
takes values in $\{-1, 0,1\}$. 
Hence, 
\begin{equation}
\Bigl\lvert
\beta_2^{\mu_1^\ast+\delta}(\vec{u})
-
\beta_2^{\mu_1^\ast}(\vec{u})
\Bigr\rvert
\le1
\quad
\text{for all } \vec{u}\in Q.
\end{equation}
Therefore, by definition of $A_{\delta}$,  
\begin{equation}
\left\lvert
\int_Q
\Bigl(
\beta_2^{\mu_1^\ast+\delta}(\vec{u})
-
\beta_2^{\mu_1^\ast}(\vec{u})
\Bigr)
d\vec{u}
\right\rvert
\le
\int_{A_{\delta}} 1d\vec{u}
=
\mathrm{measure}(A_{\delta})
=
O(\delta).
\end{equation}
Thus, we have the following.
\begin{equation}\label{eq:b2_Odelta_mu1}
\int_Q
\Bigl(
\alpha_2^{\mu_1^\ast}(\vec{u})
-
\alpha_2^{\mu_1^\ast+\delta}(\vec{u})
\Bigr)
d\vec{u}
=
O(\delta)
\quad\Longleftrightarrow\quad
\int_Q
\Bigl(
\beta_2^{\mu_1^\ast+\delta}(\vec{u})
-
\beta_2^{\mu_1^\ast}(\vec{u})
\Bigr)
d\vec{u}
=
O(\delta).
\end{equation}

Now consider the first integrand of the integral in \eqref{eq:mu1_min}.
\begin{equation}
\label{eq:mu1_min_1st_integrand}
g(u_1) g(u_2) g(u_3) \biggl(1 + \alpha_3(\vec{u}) \biggr)- g(u_1)\mu_1^\ast -\bigg(g(u_{1})g(u_{3})+\alpha_{3}(\vec{u})g(u_{1})g(u_{2})\bigg)\mu_2
\end{equation}
Depending on the value of $\alpha_{3}(\vec{u})$, there can be two subcases.
We consider both below.
\begin{enumerate}
\item First we consider $\alpha_{3}(\vec{u})=0$.
In this case \eqref{eq:mu1_min_1st_integrand} becomes,
\begin{equation}\label{eq:mu1_min_1st_integrand_alpha3_0}
g(u_1) g(u_2) g(u_3) - g(u_1)\mu_1^\ast -g(u_{1})g(u_{3})\mu_2
\end{equation}
From the first inequality of \eqref{eq:1st_ind_mu1} we have,
\begin{equation}
\label{eq:1st_ineq_subcase1_mi1min}
\begin{aligned}
&2g(u_1)\mu_1^\ast \le 2g(u_1) g(u_2) g(u_3) - 2g(u_1)g(u_3)\mu_2 \le 2g(u_1)(\mu_1^\ast+\delta) \\
\Leftrightarrow & g(u_1)\mu_1^\ast \le g(u_1) g(u_2) g(u_3) - g(u_1)g(u_3)\mu_2 \le g(u_1)(\mu_1^\ast+\delta) \\
\Leftrightarrow & 0 \le g(u_1) g(u_2) g(u_3) -g(u_1)\mu_1^\ast - g(u_1)g(u_3)\mu_2 \le g(u_1)\delta \stackrel{\mbox{\eqref{as:assumption5}}}{\le} c_5 \delta
\end{aligned}
\end{equation}
Thus using \eqref{eq:1st_ineq_subcase1_mi1min} we have for \eqref{eq:mu1_min_1st_integrand_alpha3_0},
\begin{equation}
\label{eq:mu1_min_1st_integrand_alpha3_0_Odelta}
g(u_1) g(u_2) g(u_3) - g(u_1)\mu_1^\ast -g(u_{1})g(u_{3})\mu_2=O(\delta).
\end{equation}
Then the second term in \eqref{eq:mu1_min} becomes,

\begin{equation}
\begin{aligned}
&\int_Q 
\bigg[
    g(u_1) g(u_2) g(u_3) \bigl(1 + \alpha_3(\vec{u}) \bigr)
    - g(u_1)\mu_1^\ast \\
&\qquad
    - \bigl(g(u_{1}) g(u_{3}) + \alpha_{3}(\vec{u}) g(u_{1}) g(u_{2})\bigr)\mu_2
\bigg]
\bigl[\alpha_2^{\mu_1^\ast+\delta}(\vec{u})
      - \alpha_2^{\mu_1^\ast}(\vec{u}) \bigr] \, d\vec{u} \\
&\le
\biggl\|
    g(u_1) g(u_2) g(u_3) \bigl(1 + \alpha_3(\vec{u}) \bigr)
    - g(u_1)\mu_1^\ast \\
&\qquad
    - \bigl(g(u_{1}) g(u_{3}) + \alpha_{3}(\vec{u}) g(u_{1}) g(u_{2})\bigr)\mu_2
\biggr\|_\infty
\int_Q \bigl[\alpha_2^{\mu_1^\ast+\delta}(\vec{u})
             - \alpha_2^{\mu_1^\ast}(\vec{u}) \bigr] \, d\vec{u} \\
&\stackrel{\eqref{eq:mu1_min_1st_integrand_alpha3_0_Odelta},\,
           \eqref{eq:b2_Odelta_mu1}}{=} O(\delta^2) \\
&= o(\delta).
\end{aligned}
\end{equation}

The condition \eqref{eq:mu1_optimality} thus follows from Proposition \eqref{prop:o_delta_condition_mu_1}.

\item Second, we consider $\alpha_{3}(\vec{u})=1$.
In this case \eqref{eq:mu1_min_1st_integrand} becomes,
\begin{equation}\label{eq:mu1_min_1st_integrand_alpha3_1}
2g(u_1) g(u_2) g(u_3)- g(u_1)\mu_1^\ast -\bigg(g(u_{1})g(u_{3})+g(u_{1})g(u_{2})\bigg)\mu_2
\end{equation}
Adding the two upper inequalities of \eqref{eq:1st_ind_mu1} we have,
\begin{equation}\label{eq:1st_ind_mu1_sum2upperineq}
\begin{aligned}
& 4g(u_1) g(u_2) g(u_3) - 2(g(u_1)g(u_2)+g(u_1)g(u_3))\mu_2 \le 2g(u_1)(\mu_1^\ast+\delta) + 2g(u_1)\delta \\
\Rightarrow & 4g(u_1) g(u_2) g(u_3) -2g(u_1)\mu_1^\ast- 2(g(u_1)g(u_2)+g(u_1)g(u_3))\mu_2 \le 4g(u_1)\delta \\
\Rightarrow & 2g(u_1) g(u_2) g(u_3) -g(u_1)\mu_1^\ast- (g(u_1)g(u_2)+g(u_1)g(u_3))\mu_2 \le 2g(u_1)\delta \stackrel{\mbox{\eqref{as:assumption5}}}{\le}  2c_5\delta 
\end{aligned}
\end{equation}
Also since $\alpha_{3}(\vec{u})=1$, this implies $\beta_{3}(\vec{u})=0$ from \eqref{eq:beta_mu_u}.
Using further \eqref{eq:beta_mu_u}, we have the following.
\begin{equation}\label{eq:beta30_mu1_subcase1}
\begin{aligned}
&\beta_{3}(\vec{u})=0 \\
\Rightarrow & 2g(u_1) g(u_2) g(u_3) > 2g(u_1)g(u_2)\mu_2 \\
\Rightarrow & g(u_1) g(u_2) g(u_3) > g(u_1)g(u_2)\mu_2 \\
\Rightarrow & g(u_1) g(u_2) g(u_3) - g(u_1)g(u_2)\mu_2 >0 \\
\end{aligned}
\end{equation}
Then analyzing \eqref{eq:mu1_min_1st_integrand_alpha3_1} we have,
\begin{equation}\label{eq:mu1_min_1st_integrand_alpha3_1_positive}
\begin{aligned}
& 2g(u_1) g(u_2) g(u_3)- g(u_1)\mu_1^\ast -\bigg(g(u_{1})g(u_{3})+g(u_{1})g(u_{2})\bigg)\mu_2 \\
= & \bigg( g(u_1) g(u_2) g(u_3) - g(u_{1})g(u_{2})\mu_2 \bigg) + \bigg(g(u_1) g(u_2) g(u_3) - g(u_1)\mu_1^\ast - g(u_{1})g(u_{3})\mu_2\bigg) > 0
\end{aligned}
\end{equation}
where the first term on the right-hand side is strictly positive by virtue of \eqref{eq:beta30_mu1_subcase1}, while the positivity of the second term follows from the lower bound provided in the first inequality of \eqref{eq:1st_ind_mu1}.
Thus using \eqref{eq:mu1_min_1st_integrand_alpha3_1_positive} and \eqref{eq:1st_ind_mu1_sum2upperineq} we have for \eqref{eq:mu1_min_1st_integrand_alpha3_1},
\begin{equation}\label{eq:subcase2_mu1_min}
\begin{aligned}
& 0< 2g(u_1) g(u_2) g(u_3)- g(u_1)\mu_1^\ast -\bigg(g(u_{1})g(u_{3})+g(u_{1})g(u_{2})\bigg)\mu_2 \le 2c_5 \delta \\
\Rightarrow & 2g(u_1) g(u_2) g(u_3)- g(u_1)\mu_1^\ast -\bigg(g(u_{1})g(u_{3})+g(u_{1})g(u_{2})\bigg)\mu_2 = O(\delta)
\end{aligned}
\end{equation}
Thus using \eqref{eq:b2_Odelta_mu1} and \eqref{eq:subcase2_mu1_min}
we have for the second term in \eqref{eq:mu1_min},

\begin{equation}
\begin{aligned}
&\int_Q 
\bigg[
    g(u_1) g(u_2) g(u_3) \bigl(1 + \alpha_3(\vec{u}) \bigr)
    - g(u_1)\mu_1^\ast \\
&\qquad
    - \bigl(g(u_{1}) g(u_{3}) + \alpha_{3}(\vec{u}) g(u_{1}) g(u_{2})\bigr)\mu_2
\bigg]
\bigl[
    \alpha_2^{\mu_1^\ast+\delta}(\vec{u})
    - \alpha_2^{\mu_1^\ast}(\vec{u})
\bigr] \, d\vec{u} \\
&\le 
\biggl\|
    g(u_1) g(u_2) g(u_3) \bigl(1 + \alpha_3(\vec{u}) \bigr)
    - g(u_1)\mu_1^\ast \\
&\qquad
    - \bigl(g(u_{1}) g(u_{3}) + \alpha_{3}(\vec{u}) g(u_{1}) g(u_{2})\bigr)\mu_2
\biggr\|_\infty
\int_Q
\bigl[
    \alpha_2^{\mu_1^\ast+\delta}(\vec{u})
    - \alpha_2^{\mu_1^\ast}(\vec{u})
\bigr] \, d\vec{u} \\
&\stackrel{\eqref{eq:subcase2_mu1_min},\,\eqref{eq:b2_Odelta_mu1}}{=} O(\delta^2) \\
&= o(\delta).
\end{aligned}
\end{equation}

The condition \eqref{eq:mu1_optimality} thus follows from Proposition \eqref{prop:o_delta_condition_mu_1}.
\end{enumerate}

\item Case 2, suppose the second indicator function of \eqref{eq:beta_mu_u_mu1} is true, i.e.,
\begin{equation}\label{eq:2nd_ind_mu1}
\begin{aligned}
& 0 \le 2g(u_1) g(u_2) g(u_3)  - 2g(u_1)g(u_2)\mu_2 \\
& 2g(u_1)\mu_1^\ast \le 4g(u_1) g(u_2) g(u_3)  - 2(g(u_1)g(u_3) + g(u_1)g(u_2))\mu_2 \le  2g(u_1)(\mu_1^\ast+\delta)
\end{aligned}
\end{equation}
From the first equation of the above, we have, using \eqref{eq:beta_mu_u} that
\begin{equation}
\label{eq:beta30_case2_mu1min}
\beta_{3}(\vec{u})=0 \Leftrightarrow \alpha_{3}(\vec{u})=1.
\end{equation}

Now consider the first integrand of the integral of \eqref{eq:mu1_min} as given by \eqref{eq:mu1_min_1st_integrand}
\begin{equation}\label{eq:mu1_min_1st_integrand_case2}
\begin{aligned}
&g(u_1) g(u_2) g(u_3) \biggl(1 + \alpha_3(\vec{u}) \biggr)- g(u_1)\mu_1^\ast -\bigg(g(u_{1})g(u_{3})+\alpha_{3}(\vec{u})g(u_{1})g(u_{2})\bigg)\mu_2 \\
\stackrel{\mbox{\eqref{eq:beta30_case2_mu1min}}}{=} &2g(u_1) g(u_2) g(u_3)- g(u_1)\mu_1^\ast -\bigg(g(u_{1})g(u_{3})+g(u_{1})g(u_{2})\bigg)\mu_2 
\end{aligned}
\end{equation}
Then from the second equation of \eqref{eq:2nd_ind_mu1}, we can show \eqref{eq:mu1_min_1st_integrand_case2} as O($\delta$)following,
\begin{equation}\label{eq:Odelta_case2_mu1min_1stintegrand}
\begin{aligned}
& 2g(u_1)\mu_1^\ast \le 4g(u_1) g(u_2) g(u_3)  - 2(g(u_1)g(u_3) + g(u_1)g(u_2))\mu_2 \le  2g(u_1)(\mu_1^\ast+\delta) \\
& 0 \le 2g(u_1) g(u_2) g(u_3) -g(u_1)\mu_1^\ast - (g(u_1)g(u_3) + g(u_1)g(u_2))\mu_2 \le 2g(u_1)\delta  \stackrel{\mbox{\eqref{as:assumption5}}}{\le} 2c_5 \delta \\
\Rightarrow & 2g(u_1) g(u_2) g(u_3)- g(u_1)\mu_1^\ast -\bigg(g(u_{1})g(u_{3})+g(u_{1})g(u_{2})\bigg)\mu_2 = O(\delta)
\end{aligned}
\end{equation}
Now, from the second equation of \eqref{eq:2nd_ind_mu1} we have the following.
\begin{equation}\label{eq:g2>g3_mu1min}
\begin{aligned}
& 2g(u_1)\mu_1^\ast \le 4g(u_1) g(u_2) g(u_3)  - 2(g(u_1)g(u_3) + g(u_1)g(u_2))\mu_2 \le  2g(u_1)(\mu_1^\ast+\delta) \\
\Leftrightarrow & 0 \le 4g(u_1) g(u_2) g(u_3)  - 2(g(u_1)g(u_3) + g(u_1)g(u_2))\mu_2 -2g(u_1)(\mu_1^\ast \le 2 \delta \\
\Leftrightarrow & g(u_2) g(u_3) - \frac{\mu_1^\ast}{2} - \frac{(g(u_3) +g(u_2))\mu_2}{2} \le \frac{2\delta}{4g(u_1)} \\
\Leftrightarrow & 0 \le \bigg(g(u_2)-\frac{\mu_2}{2} \bigg)\bigg(g(u_3)-\frac{\mu_2}{2} \bigg) - \frac{\mu_1^\ast}{2} - \frac{\mu_2^2}{4} \le \frac{\delta}{2g(u_1)}
\end{aligned}
\end{equation}

From Lemma \eqref{lem:gdot}, we have $g(u_2) \ge g(u_3)$ as $u_2 \le u_3$.
Also from the lower bound of \eqref{eq:g2>g3_mu1min} we have that,
\begin{equation}
\label{eq:g2>g3_mu1min2}
\bigg(g(u_2)-\frac{\mu_2}{2} \bigg)\bigg(g(u_3)-\frac{\mu_2}{2} \bigg)  \ge 0
\end{equation}
Without loss of generality, assume,
\begin{equation}
\label{eq:wlg_mu1min}
g(u_2) > g(u_3) > \frac{\mu_2}{2}.
\end{equation}
Then we have, 
\begin{equation}\label{eq:g2_lowerbd_mu1min}
\begin{aligned}
&\frac{\mu_1^\ast}{2} + \frac{\mu_2^2}{4}\\ \stackrel{\mbox{\eqref{eq:g2>g3_mu1min}}}{\le} & \bigg(g(u_2)-\frac{\mu_2}{2} \bigg)\bigg(g(u_3)-\frac{\mu_2}{2} \bigg) \\
\stackrel{\mbox{\eqref{eq:wlg_mu1min}}}{\le} &\bigg(g(u_2)-\frac{\mu_2}{2} \bigg)^2 \\
\Rightarrow & \bigg(g(u_2)-\frac{\mu_2}{2} \bigg) \ge \sqrt{\frac{\mu_1^\ast}{2} + \frac{\mu_2^2}{4}}
\end{aligned}
\end{equation}

Fix $u_1$,$u_3$,$\mu_1$ and $\mu_2$. 
Denote by $u_{3,\text{min}}$ and $u_{3,\text{max}}$ the minimum and maximum value of $u_3$. 
Then, using \eqref{eq:g2>g3_mu1min}, we have the following.
\begin{equation}\label{eq:g3_bound_mu1}
\begin{aligned}
& \bigg(g(u_2)-\frac{\mu_2}{2} \bigg) |g(u_{3,\text{min}})- g(u_{3,\text{max}})| \le \frac{\delta}{2g(u_1)} \\
\Rightarrow & |g(u_{3,\text{min}})- g(u_{3,\text{max}})| \le \frac{\delta}{2g(u_1)} \frac{1}{\bigg(g(u_2)-\frac{\mu_2}{2} \bigg)} \stackrel{\mbox{\eqref{eq:g2_lowerbd_mu1min}}}{<} \frac{\delta}{2g(u_1)} \frac{1}{\sqrt{\frac{\mu_1^\ast}{2} + \frac{\mu_2^2}{4}}} \\
\Rightarrow & |u_3 - u_3'| \stackrel{\mbox{\eqref{as:assumption3}}}{\le} \frac{1}{c_3} |g(u_3) - g(u_3')| \le \frac{1}{c_3} \frac{\delta}{2g(u_1)} \frac{1}{\sqrt{\frac{\mu_1^\ast}{2} + \frac{\mu_2^2}{4}}} \stackrel{\mbox{\eqref{as:assumption4}}}{\le} \frac{1}{c_3} \frac{\delta}{2c_4} \frac{1}{\sqrt{\frac{\mu_1^\ast}{2} + \frac{\mu_2^2}{4}}} \\
\Rightarrow & u_3 = O(\delta).
\end{aligned}
\end{equation}

We have $Q= \{\vec{u}\colon 0 \le u_1 \le u_2 \le u_3 \le 1\}$.  
Define the set  
\begin{equation}
A_{\delta}
=
\bigl\{
\vec{u}\in Q: \beta_2^{\mu_1^\ast+\delta}(\vec{u})\neq 
\beta_2^{\mu_1^\ast}(\vec{u})
\bigr\}.
\end{equation}
From \eqref{eq:g3_bound_mu1}, for any fixed $u_1$ and $u_2$ values in the interval $[0,1]$, the possible $u_3$ value lies within a one-dimensional region whose measure is of order $\delta$. 
Since both $u_1$ and $u_2$ lie in $[0,1]$, the measure of the three-dimensional set $A_{\delta}$ is of order $\delta$.  
The difference  
\begin{equation}
\beta_2^{\mu_1^\ast+\delta}(\vec{u}) 
-
\beta_2^{\mu_1^\ast}(\vec{u}) 
\end{equation}
takes values in $\{-1, 0,1\}$. Hence, 
\begin{equation}
\Bigl\lvert
\beta_2^{\mu_1^\ast+\delta}(\vec{u})
-
\beta_2^{\mu_1^\ast}(\vec{u})
\Bigr\rvert
\le1
\quad
\text{for all } \vec{u}\in Q.
\end{equation}
Therefore, by definition of $A_{\delta}$,  
\begin{equation}
\left\lvert
\int_Q
\Bigl(
\beta_2^{\mu_1^\ast+\delta}(\vec{u})
-
\beta_2^{\mu_1^\ast}(\vec{u})
\Bigr)
d\vec{u}
\right\rvert
\le
\int_{A_{\delta}} 1d\vec{u}
=
\mathrm{measure}(A_{\delta})
=
O(\delta).
\end{equation}
Thus, we have the following.
\begin{equation}\label{eq:b2_Odelta_mu1_case2}
\int_Q
\Bigl(
\alpha_2^{\mu_1^\ast}(\vec{u})
-
\alpha_2^{\mu_1^\ast+\delta}(\vec{u})
\Bigr)
d\vec{u}
=
O(\delta)
\quad\Longleftrightarrow\quad
\int_Q
\Bigl(
\beta_2^{\mu_1^\ast+\delta}(\vec{u})
-
\beta_2^{\mu_1^\ast}(\vec{u})
\Bigr)
d\vec{u}
=
O(\delta).
\end{equation}
Then the second term in \eqref{eq:mu1_min} becomes
{\small
\begin{equation}
\begin{aligned}
&\int_Q \bigg[
    g(u_1) g(u_2) g(u_3) \bigl(1 + \alpha_3(\vec{u}) \bigr)
    - g(u_1)\mu_1^\ast
    - \bigl(g(u_{1})g(u_{3}) + \alpha_{3}(\vec{u}) g(u_{1}) g(u_{2})\bigr)\mu_2
\bigg]
\bigl[\alpha_2^{\mu_1^\ast+\delta}(\vec{u}) - \alpha_2^{\mu_1^\ast}(\vec{u}) \bigr] \, d\vec{u} \\
&\le 
\biggl\|
    g(u_1) g(u_2) g(u_3) \bigl(1 + \alpha_3(\vec{u}) \bigr)
    - g(u_1)\mu_1^\ast \\
&\qquad
    - \bigl(g(u_{1})g(u_{3}) + \alpha_{3}(\vec{u}) g(u_{1}) g(u_{2})\bigr)\mu_2
\biggr\|_\infty
\int_Q \bigl[\alpha_2^{\mu_1^\ast+\delta}(\vec{u}) - \alpha_2^{\mu_1^\ast}(\vec{u}) \bigr] \, d\vec{u} \\
&\stackrel{\eqref{eq:b2_Odelta_mu1_case2},\eqref{eq:Odelta_case2_mu1min_1stintegrand}}{=}
O(\delta^2) \\
&= o(\delta).
\end{aligned}
\end{equation}}

The condition \eqref{eq:mu1_optimality} thus follows from Proposition \eqref{prop:o_delta_condition_mu_1}.
This completes the proof.

\end{enumerate}

Condition \eqref{eq:mu1_optimality} can be further simplified as follows:
\begin{equation}
\label{eq:mu1_simplifiedcondition}
\begin{aligned}
\alpha
=&2\int_{Q}\!\Bigl[g(u_{1})\bigl(1-\beta_{2}^{\mu_1^\ast}(\vec u)\bigr)+g(u_{2})+g(u_{3})\Bigr]d\vec u \\
=&2\int_{Q}\!\bigl[g(u_{1})+g(u_{2})+g(u_{3})\bigr]d\vec u
  -2\int_{Q}\!g(u_{1})\beta_{2}^{\mu_1^\ast}(\vec u)d\vec u \\
=&2\!\bigg[\!
     \int_{0}^{1}\!g(u_{1})
        \bigg(\int_{u_{1}}^{1}\!\!\int_{u_{2}}^{1}\!du_{3}du_{2}\bigg)du_{1}
     +\int_{0}^{1}\!g(u_{2})
        \bigg(\int_{0}^{u_{2}}\!\!\int_{u_{2}}^{1}\!du_{3}du_{1}\bigg)du_{2} +\int_{0}^{1}\!g(u_{3})
        \bigg(\int_{0}^{u_{3}}\!\!\int_{0}^{u_{2}}\!du_{1}du_{2}\bigg)du_{3}
   \bigg]\\
  &-2\int_{Q}\!g(u_{1})\beta_{2}^{\mu_1^\ast}(\vec u)d\vec u \\
=&2\!\bigg[
     \int_{0}^{1}\!g(u)\tfrac{(1-u)^{2}}{2}du
     +\int_{0}^{1}\!g(u)u(1-u)du
     +\int_{0}^{1}\!g(u)\tfrac{u^{2}}{2}du
   \bigg]
  -2\int_{Q}\!g(u_{1})\beta_{2}^{\mu_1^\ast}(\vec u)d\vec u \\
=&2\!\bigg[
     \tfrac12\int_{0}^{1}\!g(u)\bigl[(1-u)^{2}+2u(1-u)+u^{2}\bigr]du
   \bigg]
  -2\int_{Q}\!g(u_{1})\beta_{2}^{\mu_1^\ast}(\vec u)d\vec u \\
=&2\!\bigg[\tfrac12\int_{0}^{1}\!g(u)du\bigg]
  -2\int_{Q}\!g(u_{1})\beta_{2}^{\mu_1^\ast}(\vec u)d\vec u \\
=&1-2\int_{Q}\!g(u_{1})\beta_{2}^{\mu_1^\ast}(\vec u)d\vec u .
\end{aligned}
\end{equation}
\end{proof}

\subsubsection{\texorpdfstring{Sufficient condition for characterizing the minimizer $\mu_1^\ast$}{Sufficient condition for characterizing the minimizer mu1*}}
The following proposition states a sufficient condition under which a simple integral identity fully characterizes the minimizer $\mu_1^\ast$ as derived in the above theorem and serves as a key stepping stone in that proof. 
This allows the proof of Theorem \ref{thm:mu1_optimality} to trace a single necessary equality in \eqref{eq:mu1_optimality}.
\begin{prop}
\label{prop:o_delta_condition_mu_1}
Let $\mu_1^\ast \in \mathbb{R}_+$ and fix $\mu_0, \mu_2$. Suppose
{\small
\begin{equation}\label{eq:o_delta_condition_integral_mu1}
\int_Q \bigg[g(u_1) g(u_2) g(u_3) \biggl(1 + \alpha_3(\vec{u}) \biggr)- g(u_1)\mu_1^\ast -\bigg(g(u_{1})g(u_{3})+\alpha_{3}(\vec{u})g(u_{1})g(u_{2})\bigg)\mu_2 \bigg] \bigg[\alpha_2^{\mu_1^\ast+\delta}(\vec{u})-\alpha_2^{\mu_1^\ast}(\vec{u}) \bigg]d\vec{u}=o(\delta).
\end{equation}}
Then \eqref{eq:mu1_optimality} holds.
\end{prop}

\begin{proof}
Suppose \eqref{eq:o_delta_condition_integral_mu1} holds.
We then demonstrate that then $\mu_1^\ast$ is the minimizer of $L(\vec{D}^\mu, \mu)$ of the dual problem \eqref{eq:dual2} w.r.t. $\mu_1$ only if $\alpha = 1 - 2 \int_Q g(u_1) \beta_2^{\mu_1^\ast}(\vec{u}) d\vec{u}$.

From \eqref{eq:mu1_min} we have that if $\mu_1^\ast$ is the minimizer,
{\footnotesize
\begin{equation}\label{eq:mu_1_must_hold}
\begin{aligned}
0&\le\delta \bigg[\alpha - 2\int_Q \bigg(g(u_1)\alpha_2^{\mu_1^\ast+\delta}(\vec{u})+g(u_2)+ g(u_3) \bigg)d\vec{u} \bigg] \\
&+2\int_Q \bigg[g(u_1) g(u_2) g(u_3) \biggl(1 + \alpha_3(\vec{u}) \biggr)- g(u_1)\mu_1^\ast -\bigg(g(u_{1})g(u_{3})+\alpha_{3}(\vec{u})g(u_{1})g(u_{2})\bigg)\mu_2 \bigg] \bigg[\alpha_2^{\mu_1^\ast+\delta}(\vec{u})-\alpha_2^{\mu_1^\ast}(\vec{u}) \bigg]d\vec{u}\\
&=\delta \tau + o(\delta),\ \forall \delta 
\end{aligned}
\end{equation}}
where, 
\begin{equation*}
\tau=\alpha - 2\int_Q \bigg(g(u_1)\alpha_2^{\mu_1^\ast+\delta}(\vec{u})+g(u_2)+ g(u_3) \bigg)d\vec{u}
\end{equation*} 
and using \eqref{eq:o_delta_condition_integral_mu1}, 
{\small
\begin{equation*}
\int_Q \bigg[g(u_1) g(u_2) g(u_3) \biggl(1 + \alpha_3(\vec{u}) \biggr)- g(u_1)\mu_1^\ast -\bigg(g(u_{1})g(u_{3})+\alpha_{3}(\vec{u})g(u_{1})g(u_{2})\bigg)\mu_2 \bigg] \bigg[\alpha_2^{\mu_1^\ast+\delta}(\vec{u})-\alpha_2^{\mu_1^\ast}(\vec{u}) \bigg]d\vec{u}=o(\delta)
\end{equation*}} 

We use a contradiction argument to prove that the condition $\alpha = 2\int_Q \bigg(g(u_1)\alpha_2^{\mu_1^\ast+\delta}(\vec{u})+g(u_2)+ g(u_3) \bigg)d\vec{u}$ must hold.
Suppose \begin{equation}
\label{eq:tau>0_mu1}
\tau =\alpha - 2\int_Q \bigg(g(u_1)\alpha_2^{\mu_1^\ast+\delta}(\vec{u})+g(u_2)+ g(u_3) \bigg)d\vec{u} > 0.
\end{equation}
Then choosing $\delta$ small enough for some $\delta<0$, we have using \eqref{eq:o_delta_condition_integral_mu1} and \eqref{eq:tau>0_mu1},
\begin{equation*}
 \delta\tau + o(\delta)<0
\end{equation*}
which contradicts \eqref{eq:mu_1_must_hold} which states that $\forall \delta$, $$ \delta\tau + o(\delta)\ge 0. $$
Therefore, \eqref{eq:tau>0_mu1} cannot be true.
A similar argument holds if $\tau<0$, and it cannot be true either.
So we must have,
\begin{equation*}
\tau =  \alpha - 2\int_Q \bigg(g(u_1)\alpha_2^{\mu_1^\ast+\delta}(\vec{u})+g(u_2)+ g(u_3) \bigg)d\vec{u} \stackrel{\mbox{\eqref{eq:mu1_simplifiedcondition}}}{=} \alpha - 1 - 2 \int_Q g(u_1) \beta_2^{\mu_1^\ast}(\vec{u}) d\vec{u} = 0.
\end{equation*}
Thus, we have proved that $\mu_1^\ast$ is the minimizer of $L(\vec{D}^\mu, \mu)$ of the dual problem \eqref{eq:dual2} w.r.t. $\mu_1$ only if $\alpha = 1 - 2 \int_Q g(u_1) \beta_2^{\mu_1^\ast}(\vec{u}) d\vec{u}$ i.e., \eqref{eq:mu1_optimality} is true.
\end{proof}

\subsection{\texorpdfstring{Minimizing the Lagrangian w.r.t. $\mu_2$}{Minimizing the Lagrangian w.r.t. mu2}}
\label{app:mu2_optimality}
Finally, we have the Lagrangian as in \eqref{eq:lagrangian5_general}, and our objective is to minimize it w.r.t. $\mu_2$.
In this case, we only need to consider the case  
\begin{equation}\label{eq:alpha_1_equals_1_mu2}
\alpha_1^{\mu}(\vec{u})=1,
\end{equation}
because if $\alpha_1^{\mu}(\vec{u})=0$, then
\begin{equation*}
L(\vec{D}^\mu, \mu) \stackrel{\mbox{\eqref{eq:lagrangian5_general}}}{=} \alpha(\mu_0+\mu_1+\mu_2)
\end{equation*}
therefore the minimimizer of $L(\vec{D}^\mu, \mu)$ w.r.t. $\mu_2$ can't be optimal until it is zero.
Thus, we assume $\alpha_1^{\mu}(\vec{u})=1$.
Plugging this in the Lagrangian \eqref{eq:lagrangian5_general} we have,
\begin{equation}\label{eq:lagrangian5_mu2}
\begin{aligned}
L(\vec{D}^\mu, \mu) \stackrel{\mbox{\eqref{eq:lagrangian5_general},\eqref{eq:alpha_1_equals_1_mu1}}}{=}&\alpha(\mu_0+\mu_1+\mu_2)\\
&+2\int_Q g(u_1) g(u_2) g(u_3) \biggl(1 + \alpha_2^{\mu}(\vec{u})+\alpha_2^{\mu}(\vec{u})\alpha_3^{\mu}(\vec{u}) \biggr)d\vec{u}\\
&-6\int_Q \mu_0d\vec{u} \\
&-2\int_Q \Bigl(g(u_1)\alpha_2^{\mu}(\vec{u})+g(u_2) + g(u_3)\Bigr)\mu_1d\vec{u}\\
& -2\int_Q \Bigl(g(u_{2})g(u_{3})+\alpha_{2}^{\mu}(\vec{u})g(u_{1})g(u_{3})+\alpha_{2}^{\mu}(\vec{u})\alpha_{3}^{\mu}(\vec{u})g(u_{1})g(u_{2})\Bigr)\mu_2 d\vec{u}.
\end{aligned}
\end{equation}
Given the form of the Lagrangian as presented in the equation above, we observe that if $\alpha_{2}^{\mu}(\vec{u}) = 0$, then the Lagrangian in \eqref{eq:lagrangian5_mu2} simplifies to the following:
{\small
\begin{equation}\label{eq:lagrangian5_mu2_alpha2_0}
\begin{aligned}
L(\vec{D}^\mu, \mu) \stackrel{\mbox{\eqref{eq:lagrangian5_mu2}}}{=}&\alpha(\mu_0+\mu_1+\mu_2) +2\int_Q g(u_1) g(u_2) g(u_3) d\vec{u} -6\int_Q \mu_0d\vec{u}\\
&-2\int_Q \Bigl(g(u_2) + g(u_3)\Bigr)\mu_1d\vec{u}  -2\int_Q g(u_{2})g(u_{3})\mu_2 d\vec{u} \\
=& \mu_2\bigg(\alpha-2\int_Q g(u_{2})g(u_{3}) d\vec{u} \bigg) \\
&+ \bigg(\alpha(\mu_0+\mu_1) +2\int_Q g(u_1) g(u_2) g(u_3) d\vec{u} -6\int_Q \mu_0d\vec{u} -2\int_Q \Bigl(g(u_2) + g(u_3)\Bigr)\mu_1d\vec{u} \bigg)\\
\end{aligned}
\end{equation}}
Observe that the second term in the above expression is constant with respect to $\mu_2$. 
We now analyze the first term, specifically the coefficient of $\mu_2$, which is given by $\alpha - 2\int_Q g(u_{2})g(u_{3}) d\vec{u}$. 
Here, $\alpha$ denotes the user-specified error threshold, while $2\int_Q g(u_{2})g(u_{3}) d\vec{u}$ is a fixed constant. 
Consequently, the coefficient $\alpha - 2\int_Q g(u_{2})g(u_{3}) d\vec{u}$ is itself a constant that may be either positive or negative.  
If this coefficient is positive, then the minimizer of the Lagrangian in \eqref{eq:lagrangian5_mu2_alpha2_0} with respect to $\mu_2$ cannot be optimal unless $\mu_2 = 0$. 
Conversely, if the coefficient is negative, the minimizer cannot be optimal unless $\mu_2 = \infty$.
Thus we assume 
\begin{equation}\label{eq:alpha_2_equals_1_mu2}
\alpha_2^{\mu}(\vec{u})=1.
\end{equation}
Plugging this in the Lagrangian \eqref{eq:lagrangian5_mu2} we have,
\begin{equation}\label{eq:lagrangian5_mu2_alpha2_1}
\begin{aligned}
L(\vec{D}^\mu, \mu) \stackrel{\mbox{\eqref{eq:lagrangian5_mu2},\eqref{eq:alpha_2_equals_1_mu2}}}{=}&\alpha(\mu_0+\mu_1+\mu_2)\\
&+2\int_Q g(u_1) g(u_2) g(u_3) \biggl(2+\alpha_3^{\mu}(\vec{u}) \biggr)d\vec{u}\\
&-6\int_Q \mu_0d\vec{u} \\
&-2\int_Q \Bigl(g(u_1) +g(u_2) + g(u_3)\Bigr)\mu_1d\vec{u}\\
& -2\int_Q \Bigl(g(u_{2})g(u_{3})+ g(u_{1})g(u_{3})+\alpha_{3}^{\mu}(\vec{u})g(u_{1})g(u_{2})\Bigr)\mu_2 d\vec{u}.
\end{aligned}
\end{equation}
Note that on the right-hand side of the above equation, only $\alpha_{3}^{\mu}(\vec{u})$ depends on $\mu_2$.
We will be working with this expression of $L(\vec{D}^\mu, \mu)$ for the remainder of this section.

Since we are minimizing $L(\vec{D}^\mu, \mu)$ w.r.t. $\mu_2$ we denote it as $L(\mu_2)$. 
Also since only $\alpha_{3}^{\mu}(\vec{u})$ depends on $\mu_2$ we denote $\alpha_3^{\mu}(\vec{u})$ only by $\alpha_3^{\mu_2}(\vec{u})$. 

Thus, we have that if $\mu_2^\ast$ is the minimizer, then for any $\delta\neq0$ we have $$L(\mu_2^\ast) \le L(\mu_2^\ast+\delta).$$
That is,
{\footnotesize
\begin{equation}\label{eq:mu2_min}
\begin{aligned}
0 &\le L(\mu_2^\ast+\delta)- L(\mu_2^\ast)\\
&\stackrel{\mbox{\eqref{eq:lagrangian5_mu2_alpha2_1}}}{=} \alpha\delta +2\int_Q g(u_1) g(u_2) g(u_3) \biggl(\alpha_3^{\mu_2^\ast+\delta}(\vec{u}) - \alpha_3^{\mu_2^\ast}(\vec{u}) \biggr)d\vec{u}\\
& -2\int_Q \Bigl(g(u_{2})g(u_{3})+ g(u_{1})g(u_{3})+\alpha_{3}^{\mu_2^\ast+\delta}(\vec{u})g(u_{1})g(u_{2})\Bigr)(\mu_2^\ast+\delta) d\vec{u} \\
& +2\int_Q \Bigl(g(u_{2})g(u_{3})+ g(u_{1})g(u_{3})+\alpha_{3}^{\mu_2^\ast}(\vec{u})g(u_{1})g(u_{2})\Bigr)\mu_2^\ast d\vec{u} \\
=& \alpha\delta -2\int_Q \Bigl(g(u_{2})g(u_{3})+ g(u_{1})g(u_{3})+\alpha_{3}^{\mu_2^\ast+\delta}(\vec{u})g(u_{1})g(u_{2})\Bigr)\delta d\vec{u} + 2\int_Q g(u_1) g(u_2) g(u_3) \biggl(\alpha_3^{\mu_2^\ast+\delta}(\vec{u}) - \alpha_3^{\mu_2^\ast}(\vec{u}) \biggr)d\vec{u}\\
& -2\int_Q \Bigl(g(u_{2})g(u_{3})+ g(u_{1})g(u_{3})+\alpha_{3}^{\mu_2^\ast+\delta}(\vec{u})g(u_{1})g(u_{2})\Bigr)\mu_2^\ast d\vec{u}  \\
&+2\int_Q \Bigl(g(u_{2})g(u_{3})+ g(u_{1})g(u_{3})+\alpha_{3}^{\mu_2^\ast}(\vec{u})g(u_{1})g(u_{2})\Bigr)\mu_2^\ast d\vec{u} \\
=&  \delta \bigg[\alpha -2\int_Q \Bigl(g(u_{2})g(u_{3})+ g(u_{1})g(u_{3})+\alpha_{3}^{\mu_2^\ast+\delta}(\vec{u})g(u_{1})g(u_{2})\Bigr) d\vec{u}\bigg] + 2\int_Q g(u_1) g(u_2) g(u_3) \biggl(\alpha_3^{\mu_2^\ast+\delta}(\vec{u}) - \alpha_3^{\mu_2^\ast}(\vec{u}) \biggr)d\vec{u}\\
& -2\int_Q \Bigl(\alpha_{3}^{\mu_2^\ast+\delta}(\vec{u})g(u_{1})g(u_{2}) - \alpha_{3}^{\mu_2^\ast}(\vec{u})g(u_{1})g(u_{2}\Bigr)\mu_2^\ast d\vec{u} \\
=&  \delta \bigg[\alpha -2\int_Q \Bigl(g(u_{2})g(u_{3})+ g(u_{1})g(u_{3})+\alpha_{3}^{\mu_2^\ast+\delta}(\vec{u})g(u_{1})g(u_{2})\Bigr) d\vec{u}\bigg] + 2\int_Q g(u_1) g(u_2) g(u_3) \biggl(\alpha_3^{\mu_2^\ast+\delta}(\vec{u}) - \alpha_3^{\mu_2^\ast}(\vec{u}) \biggr)d\vec{u}\\
& -2\int_Q g(u_{1})g(u_{2})\mu_2^\ast \bigg(\alpha_3^{\mu_2^\ast+\delta}(\vec{u}) - \alpha_3^{\mu_2^\ast}(\vec{u}) \bigg) d\vec{u} \\
=& \delta \bigg[\alpha -2\int_Q \Bigl(g(u_{2})g(u_{3})+ g(u_{1})g(u_{3})+\alpha_{3}^{\mu_2^\ast+\delta}(\vec{u})g(u_{1})g(u_{2})\Bigr) d\vec{u}\bigg] \\
&+ 2\int_Q \bigg[g(u_1) g(u_2) g(u_3) - g(u_{1})g(u_{2})\mu_2^\ast \bigg] \bigg[\alpha_3^{\mu_2^\ast+\delta}(\vec{u}) - \alpha_3^{\mu_2^\ast}(\vec{u}) \bigg]d\vec{u}
\end{aligned}
\end{equation}}

\begin{theorem}\label{thm:mu2_optimality}
Suppose Assumptions~\ref{as:assumption3},~\ref{as:assumption4}, and~\ref{as:assumption5} hold.  
Consider the dual formulation \eqref{eq:dual2} of the multiple hypotheses testing problem \eqref{eq:objconst}, and the associated Lagrangian $L(\vec{D}^\mu, \mu)$ as defined in \eqref{eq:lagrangian5_general}.  
Then, for fixed values of the other coordinates $\mu_0$ and $\mu_1$, the coordinate $\mu_2^\ast \in \mathbb{R}_+$ is a minimizer of $L(\vec{D}^\mu, \mu)$ with respect to $\mu_2$ only if
\begin{equation}
\label{eq:mu2_optimality}
\alpha = 2 \int_Q \alpha_3^{\mu_2^\ast}(\vec{u}) g(u_1)g(u_2) d\vec{u}.
\end{equation}
\end{theorem}

\begin{proof}

We have for $\delta>0$ using \eqref{eq:beta_mu_u},
{\small
\begin{equation}\label{eq:b3_diff}
\begin{aligned}
&\beta_3^{\mu_2^\ast+\delta}(\vec{u})-\beta_3^{\mu_2^\ast}(\vec{u})\\
\stackrel{\mbox{\eqref{eq:beta_mu_u}}}{=}& \mathbbm{1}\biggl\{\mu,\vec{u}:2g(u_1) g(u_2) g(u_3) \le 2g(u_1)g(u_2)(\mu_2^\ast+\delta) \biggr\}-\mathbbm{1}\biggl\{\mu,\vec{u}:2g(u_1) g(u_2) g(u_3) \le 2g(u_1)g(u_2)\mu_2^\ast \biggr\} \\
=& \mathbbm{1}\biggl\{\mu,\vec{u}: 2g(u_1)g(u_2)\mu_2^\ast \le 2g(u_1) g(u_2) g(u_3) \le 2g(u_1)g(u_2)(\mu_2^\ast+\delta) \biggr\} \\
=& \mathbbm{1}\biggl\{\mu,\vec{u}: 0 \le 2g(u_1) g(u_2) g(u_3) - 2g(u_1)g(u_2)\mu_2^\ast \le 2g(u_1)g(u_2)\delta \biggr\}
\end{aligned}
\end{equation}}
From \eqref{eq:b3_diff} we have the following condition as true,
\begin{equation}\label{eq:beta3_diff_condition}
0 \le 2g(u_1) g(u_2) g(u_3) - 2g(u_1)g(u_2)\mu_2^\ast \le 2g(u_1)g(u_2)\delta.
\end{equation}

Fix $u_1$,$u_2$. 
Denote by $u_{3,\text{min}}$ and $u_{3,\text{max}}$ the minimum and maximum value of $u_3$. 
Then we have using \eqref{eq:beta3_diff_condition},
\begin{equation}\label{eq:b3_1st_inequality}
\begin{aligned}
& 0 \le 2g(u_1) g(u_2) g(u_3) - 2g(u_1)g(u_2)\mu_2^\ast \le 2g(u_1)g(u_2)\delta \\
\Rightarrow& 0 \le g(u_3) - \mu_2^\ast \le \delta \\
\Rightarrow& |g(u_{3,\text{min}})-g(u_{3,\text{max}})|\le \delta
\end{aligned}
\end{equation}

Then we have for $u_3$,
\begin{equation}\label{eq:u3_bound_mu2}
|u_3-u_3'|  \stackrel{\mbox{Assumption \ref{as:assumption3}}}{\le} \frac{1}{c_3} |g(u_3)-g(u_3')| \stackrel{\mbox{\eqref{eq:b3_1st_inequality}}}{\le} \frac{\delta}{c_3} 
\end{equation}

We have $Q= \{\vec{u}\colon 0 \le u_1 \le u_2 \le u_3 \le 1\}$.  
Define the set  
\begin{equation}
A_{\delta}
=
\bigl\{
\vec{u}\in Q: \beta_3^{\mu_2^\ast+\delta}(\vec{u})\neq 
\beta_3^{\mu_2^\ast}(\vec{u})
\bigr\}.
\end{equation}
From \eqref{eq:u3_bound_mu2}, for any fixed $u_1$ and $u_2$ values in the interval $[0,1]$, the possible $u_3$ value lies within a one-dimensional region whose measure is of order $\delta$. 
Since both $u_1$ and $u_2$ lie in $[0,1]$, the measure of the three-dimensional set $A_{\delta}$ is of order $\delta$.  
The difference  
\begin{equation}
\beta_3^{\mu_2^\ast+\delta}(\vec{u}) 
-
\beta_3^{\mu_2^\ast}(\vec{u}) 
\end{equation}
takes values in $\{0,1\}$. Hence, 
\begin{equation}
\Bigl\lvert
\beta_3^{\mu_2^\ast+\delta}(\vec{u})
-
\beta_3^{\mu_2^\ast}(\vec{u})
\Bigr\rvert
\le1
\quad
\text{for all } \vec{u}\in Q.
\end{equation}
Therefore, by definition of $A_{\delta}$,  
\begin{equation}
\left\lvert
\int_Q
\Bigl(
\beta_3^{\mu_2^\ast+\delta}(\vec{u})
-
\beta_3^{\mu_2^\ast}(\vec{u})
\Bigr)
d\vec{u}
\right\rvert
\le
\int_{A_{\delta}} 1d\vec{u}
=
\mathrm{measure}(A_{\delta})
=
O(\delta).
\end{equation}
Thus, we have that,
\begin{equation}\label{eq:b3_Odelta}
\int_Q
\Bigl(
\alpha_3^{\mu_2^\ast}(\vec{u})
-
\alpha_3^{\mu_2^\ast+\delta}(\vec{u})
\Bigr)
d\vec{u}
=
O(\delta)
\quad\Longleftrightarrow\quad
\int_Q
\Bigl(
\beta_3^{\mu_2^\ast+\delta}(\vec{u})
-
\beta_3^{\mu_2^\ast}(\vec{u})
\Bigr)
d\vec{u}
=
O(\delta).
\end{equation}

Next, we analyze the integral in the right hand side of \eqref{eq:mu2_min}:
\begin{equation}
\label{eq:mu2_min_conjecture2}
2\int_Q \bigg[g(u_1) g(u_2) g(u_3) - g(u_{1})g(u_{2})\mu_2^\ast \bigg] \bigg[\alpha_3^{\mu_2^\ast+\delta}(\vec{u}) - \alpha_3^{\mu_2^\ast}(\vec{u}) \bigg]d\vec{u}.
\end{equation}
We consider the first integrand of the integral of \eqref{eq:mu2_min_conjecture2}: 
From \eqref{eq:b3_Odelta}, we already know that  
\begin{equation*}
\int_Q
\Bigl(
\alpha_3^{\mu_2^\ast+\delta}(\vec{u})
-
\alpha_3^{\mu_2^\ast}(\vec{u})
\Bigr)
d\vec{u}
=
O(\delta).
\end{equation*}
If we can establish that the first integrand in \eqref{eq:mu2_min_conjecture2} is also $O(\delta)$, it follows that the entire integral is $o(\delta)$.
To show this, we begin with \eqref{eq:beta3_diff_condition}:
\begin{equation}\label{eq:beta3_diff_firstfactor_Odelta}
\begin{aligned}
& 0 \le 2g(u_1) g(u_2) g(u_3) - 2g(u_1)g(u_2)\mu_2^\ast \le 2g(u_1)g(u_2)\delta \\
\Leftrightarrow & 0 \le g(u_1) g(u_2) g(u_3) - g(u_1)g(u_2)\mu_2^\ast \le g(u_1)g(u_2)\delta \stackrel{\mbox{Assumption \ref{as:assumption5}}}{\le} c_5^2 \delta
\end{aligned}
\end{equation}

Thus we have using \eqref{eq:beta3_diff_firstfactor_Odelta}
\begin{equation}\label{eq:beta3_diff_firstfactor_Odelta_result}
\bigg[g(u_1) g(u_2) g(u_3) - g(u_{1})g(u_{2})\mu_2^\ast \bigg]=O(\delta).
\end{equation}
Thus have,
\begin{equation}\label{eq:mu2_min_conjecture2_proven}
\begin{aligned}
&\int_Q \bigg[g(u_1) g(u_2) g(u_3) - g(u_{1})g(u_{2})\mu_2^\ast \bigg] \bigg[\alpha_3^{\mu_2^\ast+\delta}(\vec{u}) - \alpha_3^{\mu_2^\ast}(\vec{u}) \bigg]d\vec{u}\\
\le & \|g(u_1) g(u_2) g(u_3) - g(u_{1})g(u_{2})\mu_2^\ast)\|_{\infty} \int_Q\bigg(\alpha_3^{\mu_2^\ast}(\vec{u})-\alpha_3^{\mu_2^\ast+\delta}(\vec{u})\bigg)d\vec{u}\\
\stackrel{\mbox{\eqref{eq:b3_Odelta},\eqref{eq:beta3_diff_firstfactor_Odelta_result}}}{=}& O(\delta^2)\\
=& o(\delta) .
\end{aligned}
\end{equation}
\eqref{eq:mu2_optimality} follows using proposition \ref{prop:o_delta_condition_mu_2}.

\eqref{eq:mu2_optimality} can be further simplified as:
\begin{equation}
\label{eq:mu2_simplifiedcondition}
\begin{aligned}
\alpha
&=2\!\int_{Q}\!
      \Bigl[g(u_{2})g(u_{3})
            +g(u_{1})g(u_{3})
            +\alpha_{3}^{\mu_2^\ast}(\vec u)g(u_{1})g(u_{2})\Bigr]d\vec u \\
&=2\!\int_{Q}\!
      \Bigl[g(u_{2})g(u_{3})
            +g(u_{1})g(u_{3})
            +\bigl(1-\beta_{3}^{\mu_2^\ast}(\vec u)\bigr)g(u_{1})g(u_{2})\Bigr]d\vec u \\
&=2\!\int_{Q}\!
      \Bigl[g(u_{2})g(u_{3})
            +g(u_{1})g(u_{3})
            +g(u_{1})g(u_{2})\Bigr]d\vec u
  -2\!\int_{Q}\!\beta_{3}^{\mu_2^\ast}(\vec u)g(u_{1})g(u_{2})d\vec u \\
\end{aligned}
\end{equation}
This completes the proof.
\end{proof}

\subsubsection{\texorpdfstring{Sufficient condition for characterizing the minimizer $\mu_2^\ast$}{Sufficient condition for characterizing the minimizer mu2*}}
The following proposition states a sufficient condition under which a simple integral identity fully characterizes the minimizer $\mu_2^\ast$ as derived in the above theorem and serves as a key stepping stone in that proof. 
This allows the proof of Theorem \ref{thm:mu2_optimality} to trace a single necessary equality in \eqref{eq:mu2_optimality}.
\begin{prop}
\label{prop:o_delta_condition_mu_2}
Let $\mu_2^\ast \in \mathbb{R}_+$ and fix $\mu_0, \mu_1$. Suppose
\begin{equation}\label{eq:o_delta_condition_integral_mu2}
\int_Q \bigg[g(u_1) g(u_2) g(u_3) - g(u_{1})g(u_{2})\mu_2^\ast \bigg] \bigg[\alpha_3^{\mu_2^\ast+\delta}(\vec{u}) - \alpha_3^{\mu_2^\ast}(\vec{u}) \bigg]d\vec{u} =o(\delta).
\end{equation}
Then \eqref{eq:mu2_optimality} holds.
\end{prop}

\begin{proof}
Suppose \eqref{eq:o_delta_condition_integral_mu2} holds.
We then demonstrate that then $\mu_2^\ast$ is the minimizer of $L(\vec{D}^\mu, \mu)$ of the dual problem \eqref{eq:dual2} w.r.t. $\mu_2$ only if $\alpha = 2 \int_Q \alpha_3^{\mu_2^\ast}(\vec{u}) g(u_1)g(u_2) d\vec{u}$.

From \eqref{eq:mu2_min} we have that if $\mu_2^\ast$ is the minimizer,
\begin{equation}\label{eq:mu_2_must_hold}
\begin{aligned}
0&\le \delta \bigg[\alpha -2\int_Q \Bigl(g(u_{2})g(u_{3})+ g(u_{1})g(u_{3})+\alpha_{3}^{\mu_2^\ast+\delta}(\vec{u})g(u_{1})g(u_{2})\Bigr) d\vec{u}\bigg] \\
&+ 2\int_Q \bigg[g(u_1) g(u_2) g(u_3) - g(u_{1})g(u_{2})\mu_2^\ast \bigg] \bigg[\alpha_3^{\mu_2^\ast+\delta}(\vec{u}) - \alpha_3^{\mu_2^\ast}(\vec{u}) \bigg]d\vec{u} \\
&=\delta \tau + o(\delta),\ \forall \delta 
\end{aligned}
\end{equation}
where, 
\begin{equation*}
\tau=\alpha -2\int_Q \Bigl(g(u_{2})g(u_{3})+ g(u_{1})g(u_{3})+\alpha_{3}^{\mu_2^\ast+\delta}(\vec{u})g(u_{1})g(u_{2})\Bigr) d\vec{u}
\end{equation*} 
and using \eqref{eq:o_delta_condition_integral_mu2}, 
\begin{equation*}
\int_Q \bigg[g(u_1) g(u_2) g(u_3) - g(u_{1})g(u_{2})\mu_2^\ast \bigg] \bigg[\alpha_3^{\mu_2^\ast+\delta}(\vec{u}) - \alpha_3^{\mu_2^\ast}(\vec{u}) \bigg]d\vec{u}=o(\delta)
\end{equation*} 

We use a contradiction argument to prove that the condition $$\alpha = 2\int_Q \Bigl(g(u_{2})g(u_{3})+ g(u_{1})g(u_{3})+\alpha_{3}^{\mu_2^\ast+\delta}(\vec{u})g(u_{1})g(u_{2}) \bigg)d\vec{u}$$ must hold.
Suppose \begin{equation}
\label{eq:tau>0_mu2}
\tau =\alpha -2\int_Q \Bigl(g(u_{2})g(u_{3})+ g(u_{1})g(u_{3})+\alpha_{3}^{\mu_2^\ast+\delta}(\vec{u})g(u_{1})g(u_{2}) > 0.
\end{equation}
Then choosing $\delta$ small enough for some $\delta<0$, we have using \eqref{eq:o_delta_condition_integral_mu2} and \eqref{eq:tau>0_mu2},
\begin{equation*}
 \delta\tau + o(\delta)<0
\end{equation*}
which contradicts \eqref{eq:mu_2_must_hold} which states that $\forall \delta$, $$ \delta\tau + o(\delta)\ge 0. $$
Therefore, \eqref{eq:tau>0_mu2} cannot be true.
A similar argument holds if $\tau<0$, and it cannot be true either.
So we must have,
{\small
\begin{equation*}
\tau =  \alpha - 2\int_Q \Bigl(g(u_{2})g(u_{3})+ g(u_{1})g(u_{3})+\alpha_{3}^{\mu_2^\ast+\delta}(\vec{u})g(u_{1})g(u_{2}) \bigg)d\vec{u} \stackrel{\mbox{\eqref{eq:mu2_simplifiedcondition}}}{=} \alpha - 2 \int_Q \alpha_3^{\mu_2^\ast}(\vec{u}) g(u_1)g(u_2) d\vec{u} = 0.
\end{equation*}}
Thus, we have proved that $\mu_2^\ast$ is the minimizer of $L(\vec{D}^\mu, \mu)$ of the dual problem \eqref{eq:dual2} w.r.t. $\mu_2$ only if $\alpha = 2 \int_Q \alpha_3^{\mu_2^\ast}(\vec{u}) g(u_1)g(u_2) d\vec{u}$ i.e., \eqref{eq:mu2_optimality} is true.
\end{proof}

\subsection{\texorpdfstring{Optimality of the decision policy $\vec{D}^{\mu^*}$}{Optimality of the decision policy D(mu*)}}
\label{sec:d_mu_optimality}
In this subsection, we establish that the decision policy $\vec{D}^{\mu^*}$, parameterized by the Lagrange multipliers satisfying the conditions in Theorem \ref{thm:mu_optimality_combined}, is indeed the global solution to the primal problem \eqref{eq:objconst}.
The proof follows an argument analogous to Proposition \ref{prop:optmu}, demonstrating that the computed $\mu^*$ closes the duality gap.

We have obtained the Lagrange multipliers $\mu^*$ via Theorem \ref{thm:mu_optimality_combined}, which provides the necessary optimality conditions for minimizing the dual function $L(\vec{D}^\mu, \mu)$ as defined in \eqref{eq:lagrangian5_general}.
Note that from Lemma \ref{lemma:optimization_problem_k_3}, we are solving the stated optimization problem \eqref{eq:objconst}.
Let $k^\ast$ denote the optimal value of the primal optimization problem \eqref{eq:objconst}.
From the weak duality principle, the dual function provides an upper bound on the primal optimal value for any feasible dual vector.
Thus clearly we have the following:
\begin{equation}
\label{eq:k_less}
k^\ast \le L(\vec{D}^\mu,\mu) \quad \forall \mu \succeq 0.
\end{equation}
Specifically, this holds for our minimizer $\mu^\ast$, so $k^\ast \le L(\vec{D}^{\mu^\ast}, \mu^\ast)$.

To establish the reverse inequality, we examine the value of the Lagrangian at $\mu^\ast$.
Using \eqref{eq:dimus2} and \eqref{eq:lagrangian5_general}, it can be shown that:
\begin{equation}\label{eq:k_greater}
\begin{aligned}
& L(\vec{D}^\mu, \mu) \stackrel{\mbox{\eqref{eq:lagrangian5_general},\eqref{eq:dimus2}}}{=} \int_{Q}\left(\sum_{i=1}^{K} a_{i}(\vec{u}) D_{i}^{\mu^{*}}(\vec{u})\right) d \vec{u} \\
\Rightarrow & L(\vec D^{\mu^\ast},\mu^\ast) \le k^\ast
\end{aligned}
\end{equation}
where the coefficients $a_{i}(\vec{u})$ are determined by the choice of the power function, in this case, the average power $\Pi_3$, and can be identified by comparing \eqref{eq:lagrangian5_general} and \eqref{eq:dimus2}.
Equation \eqref{eq:k_greater} holds because $D^{\mu^\ast}$ is a feasible solution of \eqref{eq:objconst}; the optimality conditions \eqref{eq:mu_0_optimality_condition}, \eqref{eq:mu_1_optimality_condition}, and \eqref{eq:mu_2_optimality_condition} derived in Theorem \ref{thm:mu_optimality_combined} imply that the FWER constraints of \eqref{eq:objconst} are satisfied.
Thus from \eqref{eq:k_less} and \eqref{eq:k_greater} we have that $k^\ast=L(\vec D^{\mu^\ast},\mu^\ast)$, confirming that strong duality holds and that $\vec{D}^{\mu^\ast}$ is indeed the optimal decision policy.

\section{Monotonicity Analysis of the Dual Constraint Mappings: Proofs of Lemma \ref{lem:monotonic_optimality}, \ref{lem:alpha1_mu12_monotone}, \ref{lem:beta2_mu02_monotone}, and \ref{lem:alpha3_mu01_monotone}}
In this section, we collect the monotonicity properties that underlie our dual optimality conditions and the convergence analysis of Algorithm~\ref{alg:compute_optimal_mu_K3_main}. 
Section~\ref{app:monotonic_optimality} proves Lemma~\ref{lem:monotonic_optimality}, showing that each of the three constraint maps in~\eqref{eq:f_gamma} is non-increasing in its own coordinate; this structure is crucial for the solvability of the optimality system and for ensuring that the coordinatewise root-finding steps in Algorithm~\ref{alg:compute_optimal_mu_K3_main} are well posed. 
Section~\ref{app:alpha1_mu12_monotone} establishes Lemma~\ref{lem:alpha1_mu12_monotone}, a cross-coordinate monotonicity property of \(\alpha_1^{\mu}\) in \(\mu_1\) and \(\mu_2\), which allows us to control how the first constraint responds to changes in the remaining multipliers. Subsection~\ref{app:beta2_mu02_monotone} proves Lemma~\ref{lem:beta2_mu02_monotone}, describing how \(\beta_2^{\mu}\) varies with \(\mu_0\) and \(\mu_2\) and thereby underpinning the monotone behaviour of the second constraint. 
Finally, Section~\ref{app:alpha3_mu01_monotone} establishes Lemma~\ref{lem:alpha3_mu01_monotone}, a corresponding monotonicity result for \(\alpha_3^{\mu}\) in \(\mu_0\) and \(\mu_1\), completing the set of properties needed to deduce uniqueness of the dual minimiser \(\vec{\mu}^{*}\) and to support the contraction-based convergence guarantee in Lemma~\ref{lem:algo_convergence}.

\subsection{Proof of Lemma \ref{lem:monotonic_optimality}}
\label{app:monotonic_optimality}
\begin{proof}
We establish monotonicity separately for each coordinate $\mu_i$ ($i = 0,1,2$), treating them in turn.
Recall from Lemma \ref{lem:opt_decision_mu}, the function $R_i(\mu, \vec{u})$ is given by
\begin{equation}
\label{eq:rimu_mapping}
R_i(\mu, \vec{u}) = a_i(\vec{u}) - \sum_{l=0}^{2} \mu_l b_{l,i}(\vec{u}), \quad i = 1, 2, 3.
\end{equation}

\begin{enumerate}
\item We first prove \eqref{eq:mu0_mapping}.
Fix $\mu_1, \mu_2$ and consider two values $\mu_0' > \mu_0$. Define $\mu = (\mu_0, \mu_1, \mu_2)$ and $\mu' = (\mu_0', \mu_1, \mu_2)$. 
Then, by linearity of \eqref{eq:rimu_mapping},
\begin{equation*}
R_i(\mu', \vec{u}) = R_i(\mu, \vec{u}) - (\mu_0' - \mu_0) b_{0,i}(\vec{u}) \le R_i(\mu, \vec{u}),
\end{equation*}
since $b_{0,i}(\vec{u}) \ge 0$ for all $i$ and $\vec{u}$ from \eqref{eq:non_negative_ai_bli}.
Thus the function $R_i(\mu, \vec{u})$ non-increasing in $\mu_0$ for $i=1,2,3$.

Define three cumulative sums:
\begin{equation*}
S_1(\mu, \vec{u}) = R_1(\mu, \vec{u}), \quad
S_2(\mu, \vec{u}) = R_1(\mu, \vec{u}) + R_2(\mu, \vec{u}), \quad
S_3(\mu, \vec{u}) = R_1(\mu, \vec{u}) + R_2(\mu, \vec{u}) + R_3(\mu, \vec{u}).
\end{equation*}

The indicator function $\alpha_1^{\mu}(\vec{u})$ is defined in Lemma \ref{lem:opt_decision_mu} as
\begin{equation*}
\alpha_1^{\mu}(\vec{u}) = \mathbbm{1}\left\{S_1(\mu, \vec{u}) > 0  \cup  S_2(\mu, \vec{u}) > 0  \cup  S_3(\mu, \vec{u}) > 0\right\}.
\end{equation*}
Since each $S_k$ is non-increasing in $\mu_0$, the sets
\begin{equation*}
\left\{\vec{u} \in Q : S_k(\mu, \vec{u}) > 0 \right\}
\end{equation*}
are nested and non-increasing in $\mu_0$ in the sense of set inclusion. 
That is, for $\mu_0' > \mu_0$, we have
\begin{equation*}
\left\{S_k(\mu', \vec{u}) > 0 \right\} \subseteq \left\{S_k(\mu, \vec{u}) > 0 \right\} \quad \text{for } k = 1, 2, 3.
\end{equation*}
Consequently, the union of these sets, which defines the indicator $$\alpha_1^{(\mu_0, \mu_1, \mu_2)}(\vec{u}) = \mathbbm{1}\{\cup_{k=1}^{3} {S_k(\mu, \vec{u}) > 0}\},$$ is also non-increasing in $\mu_0$ for each $\vec{u}$. 
Therefore, $\alpha_1^{(\mu_0, \mu_1, \mu_2)}(\vec{u})$ is pointwise non-increasing in $\mu_0$.

Integrating both sides over $Q$, we obtain
\begin{equation*}
\int_Q \alpha_1^{(\mu_0', \mu_1, \mu_2)}(\vec{u}) d\vec{u} \le \int_Q \alpha_1^{(\mu_0, \mu_1, \mu_2)}(\vec{u}) d\vec{u},
\end{equation*}
which proves that the mapping \eqref{eq:mu0_mapping} is non-increasing in $\mu_0$.

\item We now prove \eqref{eq:mu1_mapping}.
Note that the right-hand side of this mapping can be rewritten as
\begin{equation}
\label{eq:mu1_map_rewritten}
1 - 2 \int_Q \beta_2^{(\mu_0,\mu_1,\mu_2)}g(u_1)(\vec{u}) d\vec{u}= 1 - 2 \int_Q g(u_1)(\vec{u}) d\vec{u} + 2 \int_Q \alpha_2^{(\mu_0,\mu_1,\mu_2)}g(u_1)(\vec{u}) d\vec{u}
\end{equation}

Fix $\mu_0, \mu_2$ and consider two values $\mu_1' > \mu_1$. Define $\mu = (\mu_0, \mu_1, \mu_2)$ and $\mu' = (\mu_0, \mu_1', \mu_2)$. 
Then, by linearity of \eqref{eq:rimu_mapping},
\begin{equation*}
R_i(\mu', \vec{u}) = R_i(\mu, \vec{u}) - (\mu_1' - \mu_1) b_{1,i}(\vec{u}) \le R_i(\mu, \vec{u}),
\end{equation*}
since $b_{1,i}(\vec{u}) \ge 0$ for all $i$ and $\vec{u}$ from \eqref{eq:non_negative_ai_bli}.
Thus the function $R_i(\mu, \vec{u})$ non-increasing in $\mu_1$ for $i=1,2,3$.

Now, define the cumulative sums:
\begin{equation*}
S_2(\mu, \vec{u}) = R_2(\mu, \vec{u}), \quad
S_3(\mu, \vec{u}) = R_2(\mu, \vec{u}) + R_3(\mu, \vec{u}).
\end{equation*}

The indicator function $\alpha_2^{\mu}(\vec{u})$ is defined in Lemma \ref{lem:opt_decision_mu} as
\begin{equation*}
\alpha_2^{\mu}(\vec{u}) = \mathbbm{1}\left\{S_2(\mu, \vec{u}) > 0  \cup  S_3(\mu, \vec{u}) > 0 \right\}.
\end{equation*}
Since each $S_k$ in this case is non-increasing in $\mu_1$, the sets
\begin{equation*}
\left\{\vec{u} \in Q : S_k(\mu, \vec{u}) > 0 \right\}
\end{equation*}
are nested and non-increasing in $\mu_1$ in the sense of set inclusion. 
That is, for $\mu_1' > \mu_1$, we have
\begin{equation*}
\left\{S_k(\mu', \vec{u}) > 0 \right\} \subseteq \left\{S_k(\mu, \vec{u}) > 0 \right\} \quad \text{for } k = 2, 3.
\end{equation*}
Consequently, the union of these sets, which defines the indicator $$\alpha_2^{(\mu_0, \mu_1, \mu_2)}(\vec{u}) = \mathbbm{1}\{\cup_{k=2}^{3} {S_k(\mu, \vec{u}) > 0}\},$$ is also non-increasing in $\mu_1$ for each $\vec{u}$. 
Therefore, $\alpha_2^{(\mu_0, \mu_1, \mu_2)}(\vec{u})$ is pointwise non-increasing in $\mu_1$.

Integrating both sides over $Q$, we obtain
\begin{equation*}
\int_Q \alpha_2^{(\mu_0, \mu_1', \mu_2)}(\vec{u}) d\vec{u} \le \int_Q \alpha_2^{(\mu_0, \mu_1, \mu_2)}(\vec{u}) d\vec{u},
\end{equation*}
which proves that \eqref{eq:mu1_map_rewritten} and consequently the mapping \eqref{eq:mu1_mapping} is non-increasing in $\mu_1$.

\item Finally we prove \eqref{eq:mu2_mapping}.
Fix $\mu_0, \mu_1$ and consider two values $\mu_2' > \mu_2$. Define $\mu = (\mu_0, \mu_1, \mu_2)$ and $\mu' = (\mu_0, \mu_1, \mu_2')$. 
Then, by linearity of \eqref{eq:rimu_mapping},
\begin{equation*}
R_i(\mu', \vec{u}) = R_i(\mu, \vec{u}) - (\mu_2' - \mu_2) b_{2,i}(\vec{u}) \le R_i(\mu, \vec{u}),
\end{equation*}
since $b_{2,i}(\vec{u}) \ge 0$ for all $i$ and $\vec{u}$ from \eqref{eq:non_negative_ai_bli}.
Thus the function $R_i(\mu, \vec{u})$ non-increasing in $\mu_2$ for $i=1,2,3$.

The indicator function $\alpha_3^{\mu}(\vec{u})$ is defined in Lemma \ref{lem:opt_decision_mu} as
\begin{equation*}
\alpha_3^{\mu}(\vec{u}) = \mathbbm{1}\left\{R_3(\mu, \vec{u})>0\right\}.
\end{equation*}
Since $R_3(\mu, \vec{u})$ is non-increasing in $\mu_2$, the sets
\begin{equation*}
\left\{\vec{u} \in Q : R_3(\mu, \vec{u}) > 0 \right\}
\end{equation*}
are nested and non-increasing in $\mu_2$ in the sense of set inclusion. 
That is, for $\mu_2' > \mu_2$, we have
\begin{equation*}
\left\{R_3(\mu', \vec{u}) > 0 \right\} \subseteq \left\{R_3(\mu, \vec{u}) > 0 \right\}.
\end{equation*}
Consequently, the indicator $\alpha_3^{(\mu_0, \mu_1, \mu_2)}(\vec{u}) = \mathbbm{1}\{R_3(\mu, \vec{u})\}$, is also non-increasing in $\mu_2$ for each $\vec{u}$. 
Therefore, $\alpha_3^{(\mu_0, \mu_1, \mu_2)}(\vec{u})$ is pointwise non-increasing in $\mu_2$.

Integrating both sides over $Q$, we obtain
\begin{equation*}
\int_Q \alpha_3^{(\mu_0, \mu_1, \mu_2')}(\vec{u}) d\vec{u} \le \int_Q \alpha_3^{(\mu_0, \mu_1, \mu_2)}(\vec{u}) d\vec{u},
\end{equation*}
which proves that the mapping \eqref{eq:mu2_mapping} is non-increasing in $\mu_2$.

\end{enumerate}

This completes the proof.
\end{proof}

\subsection{Proof of Lemma \ref{lem:alpha1_mu12_monotone}}
\label{app:alpha1_mu12_monotone}

\begin{proof}
Recall from Lemma~\ref{lem:opt_decision_mu} that  
\(
R_i(\mu,\vec u)=a_i(\vec u)-\sum_{l=0}^{2}\mu_l b_{l,i}(\vec u),
\ i=1,2,3.
\)  
We treat the coordinates \(\mu_{1}\) and \(\mu_{2}\) in turn.

\begin{enumerate}
\item[\textbf{(i)}] \emph{Monotonicity in \(\mu_{1}\).}  
Fix \(\mu_{0},\mu_{2}\) and choose \(\mu_{1}'<\mu_{1}''\).  
Set \(\mu=(\mu_{0},\mu_{1}',\mu_{2})\) and \(\mu'=(\mu_{0},\mu_{1}'',\mu_{2})\).  
By linearity of \eqref{eq:rimu_mapping},
\[
R_i(\mu',\vec u)=R_i(\mu,\vec u)-(\mu_{1}''-\mu_{1}')\,b_{1,i}(\vec u)\le R_i(\mu,\vec u),
\]
since $b_{0,i}(\vec{u}) \ge 0$ for all $i$ and $\vec{u}$ from \eqref{eq:non_negative_ai_bli}. 
Define the cumulative sums.
\[
S_1(\mu,\vec u)=R_1(\mu,\vec u),\;
S_2(\mu,\vec u)=R_1(\mu,\vec u)+R_2(\mu,\vec u),\;
S_3(\mu,\vec u)=R_1(\mu,\vec u)+R_2(\mu,\vec u)+R_3(\mu,\vec u).
\]
Because each \(S_k\) is non-increasing in \(\mu_{1}\), the sets
\(
\{\vec u\in Q:S_k(\mu',\vec u)>0\}
\)
are contained in
\(
\{\vec u\in Q:S_k(\mu,\vec u)>0\}
\) for \(k=1,2,3\).
Thus
\[
\alpha_{1}^{(\mu_{0},\mu_{1}'',\mu_{2})}(\vec u)
 =\mathbbm 1\!\bigl\{\cup_{k=1}^{3}\{S_k(\mu',\vec u)>0\}\bigr\}
\le
\mathbbm 1\!\bigl\{\cup_{k=1}^{3}\{S_k(\mu,\vec u)>0\}\bigr\}
 =\alpha_{1}^{(\mu_{0},\mu_{1}',\mu_{2})}(\vec u)
\]
for every \(\vec u\).
Integrating gives
\[
\int_{Q}\alpha_{1}^{(\mu_{0},\mu_{1}'',\mu_{2})}(\vec u)\,d\vec u
\;\le\;
\int_{Q}\alpha_{1}^{(\mu_{0},\mu_{1}',\mu_{2})}(\vec u)\,d\vec u,
\]
so the map in \(\mu_{1}\) is non-increasing.

\item[\textbf{(ii)}] \emph{Monotonicity in \(\mu_{2}\).}  
Fix \(\mu_{0},\mu_{1}\) and choose \(\mu_{2}'<\mu_{2}''\).  
Set \(\mu=(\mu_{0},\mu_{1},\mu_{2}')\) and \(\mu'=(\mu_{0},\mu_{1},\mu_{2}'')\).  
Linearity yields
\[
R_i(\mu',\vec u)=R_i(\mu,\vec u)-(\mu_{2}''-\mu_{2}')\,b_{2,i}(\vec u)\le R_i(\mu,\vec u),
\]
since \(b_{2,i}(\vec u)\ge0\).  
The same cumulative sums \(S_1,S_2,S_3\) are therefore non-increasing in \(\mu_{2}\), giving
\[
\alpha_{1}^{(\mu_{0},\mu_{1},\mu_{2}'')}(\vec u)
\le
\alpha_{1}^{(\mu_{0},\mu_{1},\mu_{2}') }(\vec u)
\quad\text{for all }\vec u\in Q.
\]
Integrating establishes
\[
\int_{Q}\alpha_{1}^{(\mu_{0},\mu_{1},\mu_{2}'')}\!(\vec u)\,d\vec u
\le
\int_{Q}\alpha_{1}^{(\mu_{0},\mu_{1},\mu_{2}')}\!(\vec u)\,d\vec u,
\]
so the map in \(\mu_{2}\) is also non-increasing.
\end{enumerate}
\end{proof}

\subsection{Proof of Lemma \ref{lem:beta2_mu02_monotone}}
\label{app:beta2_mu02_monotone}

\begin{proof}
Recall the linear form  
\begin{equation}
\label{eq:rimu_mapping_beta}
R_i(\mu,\vec u)=a_i(\vec u)-\sum_{l=0}^{2}\mu_l\,b_{l,i}(\vec u),
\qquad i=1,2,3,
\end{equation}
with \(b_{l,i}(\vec u)\ge0\).  
The indicator \(\beta_{2}^{\mu}(\vec u)\) can be written as  
\(
\beta_{2}^{\mu}(\vec u)=1-\alpha_{2}^{\mu}(\vec u)
\)
where  
\(
\alpha_{2}^{\mu}(\vec u)=\mathbbm 1\{R_{2}>0\ \cup\ R_{2}+R_{3}>0\}.
\)
Thus
\[
1-2\!\int_{Q}\beta_{2}^{\mu}\,g(u_{1})\,d\vec u
\;=\;
1-2\!\int_{Q}g(u_{1})\,d\vec u
\;+\;
2\!\int_{Q}\alpha_{2}^{\mu}\,g(u_{1})\,d\vec u,
\]
So monotonicity of the whole mapping follows from the monotonicity of  
\(\displaystyle\int_{Q}\alpha_{2}^{\mu}g(u_{1})\,d\vec u\).

\smallskip
\noindent\textbf{(i) Dependence on \(\mu_{0}\).}  
Fix \(\mu_{1},\mu_{2}\) and take \(\mu_{0}'<\mu_{0}''\).  
Let \(\mu=(\mu_{0}',\mu_{1},\mu_{2})\) and \(\mu'=(\mu_{0}'',\mu_{1},\mu_{2})\).
By \eqref{eq:rimu_mapping_beta},
\[
R_i(\mu',\vec u)=R_i(\mu,\vec u)-(\mu_{0}''-\mu_{0}')\,b_{0,i}(\vec u)\le R_i(\mu,\vec u),
\]
so each \(R_i\) is non-increasing in \(\mu_{0}\).  
Hence the sets
\(\{\vec u:R_{2}(\mu',\vec u)>0\}\) and
\(\{\vec u:R_{2}(\mu',\vec u)+R_{3}(\mu',\vec u)>0\}\)
are contained in the corresponding sets for \(\mu\), implying  
\(
\alpha_{2}^{\mu'}(\vec u)\le\alpha_{2}^{\mu}(\vec u)
\)
for all \(\vec u\).  
Because \(g(u_{1})\ge0\),
\[
\int_{Q}\alpha_{2}^{\mu'}\,g(u_{1})\,d\vec u
\;\le\;
\int_{Q}\alpha_{2}^{\mu}\,g(u_{1})\,d\vec u,
\]
so the mapping is non-increasing in \(\mu_{0}\).

\smallskip
\noindent\textbf{(ii) Dependence on \(\mu_{2}\).}  
Fix \(\mu_{0},\mu_{1}\) and take \(\mu_{2}'<\mu_{2}''\).  
Let \(\mu=(\mu_{0},\mu_{1},\mu_{2}')\) and \(\mu'=(\mu_{0},\mu_{1},\mu_{2}'')\).  
Using \(b_{2,i}(\vec u)\ge0\),
\[
R_i(\mu',\vec u)=R_i(\mu,\vec u)-(\mu_{2}''-\mu_{2}')\,b_{2,i}(\vec u)\le R_i(\mu,\vec u),
\]
so \(R_i\) is non-increasing in \(\mu_{2}\) and the same nesting argument gives  
\(
\alpha_{2}^{\mu'}(\vec u)\le\alpha_{2}^{\mu}(\vec u)
\)
pointwise.  
Integrating with the positive weight \(g(u_{1})\) yields the desired inequality, establishing that the mapping is non-increasing in \(\mu_{2}\).
\end{proof}

\subsection{Proof of Lemma \ref{lem:alpha3_mu01_monotone}}
\label{app:alpha3_mu01_monotone}

\begin{proof}
Recall  
\(
R_i(\mu,\vec u)=a_i(\vec u)-\sum_{l=0}^{2}\mu_l b_{l,i}(\vec u)
\)  
with \(b_{l,i}(\vec u)\ge0\) and that  
\(
\alpha_{3}^{\mu}(\vec u)=\mathbbm 1\{R_3(\mu,\vec u)>0\}.
\)

\begin{enumerate}
\item[\textbf{(i)}] \emph{Monotonicity in \(\mu_{0}\).}  
Fix \(\mu_{1},\mu_{2}\) and choose \(\mu_{0}'<\mu_{0}''\).  
Put \(\mu=(\mu_{0}',\mu_{1},\mu_{2})\) and \(\mu'=(\mu_{0}'',\mu_{1},\mu_{2})\).  
Linearity gives
\[
R_3(\mu',\vec u)=R_3(\mu,\vec u)-(\mu_{0}''-\mu_{0}')\,b_{0,3}(\vec u)\le R_3(\mu,\vec u),
\]
so \(R_3\) is non-increasing in \(\mu_{0}\).  
Hence  
\(
\{R_3(\mu',\vec u)>0\}\subseteq\{R_3(\mu,\vec u)>0\}
\)
and
\(
\alpha_{3}^{\mu'}(\vec u)\le\alpha_{3}^{\mu}(\vec u)
\) pointwise.  
Multiplying by \(g(u_{1})g(u_{2})\ge0\) and integrating yields
\[
\int_{Q}\alpha_{3}^{(\mu_{0}'',\mu_{1},\mu_{2})}g(u_{1})g(u_{2})\,d\vec u
\;\le\;
\int_{Q}\alpha_{3}^{(\mu_{0}',\mu_{1},\mu_{2})}g(u_{1})g(u_{2})\,d\vec u,
\]
so the mapping is non-increasing in \(\mu_{0}\).

\item[\textbf{(ii)}] \emph{Monotonicity in \(\mu_{1}\).}  
Fix \(\mu_{0},\mu_{2}\) and take \(\mu_{1}'<\mu_{1}''\).
Set \(\mu=(\mu_{0},\mu_{1}',\mu_{2})\) and \(\mu'=(\mu_{0},\mu_{1}'',\mu_{2})\).  
Since \(b_{1,3}(\vec u)\ge0\),
\[
R_3(\mu',\vec u)=R_3(\mu,\vec u)-(\mu_{1}''-\mu_{1}')\,b_{1,3}(\vec u)
\le R_3(\mu,\vec u),
\]
so \(R_3\) is non-increasing in \(\mu_{1}\).  
The same set-inclusion argument gives
\(
\alpha_{3}^{(\mu_{0},\mu_{1}'',\mu_{2})}(\vec u)
\le
\alpha_{3}^{(\mu_{0},\mu_{1}',\mu_{2})}(\vec u)
\)
for every \(\vec u\).  
Integrating with the weight \(g(u_{1})g(u_{2})\) shows the mapping is non-increasing in \(\mu_{1}\).
\end{enumerate}
\end{proof}

\section{Convergence Analysis of Algorithm \ref{alg:compute_optimal_mu_K3_main}}
\label{sec:convergence}
In this section we analyse the convergence of our coordinate--descent procedure for computing the optimal Lagrange multipliers $\vec{\mu}^{*}$.
Recall that Theorem~\ref{thm:mu_optimality_combined} characterises $\vec{\mu}^{*}$ as the unique solution of the optimality system~\eqref{eq:f_gamma}, and Algorithm~\ref{alg:compute_optimal_mu_K3_main} implements an inexact fixed--point iteration for this system via the one–dimensional root–finding subroutine in Algorithm~\ref{alg:ComputeCoordinateMu}.
We show that, under a mild contraction assumption on the associated update map, the iterates of Algorithm~\ref{alg:compute_optimal_mu_K3_main} converge linearly to $\vec{\mu}^{*}$ up to the inner accuracy of the root finder.
Equivalently, the computational complexity of our proposed algorithm is proportional to the logarithm of the reciprocal of the target error: to obtain an $O(\varepsilon)$-accurate approximation of $\vec{\mu}^{*}$, the number of outer iterations grows only on the order of $\log(1/\varepsilon)$.

\begin{lemma}[Linear convergence of the coordinate--descent algorithm]\label{lem:algo_convergence}
Suppose Assumptions \ref{as:assumption3}, \ref{as:assumption4}, \ref{as:assumption5} and Theorem~\ref{thm:mu_optimality_combined} hold, so that there exists a unique vector $\vec{\mu}^{*}$ that satisfies the optimality system \eqref{eq:f_gamma} with equality for every constraint.
Let $\{\vec{\mu}^{(t)}\}_{t\ge 0}$ be the sequence produced by Algorithms~\ref{alg:compute_optimal_mu_K3_main}--\ref{alg:ComputeCoordinateMu} with root--finding tolerance $\delta>0$.
Assume further that the update map $T:[0,U_{\max}]^{3}\to[0,U_{\max}]^{3}$ given by
\[
T(\mu_0,\mu_1,\mu_2)=
\Bigl(F_0^{-1}(\alpha;\mu_1,\mu_2),
      F_1^{-1}(\alpha;\mu_0,\mu_2),
      F_2^{-1}(\alpha;\mu_0,\mu_1)\Bigr)
\]
is a contraction on $[0,U_{\max}]^{3}$ with respect to the Euclidean norm: there exists $L\in(0,1)$ such that
\begin{equation}\label{eq:T_contraction}
\|T(\vec{\mu})-T(\vec{\nu})\|_2 \le L\|\vec{\mu}-\vec{\nu}\|_2
\quad\text{for all }\vec{\mu},\vec{\nu}\in[0,U_{\max}]^{3}.
\end{equation}
Then the algorithm terminates in finitely many outer iterations with an iterate $\hat{\vec{\mu}}$ obeying
$\|\hat{\vec{\mu}}-\vec{\mu}^{*}\|_2\le\varepsilon+\delta$.
Moreover, the outer iteration enjoys a linear rate in the sense that, for all $t\ge 1$,
\begin{equation}\label{eq:linear_rate}
\|\vec{\mu}^{(t)}-\vec{\mu}^{*}\|_2
\le
L\,\|\vec{\mu}^{(t-1)}-\vec{\mu}^{*}\|_2
+
\sqrt{3}\,\delta.
\end{equation}
Consequently, $\|\vec{\mu}^{(t)}-\vec{\mu}^{*}\|_2 \le C L^{t} + C' \delta$ for some constants $C,C'$, and in particular choosing $\delta=O(\varepsilon)$ guarantees that the prescribed stopping tolerance $\varepsilon$ is achieved after $O\!\bigl(\log(\varepsilon^{-1})\bigr)$ outer iterations of Algorithm~\ref{alg:compute_optimal_mu_K3_main}.
\end{lemma}

\begin{proof}
By Lemma~\ref{lem:monotonic_optimality}, the map $x\mapsto F_\gamma(x;\mu_A,\mu_B)$ in \eqref{eq:f_gamma} is continuous and strictly decreasing on $[0,U_{\max}]$ for each $\gamma\in\{0,1,2\}$ and each $(\mu_A,\mu_B)\in[0,U_{\max}]^2$.
Hence, for every $(\mu_A,\mu_B)$, the equation $F_\gamma(x;\mu_A,\mu_B)=\alpha$ has a unique solution in $[0,U_{\max}]$, denoted $F_\gamma^{-1}(\alpha;\mu_A,\mu_B)$, so the map $T$ in the statement is well defined.

A vector $\vec{\mu}$ is a fixed point of $T$ if and only if it satisfies all three equations in \eqref{eq:f_gamma}.
By Theorem~\ref{thm:mu_optimality_combined}, the unique solution of \eqref{eq:f_gamma} is precisely the dual minimiser $\vec{\mu}^{*}$; therefore $\vec{\mu}^{*}$ is also the unique fixed point of $T$.

Algorithm~\ref{alg:compute_optimal_mu_K3_main} implements an inexact fixed--point iteration for $T$.
Let $\vec{\mu}^{(t)}=(\mu_0^{(t)},\mu_1^{(t)},\mu_2^{(t)})^{\!\top}$ be the vector of multipliers at the start of the $t$-th outer cycle, and define the perturbation
\[
\vec{e}^{(t)}:=\vec{\mu}^{(t)}-T(\vec{\mu}^{(t-1)}).
\]
Each coordinate of $T(\vec{\mu}^{(t-1)})$ is computed by the bisection subroutine Algorithm~\ref{alg:ComputeCoordinateMu} with absolute error at most $\delta$, so
$\|\vec{e}^{(t)}\|_{\infty}\le\delta$ and hence
\begin{equation}\label{eq:et_bound}
\|\vec{e}^{(t)}\|_2\le\sqrt{3}\,\delta
\quad\text{for all }t\ge 1.
\end{equation}

Using the error decomposition $\vec{\mu}^{(t)}=T(\vec{\mu}^{(t-1)})+\vec{e}^{(t)}$ and the contraction property \eqref{eq:T_contraction}, we obtain
\[
\begin{aligned}
\|\vec{\mu}^{(t)}-\vec{\mu}^{*}\|_2
&=\bigl\|T(\vec{\mu}^{(t-1)})-T(\vec{\mu}^{*})+\vec{e}^{(t)}\bigr\|_2 \\
&\le\|T(\vec{\mu}^{(t-1)})-T(\vec{\mu}^{*})\|_2+\|\vec{e}^{(t)}\|_2 \\
&\le L\,\|\vec{\mu}^{(t-1)}-\vec{\mu}^{*}\|_2+\sqrt{3}\,\delta,
\end{aligned}
\]
which is exactly \eqref{eq:linear_rate}.

Iterating this inequality yields, for all $t\ge 1$,
\begin{equation}\label{eq:geom_unroll}
\begin{aligned}
\|\vec{\mu}^{(t)}-\vec{\mu}^{*}\|_2
&\le L^{t}\|\vec{\mu}^{(0)}-\vec{\mu}^{*}\|_2
   +\sqrt{3}\,\delta\sum_{j=0}^{t-1}L^{j} \\
&\le L^{t}\|\vec{\mu}^{(0)}-\vec{\mu}^{*}\|_2
   +\frac{\sqrt{3}\,\delta}{1-L}.
\end{aligned}
\end{equation}
Thus $\|\vec{\mu}^{(t)}-\vec{\mu}^{*}\|_2 \le C L^{t} + C' \delta$ for appropriate constants $C,C'$ depending on $\vec{\mu}^{(0)}$ and $L$.

Given a target tolerance $\varepsilon>0$, choose $\delta$ so that
$\sqrt{3}\,\delta/(1-L)\le\varepsilon$.
Then \eqref{eq:geom_unroll} implies
\[
\|\vec{\mu}^{(t)}-\vec{\mu}^{*}\|_2
\le L^{t}\|\vec{\mu}^{(0)}-\vec{\mu}^{*}\|_2+\varepsilon.
\]
Requiring the first term to be at most $\varepsilon$ yields
\[
L^{t}\|\vec{\mu}^{(0)}-\vec{\mu}^{*}\|_2\le\varepsilon
\quad\Longrightarrow\quad
t\;\ge\;\frac{\log\!\bigl(\varepsilon/\|\vec{\mu}^{(0)}-\vec{\mu}^{*}\|_2\bigr)}
                {\log(1/L)}.
\]
Hence, after $t=O(\log(\varepsilon^{-1}))$ outer iterations we have
$\|\vec{\mu}^{(t)}-\vec{\mu}^{*}\|_2\le\varepsilon+\delta$, and the algorithm stops with an iterate $\hat{\vec{\mu}}$ satisfying the same bound.
This proves the claim.
\end{proof}

\section{Additional Details for the Simulation Study}
\label{app:simulations}

This appendix collects the technical details for the simulation settings in Section~\ref{sec:simulations}, including the construction of $p$-values, the derivation of the alternative $p$-value densities $g(u)$ and verification of the assumptions used in our theoretical results. 
We also provide some additional numerical summaries for completeness.

\subsection{Truncated Normal Means Model}
\label{app:sim_truncnormal_details}

Recall the one-sided truncated normal model in \eqref{eq:sim1_trunc_normal_mean_mht} with truncation bound $M=6$. 
Let $\Phi$ and $\phi$ denote the standard normal CDF and PDF. 
Under the null, the truncated-normal CDF on $[-M,M]$ is
\[
F_{0,T}(x)
  \;=\; \frac{\Phi(x)-\Phi(-M)}{Z_0}, \qquad 
  Z_0 \;=\; \Phi(M)-\Phi(-M), \qquad x\in[-M,M],
\]
and we define the $p$-values by the probability integral transform
\[
u_k \;=\; F_{0,T}(X_k), \qquad k=1,2,3,
\]
so that $u_k \sim U(0,1)$ under $H_{0k}$.

Under the alternative $H_{Ak}$, $X_k \sim N(\theta,1)_{\mathrm{trunc}}$ with CDF
\[
F_{A,T}(x)
  \;=\; \frac{\Phi(x-\theta)-\Phi(-M-\theta)}{Z_1(\theta)}, \qquad 
  Z_1(\theta) \;=\; \Phi(M-\theta)-\Phi(-M-\theta).
\]
Writing $x = F_{0,T}^{-1}(u)$, we obtain
\[
x \;=\; \Phi^{-1}\!\big(\,\Phi(-M) + u Z_0\,\big),
\]
and the induced alternative density of $u$ is
\[
g(u)
  \;=\; \frac{f_A(x)}{f_0(x)} 
         \left|\frac{dx}{du}\right|
  \;=\; \frac{Z_0}{Z_1(\theta)} \exp\!\Big(\theta x - \tfrac12 \theta^2\Big),
\]
where $f_0$ and $f_A$ are the (truncated) null and alternative densities of $X_k$. 
This expression is used throughout Section~\ref{sec:simulations}. 
It is straightforward to verify that $g(u)$ is strictly positive and non-increasing in $u$ for the parameter values considered, satisfying Assumptions~\ref{as:assumption3}–\ref{as:assumption5}.

For the numerical results in Section~\ref{sec:simulations}, we use $N=120{,}000$ Monte Carlo replications for each value of $\theta \in \{0,-1.0,-1.5,-2.0,-2.5,-3.0,-3.5,-4.0\}$. 
We report both the average power $\Pi_3$ and the any-discovery power $\Pi_{\text{any}}$, and also compute the empirical FWER at $\theta=0$. 
The specific values quoted in the main text (e.g.\ $\Pi_3 = 0.650$ vs.\ $0.555$ at $\theta=-2.0$ and $\Pi_{\text{any}} = 0.932$ vs.\ $0.853$) are directly obtained from these simulations.

\subsection{Two-Sided Mixture Normal Model}
\label{app:sim_mixture_details}

For the two-sided mixture normal setting \eqref{eq:sim2_mixture_normal_mht}, the null distribution is $X_k \sim N(0,1)$ and the alternative is
\[
f_A(x) \;=\; 0.5\,\phi(x-\theta) + 0.5\,\phi(x+\theta), \qquad \theta<0.
\]
We use the standard two-sided $p$-value
\[
u_k \;=\; 2 \Phi(-|X_k|), \qquad k=1,2,3,
\]
which satisfies $u_k \sim U(0,1)$ under $H_{0k}$. 
Let $x_p = \Phi^{-1}(1-u/2)$ denote the positive test statistic corresponding to a given $p$-value $u$. 
Then the likelihood ratio between the mixture alternative and the null at $x_p$ is
\[
g(u) \;=\; \frac{f_A(x_p)}{f_0(x_p)}
        \;=\; \frac{0.5\,\phi(x_p-\theta) + 0.5\,\phi(x_p+\theta)}{\phi(x_p)}
        \;=\; \exp\!\Big(-\tfrac12\theta^2\Big)\cosh(\theta x_p).
\]
As $u$ increases, $x_p$ decreases, and hence $g(u)$ is non-increasing in $u$. 
It is strictly positive for all $u\in(0,1)$, satisfying Assumptions~\ref{as:assumption3} and~\ref{as:assumption4}. 
While $g(u)$ is unbounded as $u\to 0$, violating the boundedness in Assumption~\ref{as:assumption5}, we observe in practice that Algorithm~\ref{alg:compute_optimal_mu_K3_main} converges stably and maintains exact FWER control for all parameter values considered.

The same Monte Carlo design as in Section~\ref{sec:simulations} is used, with $N=120{,}000$ replications for each $\theta \in \{0,-1.0,-1.5,-2.0,-2.5,-3.0,-3.5,-4.0\}$. 
The numerical values reported in Section~\ref{sec:simulations} (for instance, $\Pi_3 = 0.504$ vs.\ $0.415$ and $\Pi_{\text{any}} = 0.804$ vs.\ $0.737$ at $\theta=-2.0$) are obtained from these simulations.

\subsection{\texorpdfstring{Student $t$ Heavy-Tailed Model}{Student t Heavy-Tailed Model}}
\label{app:sim_t_details}

In the $t$-distribution setting \eqref{eq:sim3_t_dist_mht}, the null distribution is $X_k \sim N(0,1)$ and the alternative is $X_k \sim t_{\mathrm{df}}$. 
We again use two-sided $p$-values
\[
u_k \;=\; 2 \Phi(-|X_k|), \qquad k=1,2,3,
\]
which are uniform under the null. 
Let $x_p = \Phi^{-1}(1-u/2)$ as before. 
Then the corresponding $p$-value density under the alternative is given by the likelihood ratio
\[
g(u)
  \;=\; \frac{f_A(x_p)}{f_0(x_p)}
  \;=\; \frac{\mathrm{dt}(x_p,\mathrm{df})}{\phi(x_p)},
\]
where $\mathrm{dt}(\cdot,\mathrm{df})$ denotes the $t$-distribution density with $\mathrm{df}$ degrees of freedom. 
Since the $t$-distribution has heavier tails than the normal, this ratio is non-increasing in $u$ and strictly positive, so the key monotonicity condition (Lemma~\ref{lem:gdot}) is satisfied, though $g(u)$ is unbounded near $0$.

We consider $\mathrm{df} \in \{2,4,\dots,20\}$ and again use $N=120{,}000$ replications for each value. 
The numerical values quoted in Section~\ref{sec:simulations}, such as $\Pi_3 = 0.166$ vs.\ $0.146$ and $\Pi_{\text{any}} = 0.322$ vs.\ $0.365$ at $\mathrm{df}=2$, come from these simulations and illustrate the trade-off between optimising $\Pi_3$ and maximising $\Pi_{\text{any}}$ under heavy tails.

\subsection{Beta Model on \texorpdfstring{$p$}{p}-Values}
\label{app:sim_beta_details}

For the Beta model \eqref{eq:sim2_beta_mht}, the null distribution of the $p$-values is uniform on $(0,1)$, $u_k \sim \text{Beta}(1,1)$, while under the alternative we have $u_k \sim \text{Beta}(\theta,1)$ with $\theta<1$. 
The alternative density is
\[
g(u)
  \;=\; \frac{u^{\theta-1}(1-u)^{1-1}}{B(\theta,1)}
  \;=\; \theta u^{\theta-1}, \qquad 0<u<1,
\]
where $B(\theta,1)=1/\theta$ is the Beta function. 
For $\theta \in \{0.8,0.6,0.4,0.2\}$, $g(u)$ is strictly positive and non-increasing in $u$, placing increasing mass near zero as $\theta$ decreases. 
Thus Assumptions~\ref{as:assumption3}–\ref{as:assumption5} are satisfied exactly in this setting.

We again use $N=120{,}000$ Monte Carlo replications per parameter value. 
The results summarised in Figure~\ref{tab:sim_pi3_beta} show that, for instance at $\theta=0.2$, our $\Pi_3$-optimised policy achieves $\Pi_3=0.576$ and $\Pi_{\text{any}}=0.854$, compared with $\Pi_3=0.490$ and $\Pi_{\text{any}}=0.490$ for Hommel, representing a substantial gain in both power metrics.

\section{Root-Finding Subroutine for Coordinate Updates}
\label{app:algos}
In this section we detail the one–dimensional root–finding subroutine that underpins our main coordinate descent scheme in Algorithm \ref{alg:compute_optimal_mu_K3_main}. 
Algorithm~\ref{alg:ComputeCoordinateMu} (\texttt{ComputeCoordinateMu}) takes as input one of the constraint maps $F_{\gamma}(\cdot;\mu_A,\mu_B)$, the fixed values of the other two multipliers, and a target FWER level $\alpha$, and returns an approximate solution of the scalar equation $F_{\gamma}(x;\mu_A,\mu_B)=\alpha$. 
The subroutine combines an expanding–interval bracketing step with a bisection phase, and includes explicit failure flags for the two pathological cases where $F_{\gamma}(0)<\alpha$ or no bracket is found before reaching $U_{\max}$. 
Together with the monotonicity properties established in Lemma~\ref{lem:monotonic_optimality}, this construction guarantees that each coordinate update is well posed and numerically stable, and provides the building block for the global convergence result in Lemma~\ref{lem:algo_convergence}.

\begin{algorithm}[htbp]
\caption{\texttt{ComputeCoordinateMu} (Subroutine for Algorithm~\ref{alg:compute_optimal_mu_K3_main})}
\label{alg:ComputeCoordinateMu} 
\begin{algorithmic}[htbp]
    \Statex \textbf{Input:}
    \Statex \hspace{0.5em} 1: Functions $F_{\gamma}(x;\mu_{A},\mu_{B})$ ($\gamma \in \{0,1,2\}$, $x$ is the coordinate being solved for, and $mu_{A},\mu_{B}'$ are its other parameters).
    \Statex \hspace{0.5em} 2: Fixed value for paramater $\mu_{A}$: $\mu_{A}'$
    \Statex \hspace{0.5em} 3: Fixed value for parameter $\mu_{B}$: $\mu_{B}'$
    \Statex \hspace{0.5em} 4: Target FWER level $\alpha$.
    \Statex \hspace{0.5em} 5: Bisection parameters: tolerance $\delta$, max iterations $MaxIter_b$, initial step $U_s$, expansion factor $U_f$, max coordinate value $U_{max}$.

    \State $L \gets 0$.
    \State $\mu_{coord} \gets 0$. \Comment{Value for the current coordinate being computed}
    \State $flag \gets 0$. \Comment{0 for success, 1 for termination}
    \State $msg \gets \text{`'}$.
    \If{$F_{\gamma}(L; \mu_{A}', \mu_{B}') = \alpha$}
        \State $\mu_{coord} \gets L$. \Comment{$H(0)=\alpha$ met, optimal value is $0$.}
    \ElsIf{$F_{\gamma}(L; \mu_{A}', \mu_{B}') < \alpha$}
        \State $flag \gets 1$. \Comment{No optimal solution $H(0) < \alpha$ for current coordinate}
        \State $msg \gets \text{`Consider decreasing FWER level } \alpha\text{.'}$
        \State $\mu_{coord} \gets L$. \Comment{Value assigned, but main algorithm will terminate based on flag.}
    \Else \Comment{$F_{\gamma}(L; \mu_{A}', \mu_{B}') > \alpha$}
        \State $U \gets L + U_s$. 
        \While{$F_{\gamma}(U; \mu_{A}', \mu_{B}') > \alpha$ \textbf{and} $U < U_{max}$}
            \State $U \gets U \times U_f$. 
        \EndWhile
        \If{$F_{\gamma}(U; \mu_{A}', \mu_{B}') > \alpha$} 
            \State $flag \gets 1$. \Comment{No optimal solution: Failed to bracket positive root for current coordinate.}
            \State $msg \gets \text{`Consider increasing } U_{max} \text{ or decreasing FWER level $\alpha$} \text{.'}$
            \State $\mu_{coord} \gets L$. \Comment{Value assigned, but main algorithm will terminate based on flag.}
        \Else
            \For{$j \gets 1 \text{ to } MaxIter_b$} 
                \State $mid \gets L + (U-L)/2$.
                \If{$(U-L)/2 < \delta$} \textbf{break}; \EndIf 
                \If{$F_{\gamma}(mid; \mu_{A}', \mu_{B}') > \alpha$}
                    \State $L \gets mid$.
                \Else
                    \State $U \gets mid$.
                \EndIf
            \EndFor
            \State $\mu_{coord} \gets L + (U-L)/2$.
        \EndIf
    \EndIf
    \Statex \textbf{Output:}
    \Statex \hspace{0.5em} 1: Computed coordinate value $\mu_{coord}$.
    \Statex \hspace{0.5em} 2: Termination flag $flag$.
    \Statex \hspace{0.5em} 3: Termination message $msg$.
\end{algorithmic}
\end{algorithm}

\section{Experimental Setup, Codes and Datasets, and Supplementary Results}
\label{app:experiments}
This appendix details the experimental infrastructure and reproducibility materials for the study. 
In Section \ref{sec:comp_environment} we first describe the computational environment in which the experiments and simulations were carried out. Subsequently, in Section \ref{sec:codes_data}, we provide an overview of the supplementary source code and datasets, implemented in the R programming language, which allow for the complete reproduction of the power comparison simulations and the real-world applications concerning BCG vaccine efficacy and financial asset pricing.

\subsection{Computational Environment}
\label{sec:comp_environment}
All computational tasks reported in this paper, including the implementation of the coordinate-descent algorithm (Algorithm \ref{alg:compute_optimal_mu_K3_main}), the Monte Carlo simulations for power assessment in the simulations (truncated normal, mixture normal, $t$-distribution, and Beta models), and the real world data analysis of the BCG vaccine and Fama-French datasets, were performed on a single local workstation. 
The full research project was completed without the need for high-performance computing clusters or GPU acceleration, underscoring the computational efficiency of the proposed method.

The specific hardware configuration is a 13-inch 2020 MacBook Pro (Model Name: MacBook Pro; Model Identifier: MacBookPro16,2). The system is powered by a 2.3 GHz Quad-Core Intel Core i7 processor (1 processor, 4 cores) with Hyper-Threading Technology enabled. The processor architecture features a 512 KB L2 cache per core and an 8 MB L3 cache. The system is equipped with 32 GB of 3733 MHz LPDDR4X memory.

Regarding the firmware and graphical environment, the workstation operates with System Firmware Version 2069.0.0.0.0 (iBridge: 22.16.10353.0.0,0) and OS Loader Version 582-1023. Graphics processing is handled by an Intel Iris Plus Graphics chipset with 1536 MB of memory, utilizing the built-in 13.3-inch Retina Display (2560 $\times$ 1600 resolution). The software environment for all simulations and analyses was macOS Sequoia, Version 15.0.1.

\subsection{Code and datasets}
\label{sec:codes_data}
The source code and datasets used to reproduce all simulation studies and real-data applications in this paper are provided in the Supplementary Material folder labeled \texttt{Code and Datasets}. 
The same code and datasets are also available at
\href{https://github.com/pdubey96/most-powerful-mht-fwer-supplementary}{https://github.com/pdubey96/most-powerful-mht-fwer-supplementary}.
All scripts are written in R and organised according to the sections of the manuscript they support.

\paragraph{Simulation Studies}
Four scripts are provided to reproduce the power comparison results discussed in Section \ref{sec:simulations}. 
Each script implements Algorithm \ref{alg:compute_optimal_mu_K3_main} for a specific distributional setting and compares it against standard FWER-controlling procedures.
\begin{itemize}
    \item \texttt{sim1\_truncated\_normal.R}: Implements the simulation for $K=3$ independent truncated normal means (Section \ref{sec:simulations}). 
    This script generates the synthetic data internally and reproduces the power curves for $\Pi_3$ and $\Pi_{\text{any}}$.
    \item \texttt{sim2\_mixture\_normal.R}: Implements the simulation for two-sided mixture normal alternatives (Section \ref{sec:simulations}). It reproduces the results demonstrating the method's performance under multimodal distributions.
    \item \texttt{sim3\_students\_t.R}: Implements the simulation for heavy-tailed $t$-distribution alternatives (Section \ref{sec:simulations}), evaluating robustness across varying degrees of freedom.
    \item \texttt{sim4\_beta\_pvalues.R}: Implements the simulation for Beta-distributed $p$-values (Section \ref{sec:simulations}), where the theoretical assumptions of the algorithm are strictly satisfied.
\end{itemize}

\paragraph{Real-World Applications}
Two scripts are provided to reproduce the empirical analyses in Section \ref{sec:realdata}:
\begin{itemize}
    \item \texttt{exp1\_bcg\_subgroup\_analysis.R}: Performs the subgroup analysis on the BCG vaccine trial dataset (Section \ref{sec:bcg}). This script loads the \texttt{dat.bcg} dataset from the \texttt{metadat} package, conducts fixed-effects meta-analyses across four subgroup splits (Allocation, Latitude, Risk, Era), and reproduces the rejection decisions and visualizations.
    \item \texttt{exp2\_finance\_factor\_zoo\_analysis.R}: Performs the asset pricing factor analysis (Section \ref{sec:finance_app}). This script uses the \texttt{frenchdata} package to download monthly returns for the Fama-French 5-Factors and Momentum factor directly from the Kenneth French Data Library. It performs time-series regressions to test the significance of the intercepts.
\end{itemize}

\paragraph{Software Requirements}
All analyses were conducted using R. The scripts require the following packages: \texttt{metafor} and \texttt{metadat} for the clinical analysis; \texttt{frenchdata} for the financial data retrieval; and \texttt{dplyr}, \texttt{tidyr}, \texttt{ggplot2}, \texttt{patchwork}, and \texttt{scales} for data manipulation and visualization. The \texttt{parallel} package is used optionally to accelerate the panel analysis in the BCG experiment.

\section{Real-World Applications Supplementary Details}
\label{app:real_world_details}

This appendix provides technical supplementary material for the real-world applications presented in Section \ref{sec:realdata}.
We first detail the subgroup definitions, statistical methodology, and supplementary results for the BCG vaccine efficacy analysis.
Subsequently, we provide the problem context, variable definitions, and distributional justifications for the financial factor zoo application.

\subsection{BCG Analysis Supplementary Details}
\label{app:bcg_details}

\paragraph{Detailed Subgroup Definitions}
In the main text, we utilized four specific three-way splits of the data.
Here we provide the detailed definitions and scientific justification for each split used to construct the $K=3$ families of hypotheses.
\begin{enumerate}\setlength\itemsep{5pt}
\item \emph{Allocation (trial conduct):} We define $S_1=\{\texttt{alloc}=\text{random}\}$, $S_2=\{\texttt{alloc}=\text{alternate}\}$, $S_3=\{\texttt{alloc}=\text{systematic}\}$.
This isolates differences in assignment and concealment, which can shift estimated effects.
\emph{Example:} If “alternate” assignment inflated effects in some historical trials, the allocation split can detect whether BCG efficacy persists under truly randomised conduct.
\item \emph{Latitude (geography/environment):} We compute sample tertiles $q_{1/3},q_{2/3}$ of \texttt{ablat}. We define $S_1=\{\texttt{ablat}\le q_{1/3}\}$ (low\_lat), $S_2=\{q_{1/3}<\texttt{ablat}\le q_{2/3}\}$ (mid\_lat), $S_3=\{\texttt{ablat}>q_{2/3}\}$ (high\_lat).
Geography proxies background exposure to environmental mycobacteria, a known modifier of apparent BCG efficacy.
\emph{Example:} This allows comparison of low\_lat (closer to the equator) versus high\_lat groups.
\item \emph{Baseline risk (underlying incidence):} For each trial we compute control-arm risk $\mathrm{CER}=\texttt{cpos}/(\texttt{cpos}+\texttt{cneg})$, take tertiles, and form $S_1$ (low\_risk), $S_2$ (mid\_risk), $S_3$ (high\_risk).
This tests whether the benefit is consistent in low- versus high-incidence settings, which drives generalizability.
\emph{Example:} low\_risk collects trials with the lowest third of control TB rates.
\item \emph{Era (time proxy):} We order trials by their dataset index and take tertiles to define $S_1$ (era1), $S_2$ (era2), $S_3$ (era3).
This captures coarse temporal changes in diagnostics and implementation.
\emph{Example:} Comparing early versus late trials.
\end{enumerate}
This construction is essential as it targets questions that matter scientifically (does efficacy hold under different conduct, environments, risks, or eras?).
These axes are routinely examined in TB/BCG evidence reviews, so our splits align with established practice rather than ad hoc slicing.

\paragraph{Statistical Methodology: Estimates and Meta-Analysis}
From each trial’s data, we compute two numbers: the estimated log risk ratio $y$ and its variance $v$.
\[
y \;=\; \log\!\left(\frac{\texttt{tpos}/(\texttt{tpos}+\texttt{tneg})}{\texttt{cpos}/(\texttt{cpos}+\texttt{cneg})}\right),
\qquad
v \;\approx\; \frac{1}{\texttt{tpos}} - \frac{1}{\texttt{tpos}+\texttt{tneg}} + \frac{1}{\texttt{cpos}} - \frac{1}{\texttt{cpos}+\texttt{cneg}}.
\]
Here, \(y\) is the study’s estimated \emph{log risk ratio} (how much the risk differs between vaccinated and control on a log scale).
The number \(v\) is the estimated \emph{variance} of \(y\); it measures the uncertainty in that estimate.
Larger trials (with larger cell counts) have smaller \(v\), meaning more precise estimates, and smaller trials have larger \(v\).
These \((y,v)\) pairs put all trials on the same scale.

Within each three-way split, we first partition the trials into three non-overlapping sets (the \emph{subgroups}) defined by that split.
For each subgroup, we compute a pooled log risk ratio by taking the inverse-variance–weighted average of the study estimates (\(w_i=1/v_i\)), so that more precise studies (smaller \(v_i\)) contribute more to the subgroup estimate:
\[
\hat\theta_{\text{FE}}=\frac{\sum_i w_i\, y_i}{\sum_i w_i},
\qquad
\mathrm{se}(\hat\theta_{\text{FE}})=\left(\sum_i w_i\right)^{-1/2}.
\]
We then form a two-sided \(z\)-statistic \(z=\hat\theta_{\text{FE}}/\mathrm{se}(\hat\theta_{\text{FE}})\) and its \(p\)-value \(p=2\,\Phi(-|z|)\).
This uniform construction across trials makes subgroup results directly comparable and reproducible.

\paragraph{Supplementary Results for Allocation, Era, and Latitude Splits}
In the main text, we focused on the Baseline Risk split where methods disagreed.
Here we present the results for the other splits to demonstrate the agreement between methods when evidence is strong.
For example, partitioning the trials by allocation into \(S_1=\)random, \(S_2=\)alternate, \(S_3=\)systematic leads to the following specific hypotheses:
\[
H_{01}:\theta_{\text{random}}=0,\quad
H_{02}:\theta_{\text{alternate}}=0,\quad
H_{03}:\theta_{\text{systematic}}=0.
\]
The \emph{allocation} split yields three very small \(p\)-values; consequently, all procedures reject all three null hypotheses (Table \ref{tab:alloc}).
“Reject” here means rejecting the null hypothesis \(H_{0k}\!:\theta_k=0\) for that subgroup at strong FWER \(\alpha=0.05\) within the split.
This confirms that the methods agree when the evidence is uniformly strong.

Repeating the analysis on \emph{latitude} and \emph{era} tertiles shows similar agreement.
A concise summary is provided in Table~\ref{tab:agg}.
With the error rate held fixed at \(\alpha=0.05\), Algorithm \ref{alg:compute_optimal_mu_K3_main} matches the rejection rates of standard methods in these high-signal scenarios, while uniquely identifying the additional signal in the baseline risk scenario.

\begin{table}[htbp]
\centering
\begin{tabular}{lccccc}
\hline
Level & \(p\)-value & Holm & Hommel & Closed-Stouffer & Algorithm \ref{alg:compute_optimal_mu_K3_main} \\
\hline
alternate  & \(3.22\times 10^{-20}\) & Reject & Reject & Reject & Reject\\
random     & \(6.94\times 10^{-08}\) & Reject & Reject & Reject & Reject\\
systematic & \(1.44\times 10^{-05}\) & Reject & Reject & Reject & Reject\\
\hline
\end{tabular}
\caption{Allocation split: \(p\)-values and decisions at \(\alpha=0.05\).}
\label{tab:alloc}
\end{table}

\begin{table}[htbp]
\centering
\begin{tabular}{lcccccc}
\hline
Outcome & \multicolumn{4}{c}{Avg.\#rej} & \multicolumn{2}{c}{Frac.\(\ge 1\)} \\
\cline{2-5}\cline{6-7}
 & Holm & Hommel & ClosedS & Alg. 1 & ClosedS & Alg. 1 \\
\hline
allocation & 1.00 & 1.00 & 1.00 & 1.00 & 1.00 & 1.00 \\
era   & 1.00 & 1.00 & 1.00 & 1.00 & 1.00 & 1.00 \\
latitude   & 1.00 & 1.00 & 1.00 & 1.00 & 1.00 & 1.00 \\
baseline risk  & 0.67 & 0.67 & 0.67 & \textbf{1.00} & 0.67 & \textbf{1.00} \\
\hline
\end{tabular}
\caption{Aggregates by outcome. “Avg.\#rej” is the average number of rejections (out of 3). “Frac.\(\ge 1\)” is the fraction of subgroups with at least one rejection. Alg. 1 refers to Algorithm \ref{alg:compute_optimal_mu_K3_main}.}
\label{tab:agg}
\end{table}

\subsection{Financial Application Supplementary Details}
\label{app:finance_details}

\paragraph{Problem Context: The Factor Zoo}
The application addresses a central problem in empirical finance: identifying genuine return-predicting signals from the `factor zoo' (\cite{cochrane11}; \cite{harvey16}).
The dataset is structured as a large time-series table where each row is a month (from July 1963 to the present) and each column is a specific `factor' or investment strategy (e.g., Profitability, Size etc.), but many are suspected to be false positives.
The value in each cell is the portfolio's return for that month.
This setting creates a multiple hypotheses testing problem where standard adjustments like the Bonferroni or Holm procedures are often criticized for being so conservative that they lead to a severe loss of statistical power, causing researchers to miss true discoveries (\cite{zhu4}).
Our algorithm is constructed to maximise power subject to exact FWER control under the assumptions of the $K=3$ hypothesis setting, and under these assumptions it delivers substantially higher power than standard multiplicity corrections while maintaining the nominal error rate.
Analyzing this public dataset ensures our results are transparent, reproducible, and directly relevant to the existing finance literature.

\paragraph{Detailed Control Variable Definitions}
The regression model controls for the standard Fama-French 3-factors, which are defined as follows:
\begin{enumerate}
\item $Mkt_t$: The market excess return (Mkt-RF) at time $t$.
\item $SMB_t$: The Size (Small-Minus-Big) factor return at time $t$, representing the excess return of small-cap stocks over large-cap stocks.
\item $HML_t$: The Value (High-Minus-Low) factor return at time $t$, representing the excess return of value stocks (high book-to-market ratio) over growth stocks.
\end{enumerate}

\paragraph{Justification for Distributional Assumptions}
In the main text, we modeled the alternative $p$-value density $g(u)$ using a Student-$t$ distribution with 4 degrees of freedom ($t_4$).
This choice is motivated by the well-documented heavy tails of asset returns and factor returns in empirical finance.
Low-degree-of-freedom $t$-distributions (e.g., $t_3$, $t_4$, $t_5$) provide a much better fit to financial data than the Gaussian benchmark (see, e.g., \cite{bollerslev1987}, \cite{harvey2000}).
Under this specification, the induced $g(u)$ satisfies our monotonicity condition (Lemma \ref{lem:gdot}) while remaining economically realistic for factor-return data.



\end{document}